\documentclass[11pt,pdfa,letterpaper]{article}

\newif\ifcomments

\commentstrue

\usepackage[in]{fullpage}

\usepackage{iftex}
\ifPDFTeX
  \usepackage[utf8]{inputenc}
  \usepackage[noTeX]{mmap}
  \usepackage[T1]{fontenc}
\fi
\ifLuaTeX
  \usepackage{luatex85}
  \usepackage[noTeX]{mmap}
\fi

\usepackage{amsfonts}
\usepackage{amsmath}
\usepackage{mathtools,amsthm,amssymb,thmtools} \usepackage{xcolor}
\usepackage{graphicx}
\usepackage{comment}
\usepackage{qtree}
\usepackage{tree-dvips}
\usepackage{float}
\definecolor{linkblue}{HTML}{001487}
\usepackage[colorlinks=true,allcolors=linkblue]{hyperref}
\usepackage[nameinlink,capitalize,noabbrev]{cleveref}
\usepackage{braket}
\usepackage{mathrsfs}
\usepackage{tikz}
\usepackage{qcircuit}
\usepackage{xspace}
\usepackage{dsfont}
\usepackage{enumitem} \usepackage{longfbox}
\usepackage{csquotes}
\usepackage{authblk}
\usepackage{ninecolors}
\usepackage{longfbox}

\usepackage{times}

\usetikzlibrary{fadings}
\usetikzlibrary{patterns}
\usetikzlibrary{shadows.blur}
\usetikzlibrary{shapes}

\usepackage[style=alphabetic,minalphanames=3,maxalphanames=4,maxnames=99,backref=true]{biblatex}

\DeclareFieldFormat{eprint:iacr}{Cryptology ePrint Archive: \href{https://ia.cr/#1}{\texttt{#1}}}
\DeclareFieldFormat{eprint:iacrarchive}{Cryptology ePrint Archive: \href{https://eprint.iacr.org/archive/#1}{\texttt{#1}}}
\AtEveryBibitem{\clearlist{address}
 \clearfield{date}
 \clearfield{isbn}
 \clearfield{issn}
 \clearlist{location}
 \clearfield{month}
 \clearfield{series}

 \ifentrytype{book}{}{\clearlist{publisher}
  \clearname{editor}
 }
}

\usepackage{framed}

\newcommand{\unpleasant}[1]{\ignorespaces}

\theoremstyle{definition}
\newtheorem{theorem}{Theorem}[section]
\newtheorem*{theorem*}{Theorem}
\newtheorem{definition}[theorem]{Definition}
\newtheorem{lemma}[theorem]{Lemma}
\newtheorem{claim}[theorem]{Claim}
\Crefname{claim}{Claim}{Claims}
\newtheorem{protocol}{Protocol}

\newtheorem*{lemma*}{Lemma}

\newtheorem{corollary}[theorem]{Corollary}
\newtheorem*{corollary*}{Corollary}
\newtheorem{proposition}[theorem]{Proposition}
\newtheorem{conjecture}[theorem]{Conjecture}

\theoremstyle{definition}
\newtheorem{openproblem}{Open Problem}

\theoremstyle{remark}
\newtheorem{remark}[theorem]{Remark}
\numberwithin{equation}{section}

\usepackage{stmaryrd} \makeatletter
\renewcommand{\paragraph}{\@startsection{paragraph}{4}{\z@}{2.25ex \@plus 1ex \@minus .2ex}{-1em}{\normalfont\normalsize\bfseries}}
\makeatother
\newcommand{\class}[1]{\mathsf{#1}}
\newcommand{\poly}{\mathrm{poly}}

\newcommand{\negl}{\mathrm{negl}}

\newcommand{\statePSPACE}{\class{statePSPACE}}
\newcommand{\stateBQP}{\class{stateBQP}}

\newcommand{\QIP}{\class{QIP}}
\newcommand{\BQP}{\class{BQP}}

\newcommand{\PSPACE}{\class{PSPACE}}

\newcommand{\stateqip}{\class{stateQIP}}
\newcommand{\stateQIP}{\stateqip}

\newcommand{\unitaryPSPACE}{\class{unitaryPSPACE}}

\newcommand{\avgUnitaryPSPACE}{\class{avgUnitaryPSPACE}}
\newcommand{\avgUnitaryBQP}{\class{avgUnitaryBQP}}

\newcommand{\avgUnitaryQIP}{\class{avgUnitaryQIP}}

\newcommand{\usynth}[1]{\mathscr{#1}}

\newcommand{\unitaryBQP}{\class{unitaryBQP}}

\newcommand{\unitaryQIP}{\class{unitaryQIP}}

\newcommand{\unitary}[1]{\class{unitary#1}}
\newcommand{\avgUnitary}[1]{\class{avgUnitary#1}}

\newcommand{\Uhlmann}{\textsc{Uhlmann}}

\newcommand{\SuccinctUhlmann}{\textsc{SuccinctUhlmann}}
\newcommand{\avgSuccinctUhlmann}{\textsc{DistSuccinctUhlmann}}

\newcommand{\DistUhlmann}{\textsc{DistUhlmann}}
\newcommand{\DistSuccinctUhlmann}{\textsc{DistSuccinctUhlmann}}

\newcommand*{\interact}{\mathord{\leftrightarrows}}
\newcommand{\work}{{\reg{work}}}
\newcommand{\flag}{{\reg{flag}}}
\newcommand{\out}{{\reg{out}}}

\newcommand{\Sim}{\mathrm{Sim}}

\newcommand{\eps}{\epsilon}

\newcommand{\ketbra}[2]{\ket{#1}\!\!\bra{#2}}
\renewcommand{\cal}[1]{\mathcal{#1}}

\newcommand{\C}{\mathbb{C}}
\newcommand{\N}{\mathbb{N}}

\newcommand{\E}{\mathop{\mathbb{E}}}
\newcommand{\Tr}{\mathrm{Tr}}

\newcommand{\reg}[1]{\mathsf{#1}}

\newcommand{\Haar}{\mathrm{Haar}}

\newcommand{\Id}{\id}

\newcommand{\td}{\mathrm{td}}
\newcommand{\zs}{0 \cdots 0}

\newcommand{\fidelity}{\mathrm{F}}

\newcommand{\wt}[1]{\widetilde{#1}}

\newcommand{\setft}[1]{\textnormal{#1}}
\newcommand{\id}{\setft{id}}

\newcommand{\bits}{\ensuremath{\{0, 1\}}}
\newcommand{\linear}{\mathrm{L}}

\usepackage{mleftright}

\newcommand{\mparen}[1]{\mleft(#1\mright)}
\newcommand{\mbracket}[1]{\mleft[#1\mright]}

\newcommand{\ot}{\ensuremath{\otimes}}
\newcommand{\deq}{\coloneqq}

\newcommand{\ptr}[2]{\mbox{\rm Tr}_{#1}\mparen{#2}}
\newcommand{\pr}[1]{{\rm Pr}\mbracket{#1}}

\newcommand{\norm}[1]{\left\lVert#1\right\rVert}
\DeclareMathOperator{\pos}{Pos}
\DeclareMathOperator{\supp}{\setft{supp}}

\DeclareMathOperator{\sgn}{sgn}

\let\1\relax
\newcommand{\1}{\mathds{1}}
\newcommand{\cptp}{\setft{CPTP}}
\newcommand{\states}{\setft{S}}

\newcommand{\proj}[1]{\ket{#1}\!\!\bra{#1}}

\newcommand{\cD}{\ensuremath{\mathcal{D}}}
\newcommand{\cE}{\ensuremath{\mathcal{E}}}

\newcommand{\cN}{\ensuremath{\mathcal{N}}}

\newcommand{\cU}{\ensuremath{\mathcal{U}}}

\date{}
\title{Unitary Complexity and the Uhlmann Transformation Problem}

\author[1]{John Bostanci}
\author[1]{Yuval Efron}
\author[2]{Tony Metger}
\author[3]{\\Alexander Poremba}
\author[4]{Luowen Qian}
\author[1]{Henry Yuen}

\affil[1]{Columbia University}
\affil[2]{ETH Zurich}
\affil[3]{MIT/Boston University}
\affil[4]{Boston University/NTT Research/Northeastern University}

\begin{document}
\maketitle
\pagestyle{empty}
\thispagestyle{empty}

\begin{abstract}
State transformation problems such as compressing quantum information or breaking quantum commitments are fundamental quantum tasks.
However, their computational difficulty cannot easily be characterized using traditional complexity theory, which focuses on tasks with classical inputs and outputs. 

To study the complexity of such state transformation tasks, we introduce a framework for \emph{unitary synthesis problems}, including notions of reductions and unitary complexity classes. 
We use this framework to study the complexity of transforming one entangled state into another via local operations. 
We formalize this as the \emph{Uhlmann Transformation Problem}, an algorithmic version of Uhlmann's theorem. 
Then, we prove structural results relating the complexity of the Uhlmann Transformation Problem, polynomial space quantum computation, and zero knowledge protocols. 

The Uhlmann Transformation Problem allows us to characterize the complexity of a variety of tasks in quantum information processing, including decoding noisy quantum channels, breaking falsifiable quantum cryptographic assumptions, implementing optimal prover strategies in quantum interactive proofs, and decoding the Hawking radiation of black holes. 
Our framework for unitary complexity thus provides new avenues for studying the computational complexity of many natural quantum information processing tasks.
\end{abstract}

\vfill

 \newpage
 \setcounter{tocdepth}{2}
 \tableofcontents

 \newpage

\newcommand{\titleavgUnitaryHVPZK}{\texorpdfstring{$\avgUnitary{HVPZK}$ }{avgUnitaryHVSZK }}

\newcommand{\titleavgUnitaryQIP}{\texorpdfstring{$\avgUnitary{QIP}$ }{avgUnitaryQIP }}

\newcommand{\titleavgUnitaryPSPACE}{\texorpdfstring{$\avgUnitaryPSPACE$ }{avgUnitaryPSPACE }}
 
\pagestyle{plain}
\section{Introduction}
\label{sec:intro}

Complexity theory studies the resources required to solve computational problems.
Quantum complexity has traditionally studied the quantum resources required to solve \emph{classical} computational problems, i.e., problems that have classical inputs and outputs.
However, quantum mechanics also introduces a new kind of computational problem: preparing and transforming quantum states.
The goal of this paper is to initiate the formal complexity-theoretic study of such \emph{quantum state transformation problems}.
To this end, we extend the language of traditional complexity theory to encompass state transformation problems -- we call the resulting framework \emph{unitary complexity theory}.

The idea that the complexity of inherently quantum problems cannot easily be reduced to the complexity of classical problems has already been explored in prior works \cite{kashefi2004complexity,aaronson2016complexity,aharonov2022quantum}. 
In recent years, oracle separations~\cite{Kretschmer21,kqst23,lombardi2023one} have demonstrated that 
the complexity of breaking certain quantum cryptographic primitives is independent of the complexity of the decisional complexity classes $\class{NP}$ or $\class{QMA}$; in other words, 
even if $\class{P} = \class{NP}$, certain quantum cryptographic primitives could still remain secure. In fact, \cite{lombardi2023one} gives preliminary evidence that the ability to solve \emph{any} decision problem (even undecidable ones!) would not help with breaking quantum cryptography. 
Unitary complexity theory allows us to re-establish the link between complexity theory and cryptography in the quantum world.

Beyond this cryptographic motivation, unitary complexity allows us to relate the computational resources required for seemingly unrelated state transformation tasks.
In this paper, we focus on tasks involving \emph{Uhlmann transformations}.
The name stems from Uhlmann's theorem,~\cite{uhlmann1976transition} a fundamental result in quantum information theory that quantifies how well a bipartite pure state $\ket{C}$ can be mapped to another bipartite pure state $\ket{D}$ by only acting on a subsystem: letting $\rho$ and $\sigma$ denote the reduced density matrices on the first subsystem of $\ket{C}$ and $\ket{D}$, respectively, Uhlmann's theorem states that
\begin{equation}
    \label{eq:intro-1}
        \fidelity(\rho,\sigma) = \max_U \, | \bra{D} \id \otimes U \ket{C}|^2 \,,
\end{equation}
where $\fidelity(\rho,\sigma)$ denotes the fidelity function and the maximization is over all unitary transformations acting on the second subsystem. We call a unitary $U$ achieving equality in \cref{eq:intro-1} an \emph{Uhlmann transformation}.\footnote{Such Uhlmann transformations are unique only if $\ket{C},\ket{D}$ have full Schmidt rank.}

Uhlmann transformations are ubiquitous in quantum information processing. Some examples include:
\vspace{-6pt}
\begin{description}[itemsep=0pt]
    \item[Quantum Shannon theory.] Quantum Shannon theory is the study of the fundamental limits of quantum communication over noisy and noiseless channels. Protocols for a myriad of tasks such as state redistribution, entanglement distillation, and quantum communication over a noisy quantum channel all require performing Uhlmann transformations~\cite{hayden2008decoupling,abeyesinghe2009mother,berta2011quantum,anshu2017one}.

    \item[Quantum cryptography.] While it is known that quantum commitment schemes with information-theoretic security are impossible \cite{mayers1997unconditionally,lo1998quantum}, they are possible under computational assumptions.
    Recent oracle separations suggest that their security can be based on weaker assumptions than what is needed classically and that the existence of inherently quantum cryptographic primitives may be independent from assumptions in traditional complexity~\cite{Kretschmer21,ananth2022cryptography,morimae2022quantum,kqst23,lombardi2023one}. It can be seen from the impossibility results of Mayers--Lo--Chau \cite{mayers1997unconditionally,lo1998quantum} that the security of a quantum commitment scheme relies on the hardness of performing certain Uhlmann transformations.

    \item[Quantum gravity.]
    Attempts to unite quantum mechanics with general relativity have given rise to apparent paradoxes of whether black holes preserve information or not \cite{hawking1976breakdown}. Recently, physicists have provided intriguing arguments based on \emph{computational complexity} as possible resolutions to these paradoxes~\cite{Harlow_2013}. These arguments claim that distilling entanglement from the emitted Hawking radiation of a black hole is computationally infeasible --- this can be equivalently phrased as a statement about the hardness of an Uhlmann transformation~\cite{Harlow_2013,brakerski2022blackhole}.

    \item[Quantum complexity theory.]
The $\QIP = \PSPACE$ theorem~\cite{jain2011qip} gives a characterization of the power of (single-prover) quantum interactive proofs. Kitaev and Watrous~\cite{kitaev2000parallelization} showed that optimal prover strategies in these interactive proofs boil down to applying Uhlmann transformations at each round.
\end{description}

The fact that Uhlmann transformations appear in these various quantum tasks suggests that they might be related.
Can we formalize these relationships and show precise reductions, similarly to how e.g.~the theory of $\mathsf{NP}$ completeness relates disparate classical computational problems?
Can we formalize Uhlmann transformations as a computational problem that is, in some sense, \emph{complete} for these various tasks, similarly to how 3-SAT provides a simple complete problem that elegantly captures the hardness of NP-complete problems?
Can we provide complexity-theoretic evidence for the hardness of Uhlmann transformations? 
What computational restrictions does this place on our ability to e.g.~achieve optimal communication rates in quantum Shannon theory?

The goal of this paper is to study such questions formally. 
Our first main contribution is to provide a general formal framework for reasoning about unitary complexity (\cref{part:unitary_complexity_theory}). 
This involves extending many of the traditional notions of complexity theory, such as reductions, complexity classes, complete problems, etc.\ to quantum state transformations and requires us to deal with many subtleties that arise in the unitary setting.
Our second main contribution is to analyze the complexity of the Uhlmann Transformation Problem within this framework (\cref{part:uhlmann_general}). This in turn allows us to show relationships between unitary complexity classes such as showing that (average case versions of) the classes $\unitaryPSPACE$ and $\unitaryQIP$ are equal. 
Finally, we show how the Uhlmann transformation problem plays a central role in connecting the complexity of many natural tasks in quantum information processing (\cref{part:applications}). 
For example, we establish reductions and equivalences between Uhlmann transformation problem and the security of quantum commitment schemes, falsifiable quantum cryptographic assumptions, quantum state compression, and more.

\subsection{A fully quantum complexity theory} 

In~\cite{rosenthal2022interactive} Rosenthal and Yuen initiated the study of complexity classes for \emph{state synthesis} and \emph{unitary synthesis} problems. A state synthesis problem is a sequence $(\rho_x)_{x \in \bits^*}$ of quantum states. A \emph{state complexity class} is a collection of state synthesis problems that captures the computational resources needed to synthesize (i.e., generate) the states.
For example, \cite{rosenthal2022interactive} defined the class $\statePSPACE$ as the set of all state sequences $(\rho_x)_{x \in \bits^*}$ for which there is a polynomial-space (but possibly exponential-time) quantum algorithm $A$ that, on input $x$, outputs an approximation to the state $\rho_x$. 

\emph{Unitary complexity classes}, which are the focus of this work, describe the computational resources needed to perform state \emph{transformations}, formalized as \emph{unitary synthesis problems}. A unitary synthesis problem is a sequence of unitary\footnote{In our formal definition of unitary synthesis problems (see \cref{sec:defs}), the $U_x$'s are technically partial isometries, which is a promise version of unitaries, but we gloss over the distinction for now.}
operators $(U_x)_{x \in \bits^*}$ and a unitary complexity class is a collection of unitary synthesis problems. For example the class $\unitaryBQP$ is the set of all sequences of unitary operators $(U_x)_{x \in \bits^*}$ where there is a polynomial-time quantum algorithm $A$ that, given an \emph{instance} $x \in \bits^*$ and a quantum system $\reg{B}$ as input, (approximately) applies $U_x$ to system $\reg{B}$.
As a simple example, any sequence of unitaries $(U_x)$ where $x$ is simply (an explicit encoding of) a sequence of quantum gates that implement the unitary is obviously in $\unitaryBQP$, since given $x$, the algorithm $A$ can just execute the circuit specified by $x$ in time polynomial in the length of $x$.
On the other hand, $x$ could also specify a unitary in a sequence in a more implicit way (e.g.\ by circuits for two quantum states between which $U_x$ is meant to be the Uhlmann transformation), in which case the sequence $(U_x)_x$ could be harder to implement.

The reason we say that the algorithm $A$ is given a \emph{system} instead of a \emph{state} is to emphasize that the state of the system is not known to the algorithm ahead of time, and in fact the system may be part of a larger entangled state. Thus the algorithm has to coherently apply the transformation $U_x$ to the given system, maintaining any entanglement with an external system. This makes unitary synthesis problems fundamentally different, and in many cases harder to analyse, than state synthesis problems.

Traditional complexity classes like $\sf P$, $\sf NP$, and $\sf BQP$ have proven to be powerful ways of organizing and comparing the difficulty of different decision problems. In a similar way, state and unitary complexity classes are useful for studying the complexity of quantum states and of quantum state transformations. We can then ask about the existence of complete problems, reductions, inclusions, separations, closure properties, and more. Importantly, state and unitary complexity classes provide a useful language to articulate questions and conjectures about the computational hardness of inherently quantum problems. For example, we can ask whether $\unitaryPSPACE$ is contained in $ \unitaryBQP^{\PSPACE}$ --- in other words, can polynomial-space-computable unitary transformations be also computed by a polynomial-time quantum computer that is given oracle access to a $\PSPACE$ decision oracle?\footnote{
    This is an open question, and is related to the ``Unitary Synthesis Problem'' raised by Aaronson and Kuperberg~\cite{aaronson2007quantum}.
}

\paragraph{Unitary synthesis problems, classes, and reductions.} We begin by giving general definitions for unitary synthesis problems and a number of useful unitary complexity classes, e.g.~$\unitaryBQP$ and $\unitaryPSPACE$. We then define a notion of \emph{reductions} between unitary synthesis problems. Roughly speaking, we say that a unitary synthesis problem $\usynth{U} = (U_x)_x$ polynomial-time reduces to $\usynth{V} = (V_x)_x$ if an efficient algorithm for implementing $\usynth{V}$ implies an efficient algorithm for implementing $\usynth{U}$. 

Next, we define \emph{distributional} unitary complexity classes that capture the \emph{average case complexity} of solving a unitary synthesis problem. Here, the unitary only needs to be implemented on an input state \emph{randomly chosen} from some distribution $\cal{D}$ which is known ahead of time.
This is a natural generalisation of traditional average-case complexity statements to the unitary setting.
This notion turns out to be particularly natural in the context of entanglement transformation problems because it is closely related to implementing the unitary on part of an entangled state $\ket{\psi}$.

The notion of average case complexity turns out to be central to our paper: nearly all of our results are about average-case unitary complexity classes and the average-case complexity of the Uhlmann Transformation Problem. 
Thus the unitary complexity classes we mainly deal with will be $\avgUnitaryBQP$ and $\avgUnitaryPSPACE$, which informally mean sequences of unitaries that can be implemented by time-efficient and space-efficient quantum algorithms, respectively, and where the implementation error is measured with respect to inputs drawn from a fixed distribution over quantum states.

See \cref{sec:defs} for details as well as more discussion regarding the choices we made for our definitions. 

\paragraph{Interactive proofs for unitary synthesis.} We then explore models of \emph{interactive proofs} for unitary synthesis problems. Roughly speaking, in an interactive proof for a unitary synthesis problem $\usynth U = (U_x)_x$, a polynomial-time verifier receives an instance $x$ and a quantum system $\reg{B}$ as input, and interacts with an all-powerful but untrusted prover to try to apply $U_x$ to system $\reg{B}$. 
As usual in interactive proofs, the main challenge is that the verifier does not trust the prover, so the protocol has to test whether the prover actually behaves as intended. 
We formalize this with the complexity classes $\unitaryQIP$ and $\avgUnitaryQIP$, which capture unitary synthesis problem that can be verifiably implemented in this interactive model.
This generalizes the interactive state synthesis model studied by~\cite{rosenthal2022interactive,metger2023stateqip}.\footnote{The class $\unitaryQIP$ was also briefly discussed informally by Rosenthal and Yuen~\cite{rosenthal2022interactive}.} The primary difference between the state synthesis and unitary synthesis models is that in the former, the verifier  starts with a fixed input state (say, the all zeroes state), while in the latter the verifier receives a quantum system $\reg{B}$ in an unknown state that has to be transformed by $U_x$. 
See \cref{sec:protocols} for more details.

\paragraph{Zero-knowledge unitary synthesis.} In the context of interactive protocols, we also introduce a notion of \emph{zero-knowledge protocols} for unitary synthesis problems. Roughly speaking, a protocol is zero-knowledge if the interaction between the verifier and prover can be efficiently reproduced by an algorithm (called the \emph{simulator}) that does not interact with the prover at all. This way, the verifier can be thought of as having learned no additional knowledge from the interaction aside from the fact that the task was solved~\cite{goldwasser1989knowledge}. The counterintuitive concept of zero-knowledge proofs has been one of the most consequential discoveries in complexity theory and cryptography. 

Motivated by this, we introduce the unitary complexity class $\avgUnitary{HVSZK}$,\footnote{The ``$\mathsf{HV}$'' modifier signifies that the zero-knowledge property is only required to hold with respect to verifiers that honestly follow the protocol, and the ``$\mathsf{S}$'' in ``$\mathsf{SZK}$'' signifies that it is \emph{statistical} zero-knowledge.} which is a unitary synthesis analogue of the decision class $\mathsf{HVQSZK}$ in traditional complexity theory~\cite{watrous2006zero}, which captures the concept of \emph{honest-verifier quantum zero-knowledge proofs}. Interestingly, for reasons that we explain in more detail in \cref{sec:zk}, the average-case aspect of $\avgUnitary{HVSZK}$ appears to be necessary to obtain a nontrivial definition of zero-knowledge in the unitary synthesis setting. 

\medskip
\vspace{1em}

Just like there is a zoo of traditional complexity classes~\cite{Zoo}, we expect that many unitary complexity classes can also be meaningfully defined and explored. In this paper we focus on the ones that turn out to be tightly related to the Uhlmann Transformation Problem. We discuss these relationships next. 

\begin{remark}
    For simplicity, in the introduction we present informal statements of our results that gloss over some technical details that would otherwise complicate the result statement. For example, we do not distinguish between unitary synthesis problems and distributional versions of them. After each informal result statement we point the reader to where the formal result is stated and proved.
\end{remark}

\subsection{Structural results about the Uhlmann Transformation Problem}

Equipped with the proper language to talk about unitary synthesis problems, we present the Uhlmann Transformation Problem in \Cref{part:uhlmann_general} of this paper. We define the unitary synthesis problem $\Uhlmann$ to be the sequence $(U_x)_{x \in \bits^*}$ where we interpret an instance $x$ as an explicit encoding (as a list of gates) of a pair of quantum circuits $(C,D)$ such that $C$ and $D$, on the all-zeroes input, output pure bipartite states $\ket{C},\ket{D}$ on the same number of qubits, and $U_x$ is an associated Uhlmann transformation mapping $\ket{C}$ to $\ket{D}$ by acting on a local system.
Usually, we will assume that $C$ and $D$ output $2n$ qubits (for some $n$ specified as part of $x$) and the Uhlmann transformation acts on the last $n$ qubits.
If $x$ does not specify such a pair, then an algorithm implementing the unitary synthesis problem is allowed to behave arbitrarily on such $x$; this is formally captured by allowing partial isometries as part of unitary synthesis problems in \cref{def:unitary_synth_problem}. 

Furthermore, for a parameter $0 \leq \kappa \leq 1$ we define the problem $\Uhlmann_\kappa$, which is the same as $\Uhlmann$, except that it is restricted to instances corresponding to states $\ket{C},\ket{D}$ where the fidelity between the reduced density matrices $\rho,\sigma$ of $\ket{C},\ket{D}$ respectively on the first subsystem is at least $\kappa$; recall by Uhlmann's theorem that $\kappa$ lower bounds how much overlap $\ket{C}$ can achieve with $\ket{D}$ by a local transformation. By definition, $\Uhlmann_\kappa$ instances are at least as hard as $\Uhlmann_{\kappa'}$ instances when $\kappa \leq \kappa'$.
We provide formal definitions of $\Uhlmann$, $\Uhlmann_\kappa$, and their distributional versions in \cref{sec:uhlmann}.

\paragraph{Zero-knowledge and the Uhlmann Transformation Problem.} We show that the Uhlmann Transformation Problem (with fidelity parameter $\kappa = 1$) \emph{characterizes} the complexity of the unitary complexity class $\avgUnitary{HVPZK}$, which is the unitary synthesis version of the decision classes $\class{PZK}$ and $\class{HVQPZK}$ \cite{watrous2002limits}. Here, $\mathsf{PZK}$ stands for ``perfect zero knowledge'', and refers to the special case of statistical zero-knowledge where the simulator can \emph{perfectly} reproduce the view of the verifier. 

\begin{theorem}[Informal] \label{thm:pzk_complete_intro}
    $\Uhlmann_1$ is complete for $\avgUnitary{HVPZK}$ under polynomial-time reductions. 
\end{theorem}

This is formally stated and proved in \Cref{sec:pzk-completeness}. To show completeness we have to prove two directions. The first direction is to show that every (distributional) unitary synthesis problem in $\avgUnitary{HVPZK}$ polynomial-time reduces to (the distributional version of) $\Uhlmann_{1}$.
This uses a characterization of quantum interactive protocols due to Kitaev and Watrous~\cite{kitaev2000parallelization}.

The second direction is to show that $\Uhlmann_1$ is in $\avgUnitary{HVPZK}$ by exhibiting an (honest-verifier) zero-knowledge protocol to solve the Uhlmann Transformation Problem. Our protocol is rather simple: in the average case setting, we assume that the verifier receives the last $n$ qubits of the state $\ket{C} = C\ket{0^{2n}}$, and the other half is inaccessible. Its goal is to transform, with the help of a prover, the global state $\ket{C}$ to $\ket{D}$ by only acting on the last $n$ qubits that it received as input.
To this end, the verifier generates a ``test'' copy of $\ket{C}$ on its own, which it can do because $C$ is a polynomial-size circuit. The verifier then sends to the prover two registers of $n$ qubits; one of them is the first half of the test copy and one of them (call it $\reg{A}$) holds the ``true'' input state. The two registers are randomly shuffled. The prover is supposed to apply the Uhlmann transformation $U$ to both registers and send them back. The verifier checks whether the ``test'' copy of $\ket{C}$ has been transformed to $\ket{D}$ by applying the inverse circuit $D^\dagger$ to the test copy and checking if all qubits are zero.  If so, it accepts and outputs the register $\reg{A}$, otherwise the verifier rejects.

If the prover is behaving as intended, then both the test copy and the ``true'' copy of $\ket{C}$ are transformed to $\ket{D}$. Furthermore, the prover cannot tell which of its two registers corresponds to the test copy, and thus if it wants to pass the verification with high probability, it has to apply the correct Uhlmann transformation on both registers. This shows that the protocol satisfies the completeness and soundness properties of an interactive proof. The zero-knowledge property is also straightforward: if both the verifier and prover are acting according to the protocol, then before the verifier's first message to the prover, the reduced state of the verifier is $\ketbra{C}{C} \otimes \rho$ (where $\rho$ is the reduced density matrix of $\ket{C}$), and at the end of the protocol, the verifier's state is $\ketbra{D}{D} \otimes U \rho U^\dagger$. Both states can be produced in polynomial time.

One may ask: if the simulator can efficiently compute the state $U \rho U^\dagger$ without the help of the prover, does that mean the Uhlmann transformation $U$ can be implemented in polynomial time? The answer is no, since the simulator only has to prepare the appropriate reduced state (i.e.~essentially solve a state synthesis task), which is easy since the starting and ending states of the protocol are efficiently computable; in particular, $U \rho U^\dagger$ is (approximately) the reduced state of $\ket{D}$, which is easy to prepare. In contrast, the verifier has to implement the Uhlmann transformation on a \emph{specific} set of qubits that are entangled with a \emph{specific} external register, i.e.\ it has to perform a state transformation task that preserves coherence with the purifying register.
This again highlights the distinction between state and unitary synthesis tasks.

\paragraph{A complete problem for $\avgUnitary{HVSZK}$?} It is natural to wonder about the complexity of $\Uhlmann_\kappa$ for fidelity promise $\kappa < 1$. In other words, the reduced density matrices of the two states $\ket{C},\ket{D}$ are not exactly equal. A reasonable conjecture is that $\Uhlmann_\kappa$ (for non-negligible $\kappa$, say), is complete for $\avgUnitary{HVSZK}$. This would correspond to the famous classical complexity result that the problem of distinguishing between whether two probability distributions (represented via sampling circuits) are close or far in trace distance is a $\mathsf{SZK}$-complete problem~\cite{10.1145/636865.636868}.

In \Cref{sec:polarize} we argue that this conjecture is true assuming that a unitary version of the \emph{polarization lemma} holds, which was instrumental for the $\mathsf{SZK}$-completeness result of Sahai and Vadhan~\cite{10.1145/636865.636868}. The unitary polarization lemma, if true, would state that $\Uhlmann_\kappa$ polynomial-time reduces to $\Uhlmann_{1 - 2^{-\poly(n)}}$ for all inverse polynomial $\kappa$.

\paragraph{The succinct Uhlmann Transformation Problem.} We also define a \emph{succinct} version of the Uhlmann Transformation Problem (denoted by $\SuccinctUhlmann$), where the string $x$ encodes a pair $(\hat{C},\hat{D})$ of \emph{succinct descriptions} of quantum circuits $C,D$. By this we mean that $\hat{C}$ (resp.~$\hat{D}$) is a classical circuit that, given a number $i \in \N$ written in binary, outputs the $i$'th gate in the quantum circuit $C$ (resp.~$D$). Thus the circuits $C$, $D$ in general can have \emph{exponential} depth (in the length of the instance string $x$) and generate states $\ket{C},\ket{D}$ that are unlikely to be synthesizable in polynomial time. Thus the task of synthesizing the Uhlmann transformation $U$ that maps $\ket{C}$ to a state with maximum overlap with $\ket{D}$, intuitively, should be much harder than the non-succinct version. We confirm this intuition with the following result:

\begin{theorem}[Informal]
\label{thm:intro:succinct-uhlmann-unitary-pspace-completeness}
    $\SuccinctUhlmann$ is complete for $\avgUnitaryPSPACE$ under polynomial-time reductions.
\end{theorem}

The class $\avgUnitaryPSPACE$ corresponds to distributional unitary synthesis problems that can be solved using a polynomial-space (but potentially exponential-depth) quantum algorithm. 
The fact that $\SuccinctUhlmann \in \avgUnitaryPSPACE$ was already proved by Metger and Yuen~\cite{metger2023stateqip}, who used this to show that optimal prover strategies for quantum interactive proofs can be implemented in $\avgUnitaryPSPACE$.\footnote{This was phrased in a different way in their paper, as $\avgUnitaryPSPACE$ was not yet defined.} The fact that $\avgUnitaryPSPACE$ reduces to $\SuccinctUhlmann$ is because solving a distributional unitary synthesis problem $(U_x)_x$ in $\avgUnitaryPSPACE$ is equivalent to applying a local unitary that transforms an entangled state $\ket{\psi_x}$ representing the distribution to $(\id \otimes U_x)\ket{\psi_x}$. This is nothing but an instance of the $\SuccinctUhlmann$ transformation problem. We refer to the proof of \cref{lem:succuhl_hard_for_pspace} for details.

We show another completeness result for $\SuccinctUhlmann$:

\begin{theorem}[Informal]
\label{thm:intro:succinct-uhlmann-unitary-qip-completeness}
    $\SuccinctUhlmann$ is complete for $\avgUnitaryQIP$ under polynomial-time reductions.
\end{theorem}

Here, the class $\avgUnitaryQIP$ is like $\avgUnitary{HVPZK}$ except there is no requirement that the protocol between  the honest verifier and prover can be efficiently simulated. The proof of \Cref{thm:intro:succinct-uhlmann-unitary-qip-completeness} starts similarly to the proof of the $\avgUnitary{HVPZK}$-completeness of $\Uhlmann$, but requires additional ingredients, such as the state synthesis protocol of~\cite{rosenthal2022interactive,rosenthal2024efficient} and the ability to simulate reflections about a state, given copies of the state~\cite{JLS18}. We prove this by showing that $\SuccinctUhlmann$ is contained in $\avgUnitaryQIP$ (\Cref{lem:succuhl_in_qip}), $\avgUnitaryQIP \subseteq \avgUnitaryPSPACE$ (\Cref{lem:qip_in_pspace}), and then argue that $\avgUnitaryPSPACE$ is polynomial-time reducible to $\SuccinctUhlmann$ (\Cref{lem:succuhl_hard_for_pspace}). 

\Cref{thm:intro:succinct-uhlmann-unitary-pspace-completeness,thm:intro:succinct-uhlmann-unitary-qip-completeness} imply the following unitary complexity analogue of the $\QIP = \PSPACE$ theorem~\cite{jain2011qip} and the $\stateQIP = \statePSPACE$ theorem~\cite{rosenthal2022interactive,metger2023stateqip}:

\begin{theorem}
    $\avgUnitaryQIP = \avgUnitaryPSPACE$.
\end{theorem}

This partially answers an open question of~\cite{rosenthal2022interactive,metger2023stateqip}, who asked whether $\unitaryQIP = \unitaryPSPACE$ (although they did not formalize this question to the same level as we do here).

\subsection{Centrality of the Uhlmann Transformation Problem}

In \Cref{part:applications} of this paper we relate the Uhlmann Transformation Problem to quantum information processing tasks in a variety of areas: quantum cryptography, quantum Shannon theory, and high energy physics. We show that the computational complexity of a number of these tasks is in fact essentially \emph{equivalent} to the hardness of $\Uhlmann$. For some other problems we show that they are efficiently reducible to $\Uhlmann$ or $\SuccinctUhlmann$. 
Although some of these connections have been already observed in prior work, we believe that the framework of unitary complexity theory formalizes and clarifies the relationships between these different problems. 

We proceed to give a high level overview of our applications of the Uhlmann Transformation Problem and unitary complexity theory. 

\subsubsection{Quantum cryptography}

We show that the Uhlmann Transformation Problem is deeply intertwined with the security of quantum cryptography. First, we show the security of quantum commitment schemes is \emph{equivalent} to the average-case hardness of the Uhlmann Transformation Problem.

\paragraph{Quantum commitments.} A bit commitment scheme is a fundamental cryptographic primitive that allows two parties (called a \emph{sender} and \emph{receiver}) to engage in a two-phase communication protocol:
in the first phase (the ``commit phase''), the sender sends a commitment (i.e.\ some string) to a bit $b$ to the receiver;  the \emph{hiding} property of a bit commitment scheme ensures that the receiver cannot decide the value of $b$ from this commitment string alone.
In the second phase (the ``reveal phase''), the sender sends another string to the receiver that allows the receiver to compute the value of $b$; the \emph{binding} property of commitments ensures that the sender can only reveal the correct value of $b$, i.e.\ if the sender sent a reveal string that was meant to convince the receiver it had committed to a different value of $b$, the receiver would detect this. 

Commitment schemes --- even quantum ones --- require efficiency constraints on the adversary \cite{mayers1997unconditionally,lo1998quantum}; at least one of the hiding or binding properties must be computational.
In classical cryptography, commitment schemes can be constructed from one-way functions~\cite{Naor2003}, but recent works suggest the possibility of basing quantum commitment schemes on weaker, inherently quantum assumptions such as the existence of pseudorandom states~\cite{Kretschmer21,ananth2022cryptography,morimae2022quantum,kqst23} or EFI pairs~\cite{brakerski2022computational}. 

The following theorem shows that the existence of secure quantum commitment schemes is essentially equivalent to $\Uhlmann$ being hard on average. Roughly speaking, hardness on average means that there is an efficiently sampleable distribution over pairs of quantum circuits $(C,D)$ such that all polynomial-time algorithms fail to implement the Uhlmann transformation corresponding to $(\ket{C},\ket{D})$ with non-negligible probability over the sampling of $(C,D)$. 
\begin{theorem}[Informal]
\label{thm:intro-commitments}
$\Uhlmann_{1-\eps}$ for some negligible $\eps$ is hard on average if and only if secure quantum commitments exist. 
\end{theorem}
This theorem is formally stated and proved as \Cref{thm:uhlmann-hardness-implies-commmitments}. This formalizes a connection between Uhlmann transformations and quantum commitments that was suggested by Yan in his in-depth study of properties of quantum bit commitments~\cite{yan2023general}. The necessity for the hardness of $\Uhlmann$ is implicit in the original impossibility proofs of information-theoretic security for commitments~\cite{mayers1997unconditionally,lo1998quantum}; the sufficiency is due to the fact that \emph{non-interactive quantum commitments} can be constructed from hard $\Uhlmann$ instances. 

Given the close connection between zero knowledge protocols for unitary synthesis and the Uhlmann Transformation Problem, we also prove the following:
\begin{theorem}[Informal]
\label{thm:intro-commitment-szk}
    If there is a hard distribution of instances for $\avgUnitary{HVSZK}$, then secure quantum commitments exist. 
\end{theorem}
We note that this would follow as an immediate corollary if we were able to prove that $\Uhlmann_{1 - \eps}$ is a complete problem for $\avgUnitary{HVSZK}$; however as mentioned previously this remains a conjecture. We instead prove \Cref{thm:intro-commitment-szk} directly by showing that hard-on-average problems in $\avgUnitary{HVSZK}$ implies $\Uhlmann_{1 - \eps}$ is hard on average.

This is analogous to the classical result of Ostrovsky~\cite{ostrovsky1991one} who showed that if the classical complexity class $\mathsf{SZK}$ is hard on average, then one-way functions (and thus secure bit commitments~\cite{naor1991bit}) exist. This is formally stated and proved as \Cref{thm:hvszk-hardness-implies-commitments}.

\paragraph{Minimal assumptions in quantum cryptography.} In classical cryptography, the existence of one-way functions is considered a \emph{minimal assumption} in the sense that the security of virtually all (classical) cryptography implies it~\cite{IL89,impagliazzo1995personal}. It is a fascinating open question of what is the minimal assumption (if there exists one) in quantum cryptography; as of writing the leading contender for the minimal quantum cryptographic assumption is the existence of quantum commitments, meaning that many quantum cryptographic primitives can be shown to imply the existence of quantum commitments~\cite{brakerski2022computational,khurana2023commitments}. If quantum commitments are indeed minimal (mirroring the setting of classical cryptography), then this would show that the hardness of the Uhlmann Transformation Problem is necessary for computationally secure quantum cryptography.

\paragraph{Breaking falsifiable quantum cryptographic assumptions.}
While we don't know yet if the hardness of the Uhlmann Transformation Problem is necessary for computational quantum cryptography, we show that the hardness of the \emph{succinct} Uhlmann Transformation Problem is necessary for the security of a wide class of quantum cryptographic primitives. We consider the general notion of a \emph{falsifiable quantum cryptographic assumption}, which can be seen as a quantum analogue of the notion of a falsifiable assumption considered by Naor~\cite{Naor2003} as well as Gentry and Wichs~\cite{cryptoeprint:2010/610}. 
Our notion of a falsifiable quantum cryptographic assumption captures almost any reasonable definition of security in quantum cryptography which can be phrased in terms of an interactive \emph{security game} between an adversary and a challenger.
We show the following generic upper bound on the complexity of breaking falsifiable quantum cryptographic assumptions (see \cref{thm:falsifiable} for the formal statement):

\begin{theorem}[Informal]
A falsifiable quantum cryptographic assumption is either information-theoretically secure, or the task of breaking security reduces to $\SuccinctUhlmann$.
\end{theorem}

Since $\SuccinctUhlmann$ is complete for $\avgUnitaryPSPACE$ (\cref{thm:intro:succinct-uhlmann-unitary-pspace-completeness}), this means that $\avgUnitaryBQP \neq \avgUnitaryPSPACE$ is a necessary complexity-theoretic assumption for computational quantum cryptography. 
This suggests that unitary complexity provides the appropriate framework to establish a close link between complexity theory and quantum cryptography, as recent work~\cite{Kretschmer21,ananth2022cryptography,morimae2022quantum,kqst23,lombardi2023one} has shown that traditional complexity theoretic assumptions are not always linked to quantum cryptography in the way one would expect.

\subsubsection{Quantum Shannon theory applications} 

Quantum Shannon theory studies the achievability and limits of quantum communication tasks (see~\cite{wilde2013quantum,khatri2020principles,renes2015quantum} for a comprehensive overview).
While the information-theoretic aspects of quantum communication tasks are well-understood, the complexity of implementing these protocols has received remarkably little attention.
Here, we study the computational complexity of some fundamental tasks in quantum Shannon theory, namely noisy channel decoding and compression of quantum states using our framework for unitary complexity and our results on the Uhlmann transformation problem.\footnote{We also note that in independent work after the publication of our results, Arnon-Friedman, Brakerski, and Vidick have investigated the computational aspects of entanglement distillation~\cite{arnon2023computational}, showing that in general entanglement distillation is computationally infeasible assuming quantum commitments exist.
It would be interesting to connect their results to our framework for unitary complexity to build up a more rigorous theory of the complexity of quantum Shannon tasks.}

\paragraph{Decodable channel problem.} Consider a quantum channel $\cal{N}$ 
that maps a register $\reg{A}$ to a register $\reg{B}$. Suppose that the channel $\cal{N}$ is \emph{decodable}, meaning that it is possible to information-theoretically (approximately) recover the information sent through the channel; i.e., there exists a decoding channel $\cD$ mapping register $\reg{B}$ back to register $\reg{A}$ such that $\cD_{\reg{B} \to \reg{A}'} \Big ( \cal{N}_{\reg{A} \to \reg{B}} (\Phi_{\reg{AR}}) \Big) \approx \Phi_{\reg{A}' \reg{R}}$, where $\ket{\Phi}_{\reg{AR}}$ is the maximally entangled state. Note that the register $\reg{R}$ is not touched.

Important examples of decodable channels come from coding schemes for noisy quantum channels: suppose $\cal{K}$ is a noisy quantum channel that has capacity $C$ (meaning it is possible to (asymptotically) transmit $C$ qubits through $\cal{K}$). Let $\cE$ denote a channel that takes $C$ qubits and maps it to an input to $\cal{K}$. For example, we can think of $\cE$ as an encoder for a quantum error-correcting code. If $\cE$ is a good encoding map, the composite channel $\cal{N} : \rho \mapsto \cal{K} ( \cE (\rho))$ is decodable.

We define the \emph{Decodable Channel Problem}: given as input a circuit description of a channel $\cal{N}$ that maps register $\reg{A}$ to register $\reg{B}$ and furthermore is promised to be decodable, and given the register $\reg{B}$ of the state $(\cal{N} \otimes \id)(\Phi_{\reg{AR}})$, decode and output a register $\reg{A}'\equiv A$ such that the final joint state of $\reg{A}' \reg{R}$ is close to $\ket{\Phi}$. Although it is information-theoretically possible to decode the output of $\cal{N}$, it may be computationally intractable to do so. In fact, we can characterize the complexity of the Decodable Channel Problem:

\begin{theorem}[Informal]
    The Decodable Channel Problem can be solved in polynomial-time up to inverse polynomial error if and only if $\Uhlmann$ can be solved in polynomial-time up to inverse polynomial error.\end{theorem}

This theorem is formally stated and proved as \Cref{thm:complexity-decodable-channels}; since we do not expect that $\Uhlmann$ is solvable in polynomial-time, this suggests that the Decodable Channel Problem is computationally hard in general. The main idea behind the upper bound (Decodable Channel Problem is easy if $\Uhlmann$ is easy) is that a channel $\cal{N}$ is decodable if and only if the output of the \emph{complementary channel}\footnote{The output of the complementary channel can be thought of as the qubits that a purification (formally, a Stinepring dilation) of the channel $\cal{N}$ discards to the environment.} $\cal{N}^c$, when given register $\reg{A}$ of the maximally entangled state $\ket{\Phi}_{\reg{AR}}$, is approximately unentangled with register $\reg{R}$. Thus by Uhlmann's theorem there exists an Uhlmann transformation acting on the output of the channel $\cal{N}$ that recovers the maximally entangled state. If $\Uhlmann \in \avgUnitaryBQP$, then this transformation can be performed efficiently. 

The proof of the lower bound (Decodable Channel Problem is hard if $\Uhlmann$ is hard) draws inspiration from quantum commitments. As discussed earlier, the hardness of $\Uhlmann$ essentially implies the existence of secure quantum commitments, and in particular one where the hiding property is computational. From this, we can construct a hard instance of the Decodable Channel Problem:
consider a channel $\cal{N}$ that takes as input a single bit $\ket{b}$, and then outputs the commitment register of the commitment to bit $b$ (and discards the reveal register). The ability to decode this ``commitment channel'' implies the ability to break the hiding property of the underlying commitment scheme, and therefore decoding must be computationally hard.

\paragraph{Compression of quantum information.} Another fundamental task in information theory --- both classical and quantum --- is compression of data. Shannon's source coding theorem shows that the Shannon entropy of a random variable $X$ characterizes the rate at which many independent copies of $X$ can be compressed~\cite{shannon1948mathematical}. Similarly, Schumacher proved that the von Neumann entropy of a density matrix $\rho$ characterizes the rate at which many independent copies of $\rho$ can be (coherently) compressed~\cite{schumacher1995quantum}. 

We consider the \emph{one-shot} version of the information compression task, where one is given just one copy of a density matrix $\rho$ (rather than many copies) and the goal is to compress it to as few qubits as possible while being able to recover the original state within some error. In the one-shot setting the von Neumann entropy no longer characterizes the optimal compression of $\rho$; instead this is given by a one-shot entropic quantity known as the \emph{smoothed max-entropy}~\cite{tomamichel2012framework}. What is the computational effort required to perform near-optimal one-shot compression of quantum states? Our next result gives upper and lower bounds for the computational complexity of this task:

\begin{theorem}[Informal]
    Quantum states can be optimally compressed to their smoothed max entropy in polynomial-time if $\Uhlmann_{1 - \eps} \in \avgUnitaryBQP$ for some negligible $\eps$. Furthermore, if stretch pseudorandom state generators exist, then optimal compression of quantum states cannot be done in polynomial time.
\end{theorem}
This theorem is formally stated and proved as~\Cref{thm:comp-compression,thm:comp-compression-lb}.
The upper bound (i.e., compression is easy if $\Uhlmann$ is easy) is proved using a powerful technique in quantum information theory known as \emph{decoupling}~\cite{dupuis2010decoupling}.
The hardness result for compression is proved using a variant of \emph{pseudorandom states}, a cryptographic primitive that is a quantum analogue of pseudorandom generators~\cite{JLS18}.

\subsubsection{Black-hole radiation decoding}
In recent years, quantum information and quantum complexity have provided a new lens on long-standing questions surrounding the quantum-mechanical description of black holes. \cite{preskill1992black,almheiri2013black,Harlow_2013,brown2016complexity,susskind2016computational,bouland2019computational,yang2023complexity}.
We consider applications of the Uhlmann Transformation Problem to computational tasks arising from this research. 

In particular, we consider the Harlow-Hayden \emph{black hole radiation decoding task}~\cite{Harlow_2013}, which is defined as follows.  We are given as input a circuit description of a tripartite state $\ket{\psi}_{\reg{BHR}}$ that represents the global pure state of a single qubit (register $\reg{B}$), the interior of a black hole (register $\reg{H}$), and the Hawking radiation that has been emitted by the black hole (register $\reg{R}$). Moreover, we are promised that it is possible to \emph{decode} from the emitted radiation $\reg{R}$ a single qubit $\reg{A}$ that forms a maximally entangled state $\ket{\mathrm{EPR}} = \frac{1}{\sqrt{2}} (\ket{00} + \ket{11})$ with register $\reg{B}$. The task is to perform this decoding when given register $\reg{R}$ of a system in the state $\ket{\psi}$. 

Harlow and Hayden~\cite{Harlow_2013} showed that the decoding task is computationally intractable assuming that $\mathsf{SZK} \not\subseteq \mathsf{BQP}$. 
However, precisely characterizing the task's complexity (i.e., providing an equivalence rather than a one-way implication) appears to require the notions of a fully quantum complexity theory. Brakerski recently showed that this task is equivalent to breaking the security of a quantum cryptographic primitive known as EFI pairs~\cite{brakerski2022blackhole}. We reformulate this equivalence in our unitary complexity framework to show that black hole radiation decoding (as formalised above) can be solved in polynomial-time if and only if $\Uhlmann \in \avgUnitaryBQP$.

\subsection{Summary and future directions}

Computational tasks with quantum inputs and/or outputs are ubiquitous throughout quantum information processing. The traditional framework of complexity theory, which is focused on computational tasks with classical inputs and outputs, cannot naturally capture the complexity of these ``fully quantum'' tasks. 

In this paper we introduce a framework to reason about the computational complexity of unitary synthesis problems. We then use this framework to study Uhlmann's theorem through an algorithmic lens, i.e.\ to study the complexity of Uhlmann transformations. We prove that variants of the Uhlmann Transformation Problem are complete for some unitary complexity classes, and then explore relationships between the Uhlmann Transformation Problem and computational tasks in quantum cryptography, quantum Shannon theory, and high energy physics. 

The study of the complexity of state transformation tasks is a very new field and we hope that our formal framework of unitary complexity theory and our findings about the Uhlmann Transformation Problem provide a useful starting point for a rich theory of the complexity of ``fully quantum'' problems. 
Many questions in this direction have yet to be explored.
Throughout this paper, we have included many concrete open problems, which we hope will spark future research in this new direction in complexity theory.
Additionally, our work suggests some high-level, open-ended future directions to explore:

\paragraph{Populating the zoo.} An important source of the richness of computational complexity theory is the variety of computational problems that are studied. For example, the class $\class{NP}$ is so interesting because it contains many complete problems that are naturally studied across the sciences~\cite{papadimitriou1997np}, and the theory of $\mathsf{NP}$-completeness gives a unified way to relate them to each other. 

Similarly, a fully quantum complexity theory should have its own zoo of problems drawn from a diverse range of areas. We have shown that core computational problems in quantum cryptography, quantum Shannon theory, and high energy physics can be related to each other through the language of unitary complexity theory. What are other natural problems in e.g.~quantum error-correction, quantum metrology, quantum chemistry, or condensed matter physics, and what can we say about their computational complexity?

\paragraph{The crypto angle.} Complexity and cryptography are intimately intertwined. Operational tasks in cryptography have motivated models and concepts that have proved indispensible in complexity theory (such as pseudorandomness and zero-knowledge proofs), and conversely complexity theory has provided a rigorous theoretical foundation to study cryptographic hardness assumptions. 

We believe that there can be a similarly symbiotic relationship between quantum cryptography and a fully quantum complexity theory. 
Recent quantum cryptographic primitives such as quantum pseudorandom states~\cite{JLS18} or one-way state generators~\cite{morimae2022quantum}
are unique to the quantum setting, and the relationships between them are barely understood. For example, an outstanding question is whether there is a meaningful \emph{minimal hardness assumption} in quantum cryptography, just like one-way functions are in classical cryptography. 
Can a fully quantum complexity theory help answer this question about minimal quantum cryptographic assumptions, or at least provide some guidance? For example, there are many beautiful connections between one-way functions, average-case complexity, and Kolomogorov complexity~\cite{IL89,impagliazzo1995personal,liu2020one}. Do analogous results hold in the fully quantum setting? 

\paragraph{The learning theory angle.} Quantum learning theory has also seen rapid development, particularly on the topic of quantum state learning~\cite{aaronson2007learnability,huang2020predicting,buadescu2021improved,anshu2023survey}. Learning quantum states or quantum processes can most naturally be formulated as tasks with quantum inputs. Traditionally these tasks have been studied in the information-theoretic setting, where sample complexity is usually the main measure of interest. However we can also study the computational difficulty of learning quantum objects. What does a complexity theory of quantum learning look like?

\paragraph{Traditional versus fully quantum complexity theory.} While traditional complexity theory appears to have difficulty reasoning about fully quantum tasks, can we obtain \emph{formal} evidence that the two theories are, in a sense, independent of each other? For example, can we show that $\mathsf{P} = \mathsf{PSPACE}$ does not imply $\unitaryBQP = \unitaryPSPACE$? One would likely have to show this in a \emph{relativized} setting, i.e., exhibit an oracle $O$ relative to which $\mathsf{P}^O = \mathsf{PSPACE}^O$ but $\unitaryBQP^O \neq \unitaryPSPACE^O$. Another way would be to settle Aaronson and Kuperberg's ``Unitary Synthesis Problem''~\cite{aaronson2007quantum} in the negative; see~\cite{lombardi2023one} for progress on this. Such results would give compelling evidence that the reasons for the hardness of unitary transformations are intrinsically different than the reasons for the hardness of a Boolean function. More generally, what are other ways of separating traditional from fully quantum complexity theory?

\subsection*{Guide for readers}
Although the paper is rather long, the material is organized in a way that supports random-access reading -- depending on your interests, it is not necessary to read Section $X$ before reading Section $X+1$. All sections depend on the basic definitions of unitary complexity theory (\cref{sec:defs}) and the basic definitions of the Uhlmann Transformation Problem (\cref{sec:uhlmann}). From then on, it's choose-your-own-adventure. If you are interested in:
\begin{itemize}
    \item \textbf{Structural results about the complexity of $\Uhlmann$}. Read \cref{sec:protocols,sec:structural-uhlmann,sec:structural-succinct-uhlmann}.

    \item \textbf{Quantum cryptography}. Read \cref{sec:qcrypto}. It may be helpful to review the definitions of quantum interactive protocols (\cref{sec:prelims,sec:protocols}).

    \item \textbf{Quantum Shannon theory}. Read \cref{sec:shannon}. It may be helpful to read the section on quantum commitments (\cref{sec:commitments}).

    \item \textbf{Quantum gravity}. Read \cref{sec:gravity}. It may be helpful to read the section on the Decodable Channel Problem (\cref{ssec:Noisy_Channel_Decoding}).
\end{itemize}

\subsection*{Acknowledgments} 

We thank Anurag Anshu, Lijie Chen, Andrea Coladangelo, Sam Gunn, Yunchao Liu, Joe Renes, Renato Renner, and Mehrdad Tahmasbi for helpful discussions.
We thank Fred Dupuis for his help with understanding the decoupling results in his thesis.
We are especially grateful to William Kretschmer, Fermi Ma, and John Wright for their thorough discussions on the first draft of the paper.
We also thank Tomoyuki Morimae for their helpful comments on a preliminary version of this work. We thank anonymous conference reviewers for their helpful feedback. 
JB and HY are supported by AFOSR award FA9550-21-1-0040, NSF CAREER award CCF-2144219, and the Sloan Foundation. 
TM acknowledges support from SNSF Project Grant No. 200021\_188541 and AFOSR-Grant No. FA9550-19-1-0202.
AP is partially supported by AFOSR YIP (award number FA9550-16-1-0495), the Institute for Quantum Information and Matter (an NSF Physics Frontiers Center; NSF Grant PHY-1733907) and by a grant from the Simons Foundation (828076, TV).
LQ is supported by DARPA under Agreement No.\ HR00112020023.
We thank the Simons Institute for the Theory of Computing, where some of this work was conducted.  \section{Preliminaries}
\label{sec:prelims}

\subsection{Notation} \label{sec:notation}
For a bit string $x \in \bits^*$, we denote by $|x|$ its length (not its Hamming weight).
When $x$ describes an instance of a computational problem, we will often use $n = |x|$ to denote its size.

A function $\delta:\N \to [0,1]$ is an \emph{inverse polynomial} if there exists a polynomial $p$ such that $\delta(n) \leq 1/p(n)$ for all sufficiently large $n$. A function $\epsilon:\N \to [0,1]$ is \emph{negligible} if for every polynomial $p$, for all sufficiently large $n$ we have $\epsilon(n) \leq 1/p(n)$. Furthermore for convenience we also assume (unless otherwise stated) all polynomials and error functions are \emph{monotonic}, i.e., for polynomials $p$ we assume that $p(n+1) \geq p(n)$ for all $n$ and for error functions $\eps:\N \to [0,1]$ we have $\eps(n+1) \leq \eps(n)$. When we talk about polynomial or negligible functions with multiple arguments (e.g., $\poly(n,r)$ or $\negl(n,r)$), we mean that it is a polynomial or negligible function in the \emph{sum} of the two arguments (i.e., $\poly(n,r) = \poly(n+r)$ and $\negl(n,r) = \negl(n+r)$). 

A \emph{register} $\reg{R}$ is a named finite-dimensional complex Hilbert space. If $\reg{A}, \reg{B}, \reg{C}$ are registers, for example, then the concatenation $\reg{A} \reg{B} \reg{C}$ denotes the tensor product of the associated Hilbert spaces. We abbreviate the tensor product state $\ket{0}^{\ot n}$ as $\ket{0^n}$. For a linear transformation $L$ and register $\reg R$, we write $L_{\reg R}$ to indicate that $L$ acts on $\reg R$, and similarly we write $\rho_{\reg R}$ to indicate that a state $\rho$ is in the register $\reg R$. We write $\Tr(\cdot)$ to denote trace, and $\Tr_{\reg R}(\cdot)$ to denote the partial trace over a register $\reg R$.

We denote the set of linear transformations on $\reg R$ by $\linear(\reg R)$, and linear transformations from $\reg R$ to another register $\reg S$ by $\linear(\reg R, \reg S)$. We denote the set of positive semidefinite operators on a register $\reg{R}$ by $\pos(\reg{R})$. 
The set of density matrices on $\reg R$ is denoted $\states(\reg R)$.
For a pure state $\ket\varphi$, we write $\varphi$ to denote the density matrix $\ketbra{\varphi}{\varphi}$. We denote the identity transformation by $\id$.
For an operator $X \in \linear(R)$, we define $\| X \|_\infty$ to be its operator norm, and $\| X\|_1 = \Tr(|X|)$ to denote its trace norm, where $|X| = \sqrt{X^\dagger X}$.
We write $\td(\rho,\sigma) = \frac{1}{2} \| \rho - \sigma \|_1$ to denote the trace distance between two density matrices $\rho, \sigma$, and $\fidelity(\rho,\sigma) = \| \sqrt{\rho} \sqrt{\sigma} \|_1^2$ for the fidelity between $\rho, \sigma$.\footnote{We note that in the literature there are two versions of fidelity that are commonly used; here we use the \emph{squared} version of it.} Throughout the paper we frequently invoke the following relationship between fidelity and trace distance:

\begin{proposition}[Fuchs-van de Graaf inequalities]
For all density matrices $\rho,\sigma$ acting on the same space, we have that
\[
    1 - \sqrt{\fidelity(\rho,\sigma)} \leq \td(\rho,\sigma) \leq \sqrt{1 - \fidelity(\rho,\sigma)}\,.
\]
\end{proposition}

A \emph{quantum channel} from register $\reg A$ to $\reg B$ is a completely positive and trace-preserving (CPTP) map from $\linear(\reg A)$ to $\linear(\reg B)$.
For simplicity, we often write $\cN: \reg A \to \reg B$ instead of $\cN: \linear(\reg A) \to \linear(\reg B)$ when it is clear that $\cN$ is a channel. 
We denote the set of quantum channels as $\cptp(\reg A, \reg B)$.
We also call a channel a \emph{superoperator}.
For a channel $\Phi$, we write $\supp(\Phi)$ to denote the number of qubits it takes as input. 
We call a channel unitary (resp.~isometric) if it conjugates its input state with a unitary (resp.~isometry).
The diamond norm of a channel $\Phi \in \cptp(\reg A,\reg B)$ is defined as $\| \Phi \|_\diamond = \max_\rho \| (\Phi \otimes \id_{\reg C})(\rho) \|_1$ where the maximization is over all density matrices $\rho \in \states(\reg A \otimes \reg C)$ where $\reg C$ is an arbitrary register.  

Another important type of quantum operation is a $\emph{measurement}$. In general a quantum measurement is described by a finite set of positive semidefinite matrices $\mathcal{M} = \{M_i\}_{i}$ satisfying $\sum_{i} M_i = \id$. Performing a measurement on a state $\rho$ results in outcome $i$ with probability $\Tr[M_i \rho]$. Conditioned on outcome $i$, the post-measurement state is 
\begin{equation}
    \rho|_{M_i} = \frac{\sqrt{M_i}\rho \sqrt{M_i}}{\Tr(M_i \rho)}\,.
\end{equation}
The gentle measurement lemma is an important property about quantum measurements that connects the trace distance between a state and its post-measurement state to the probability that the measurement accepts.
\begin{proposition}[Gentle Measurement lemma]\label{prop:gentle_measurement}
    Let $\rho$ be a density matrix and $\Lambda$ be a positive semidefinite Hermitian matrix. If $\Tr[\Lambda \rho] \geq 1 - \epsilon$, then $\fidelity(\rho,\rho|_{\Lambda}) \geq 1 - \epsilon$ and $\norm{\rho - \rho|_{\Lambda}}_{1} \leq 2 \sqrt{\epsilon}$.  
\end{proposition}
\noindent A proof of this can be found in, e.g.,~\cite[Lemma 9.4.1]{wilde2013quantum}.

\subsection{Partial isometries and channel completions}

Usually, operations on a quantum state can be described by a unitary matrix, an isometry (if new qubits are introduced), or more generally a quantum channel (if one allows incoherent operations such as measuring or discarding qubits).
However, we will find it useful to consider operations whose action is only defined on a certain subspace. Outside of this ``allowed subspace'' of input states, we do not want to make a statement about how the operation changes a quantum state.
Such operations can be described by partial isometries.

\begin{definition}[Partial isometry] \label{def:partial_iso}
A linear map $U \in \linear(\reg A, \reg B)$ is called a partial isometry if there exists a projector $\Pi \in \linear(\reg A)$ and an isometry $\tilde U \in \linear(\reg A, \reg B)$ such that $U = \tilde U \Pi$.
We call the image of the projector $\Pi$ the \emph{support} of the partial isometry $U$.
\end{definition}

Of course in practice we cannot implement a partial isometry because it is not a trace-preserving operation, as states in the orthogonal complement of the support are mapped to the 0-vector.
We therefore define a \emph{channel completion} of a partial isometry as any quantum channel that behaves like the partial isometry on its support, and can behave arbitrarily on the orthogonal complement of the support.

\begin{definition}[Channel completion] \label{def:channel_completion}
Let $U \in \linear(\reg A, \reg B)$ be a partial isometry. 
A \emph{channel completion of $U$} is a quantum channel $\Phi \in \cptp(\reg A, \reg B)$ such that for any input state $\rho \in \states(\reg A)$, 
\begin{align*}
    \Phi(\Pi \rho \Pi) = U \rho  U^\dagger\,,
\end{align*}
where $\Pi \in \linear(\reg A)$ is the projector onto the support of $U$.
If $\Phi$ is a unitary or isometric channel, we also call $\Phi$ a \emph{unitary} or \emph{isometric completion} of the partial isometry.

An \emph{$\epsilon$-error channel completion of $U$} is a quantum channel $\tilde{\Phi}$ that is $\epsilon$-close in diamond norm to a channel completion $\Phi$ of $U$. 
\end{definition}

\subsection{Quantum circuits} \label{sec:qcircuits}
For convenience we assume that all quantum circuits use gates from the universal gate set $\{ H, \mathit{CNOT}, T \}$~\cite[Chapter 4]{nielsen2000quantum} (although our results hold for any universal gate set consisting of gates with algebraic entries). A \emph{unitary quantum circuit} is one that consists only of gates from this gate set. A \emph{general quantum circuit} is a quantum circuit that can additionally have non-unitary gates that (a) introduce new qubits initialized in the zero state, (b) trace them out, or (c) measure them in the standard basis. We say that a general quantum circuit has size $s$ if the total number of gates is at most $s$. We say that a general quantum circuit uses space $s$ if the number of qubits involved at every time step of the computation is at most $s$. The description of a general quantum circuit is a sequence of gates (unitary or non-unitary) along with a specification of which qubits they act on.
A general quantum circuit $C$ implements a quantum channel; we will abuse notation slightly and also use $C$ to denote the channel.
For a unitary quantum circuit $C$ we will write $\ket{C}$ to denote the state $C \ket{0 \dots 0}$.

\begin{definition}[Polynomial-size and polynomial-space circuit families] \label{def:poly_circuits}
A family of general quantum circuits $(C_{x, r})_{x \in \{0,1\}^*, r \in \N}$ has \emph{polynomial-size} (resp.~\emph{polynomial-space}) if there exists a polynomial $p$ such that $C_{x, r}$ has size (resp.~uses space) at most $p(|x|, r)$. 
\end{definition}

The reason circuit families are indexed by a pair $(x,r)$ is because in our unitary complexity theory framework, algorithms get both an \emph{instance} $x \in \{0,1\}^n$, which is a string that specifies \emph{which} state transformation to implement, and a \emph{precision parameter} $r \in \N$, which specifies \emph{how accurately} to implement the state transformation. The complexity of the circuit family is measured as a function of the length $|x|$ and the precision parameter $r$. 
The formal sense in which $r$ is a precision parameter will become clear when we define the notion of unitary synthesis problems in \Cref{sec:defs}.

\begin{remark}
    Sometimes we work with circuit families $\{C_x\}_{x \in \bits^*}$ that are \emph{not} indexed by a precision parameter $r$ (e.g., the precision may be fixed). In that case, we say that the family has polynomial size if $C_x$ has size $\poly(|x|)$. 
\end{remark}

\begin{definition}[Uniform circuit families] \label{def:uniform_ckt}
A family of general polynomial-size (resp.~polynomial-space) quantum circuits $(C_{x, r})_{x \in \bits^*, r \in \N}$ is called \emph{time-uniform} (resp.~\emph{space-uniform}) if there exists a classical polynomial-time (resp.~polynomial-space) Turing machine that on input $(x, 1^r)$ outputs the description of $C_{x,r}$,\footnote{We adopt the convention that the output tape of a Turing machine is write-only, and does not contribute to the space usage of the Turing machine. Thus a polynomial-space Turing machine can output an exponentially-long string.} where $1^r$ means $r$ written in unary. 
For brevity, we use \emph{uniform} to mean \emph{time-uniform}.
We call a time-uniform (resp.~space-uniform) family of quantum circuits a \emph{polynomial-time quantum algorithm} (resp.~\emph{polynomial-space quantum algorithm}).
\end{definition}

Finally, it is occasionally useful to defer measurements in a circuit and consider its unitary purification:
\begin{definition}[Unitary purification of a general quantum circuit]
\label{def:unitary-purification}
A \emph{unitary purification} (or \emph{dilation}) of a general quantum circuit $C$ is a unitary circuit $\tilde{C}$ formed by performing all measurements in $C$ coherently (with the help of additional ancillas) and not tracing out any qubits. 
\end{definition}

The following proposition relates a general quantum circuit to its unitary purification; it follows directly from the definition of the unitary purification. This proposition also demonstrates that the unitary purification $\tilde{C}$ of a general quantum circuit $C$ is a specific \emph{Stinespring dilation} of the quantum channel corresponding to $C$.
\begin{proposition}
    Let $C$ be a size-$m$ general quantum circuit acting on $n$ qubits, and let $\tilde{C}$ be a unitary purification. Let register $\reg{R}$ denote all the qubits that are traced out in the original circuit $C$ as well as the ancilla qubits introduced for the purification. Then for all states $\rho$, 
    \[
        C(\rho) = \Tr_{\reg{R}}(\tilde{C} \, \rho \tilde{C}^\dagger)\,.
    \]
    Furthermore, $\tilde{C}$ acts on at most $n + m$ qubits and has size at most $m$. 
\end{proposition}

\begin{remark}
    Throughout this paper, whenever we refer to a family of quantum circuits $(C_{x,r})_{x,r}$, we mean a \emph{general} family (i.e., non-unitary operations are allowed). 
\end{remark}

\subsection{Quantum interactive protocols}
\label{sec:interactive-protocols}
We present the model of quantum interactive protocols. (For a more in-depth account we refer the reader to the survey of Vidick and Watrous~\cite{vidick2016quantum}.) Since in quantum computing the standard model of computation is the quantum circuit model (rather than quantum Turing machines), we model the verifier in a quantum interactive protocol as a sequence of \emph{verifier circuits}, one for each input length. A verifier circuit is itself a tuple of quantum circuits that correspond to the operations performed by the verifier in each round of the protocol.

More formally, a \emph{$k$-round quantum verifier circuit} $C = (C_j)_{j \in [k]}$ is a tuple of general quantum circuits that each act on a pair of registers $(\reg{V},\reg{M})$. The circuit $C_j$ should be thought of as the verifier's actions in the $j$-th round of the protocol. The register $\reg{V}$ is further divided into disjoint sub-registers $(\reg{V}_\work, \reg{V}_\flag, \reg{V}_\out)$. The register $\reg V_\work$ is the verifier circuit's ``workspace'', the register $\reg V_\flag$ is a single qubit indicating whether the verifier accepts or rejects, and the register $\reg V_\out$ holds the verifier's output (if applicable). The register $\reg{M}$ is the message register. The size of a verifier circuit $C$ is the sum of the circuit sizes of the $C_j$'s.

A \emph{quantum prover} $P$ for a verifier circuit $C$ is a unitary that acts on $\reg{M}$ as well as a disjoint register $\reg{P}$.
Note that we could also define the prover to be a collection of unitaries, one for each round, in analogy to the verifier; the two definitions are equivalent since we can always combine the single-round unitaries into a larger unitary that keeps track of which round is being executed and applies the corresponding single-round unitary.
Since we will rarely deal with prover unitaries for individual rounds, we will find it more convenient to just treat the prover as one large unitary.
Furthermore, since the prover register is of unbounded size, we can assume without loss of generality that the prover applies a unitary (rather than a quantum channel).

Let $\ket{\psi}$ denote a quantum state whose size is at most the number of qubits in $\reg{V}_\work$. We write $C(\ket{\psi}) \interact P$ to denote the interaction between the verifier circuit $C$ and the prover $P$ on input $\ket{\psi}$, which is defined according to the following process. The initial state of the system is $\ket{\phi_0} = \ket{\psi,\zs}_{\reg V_\work} \ket\zs_{\reg{V}_\flag \reg{V}_\out \reg{M} \reg{P}}$. Inductively define $\ket{\phi_i} = P \ket{\phi_{i-1}}$ for odd $i \leq 2k$, and $\ket{\phi_i} = C_{i/2} \ket{\phi_{i-1}}$ for even $i \leq 2k$. We say that $C(\ket{\psi}) \interact P$ accepts (resp.~rejects) if measuring the register $\reg{V}_\flag$ in the state $\ket{\phi_{2k}}$ in the standard basis yields the outcome $1$ (resp.~$0$). We say that the \emph{output of $C(\ket{\psi}) \interact P$ conditioned on accepting} is the density matrix
\[
\frac{\Tr_{\reg{V} \reg{M} \reg{P} \setminus \reg{V}_\out} \left( \ketbra{1}{1}_{\reg{V}_\flag} \cdot \phi_{2k} \right)}{\Tr \left( \ketbra{1}{1}_{\reg{V}_\flag} \cdot \phi_{2k} \right)}\,;
\]
in other words, it is the reduced density matrix of $\ket{\phi_{2k}}$ on register $\reg{V}_\out$, conditioned on $C(\ket{\psi}) \interact P$ accepting. (If the probability of accepting is $0$, then we leave the output undefined.) 

A \emph{quantum verifier} $V = (V_{x,r})_{x \in \bits^*,r \in \N}$ is a time-uniform sequence of polynomial-size and polynomial-round quantum verifier circuits.
As in \cref{def:poly_circuits}, the index $x$ should be thought of as the instance, i.e., which unitary the verifier is trying to implement using the interactive protocol, and $1/r$ is the precision parameter, i.e., to within what error the verifier implements the desired unitary; see e.g.~\cref{def:unitaryQIP} for an example of how these parameters are used formally.

\subsection{Quantum state complexity classes}
\label{sec:prelims-state-complexity}

Here we present the definitions of some state complexity classes that were introduced in~\cite{rosenthal2022interactive}. Intuitively, they are classes of sequences of quantum states that require certain resources to be synthesized (e.g., polynomial time or space).

\begin{definition}[$\stateBQP$, $\statePSPACE$] \label{def:stateclasses}
	The class $\stateBQP$ (resp.~$\statePSPACE$) is the set of all sequences of density matrices $(\rho_x)_{x \in \bits^*}$ such that there exists a time-uniform polynomial-size (resp.~space-uniform polynomial-space) family of general quantum circuits $(C_{x,r})_{x \in \bits^*,r \in \N}$, where $C_{x,r}$ takes no inputs, and for every $x \in \bits^*$ and $r \in \N$, the output $\sigma_{x,r}$ of $C_{x,r}$ satisfies 
	\[
	\td(\sigma_{x,r}, \rho_x) \leq \frac{1}{r}\,.
	\]
\end{definition}

We also present the class of states that can be synthesized by an efficient verifier interacting with an all-powerful prover:

\begin{definition}[{$\stateQIP$}]
	\label{def:stateQIP}
        Let $c,s:\N \times \N \to [0,1]$ be functions. The class $\stateQIP_{c,s}$ is the set of state sequences $(\rho_x)_{x \in \bits^*}$ where there exists a time-uniform polynomial-time quantum verifier $V = (V_{x,r})_{x \in \bits^*,r\in \N}$ such that for all $x \in \bits^*$ and $r \in \N$:
	\begin{itemize}
		\item \emph{Completeness:} There exists a quantum prover $P^*$ (called an \emph{honest prover}) such that
		\begin{equation*}
			\pr {V_{x,r} \interact P^* \text{ accepts}} \geq c(|x|,r)~. \end{equation*}
\item \emph{Soundness:} For all quantum provers $P$, it holds that
		\[
		   \text{if } \quad \pr {V_{x,r} \interact P \text{ accepts}} \geq s(|x|,r) \qquad \text{then} \qquad \td(\sigma, \rho_x) \leq \frac{1}{r}\,,
		\]
		where $\sigma$ denotes the output of $V_{x,r} \interact P$ conditioned on accepting.
	\end{itemize}
	Here the probabilities are over the randomness of the interaction, assuming the verifier starts with the all-zero input state. 

 Finally, define 
    \[
        \stateQIP = \bigcup_{\epsilon(n) \, \negl} \stateQIP_{1-\epsilon,\frac{1}{2}}
    \]
    where the union is over all negligible functions $\epsilon(n,r)$. The choice of $s = \frac{1}{2}$ is without loss of generality;~\cite[Lemma 4.4]{rosenthal2022interactive} shows that the soundness parameter can made exponentially small. 
\end{definition}

The main results of~\cite{rosenthal2022interactive,metger2023stateqip} combined imply the following state synthesis analogue of the famous $\mathsf{QIP} = \mathsf{IP} = \mathsf{PSPACE}$ theorem in traditional complexity theory~\cite{shamir1992ip,lund1992algebraic,jain2011qip}:

\begin{theorem}[\cite{rosenthal2022interactive,metger2023stateqip}]
\label{thm:stateqip-statepspace}
    $\stateQIP = \statePSPACE$.
\end{theorem}

\noindent This theorem will be used as a ``subroutine'' in our characterization of $\avgUnitaryQIP$ and $\avgUnitaryPSPACE$ in \Cref{sec:structural-succinct-uhlmann}.

\begin{remark} \label{rem:state_complexity_defs}
We discuss two refinements in our definitions of state complexity classes compared to~\cite{rosenthal2022interactive,metger2023stateqip}.
These modifications do not impact \cref{thm:stateqip-statepspace} and readers unfamiliar with~\cite{rosenthal2022interactive,metger2023stateqip} may safely skip this discussion.

Firstly, the state sequences in~\cite{rosenthal2022interactive,metger2023stateqip} are indexed by natural numbers $n \in \N$, rather than strings $x \in \bits^*$. The results in~\cite{rosenthal2022interactive,metger2023stateqip} can be easily adapted to hold for state sequences indexed by strings; complexity measures are then functions of the length of $x$. Consider the containment $\statePSPACE \subseteq \stateQIP$; instead of the $\stateQIP$ verifier receiving $n$ as input to specify the state $\rho_n$ to synthesize as in~\cite{rosenthal2022interactive}, the verifier now receives $x$ to specify the state $\rho_x$. For the reverse containment $\stateQIP \subseteq \statePSPACE$, instead of the $\statePSPACE$ algorithm receiving $n$ as input to specify $\rho_n$, it instead receives $x$ as input to specify the state $\rho_x$.

Secondly, the state complexity classes in~\cite{rosenthal2022interactive,metger2023stateqip} are all parameterized by an error function $\delta$; for example $\stateBQP_\delta$ was defined, where $\delta(n)$ is a function describing the closeness of the output to the target state. Without the $\delta$ subscript, the class $\stateBQP$ is defined to be the intersection over $\stateBQP_\delta$ for all inverse polynomials $\delta(n) = 1/p(n)$. (The classes $\stateQIP,\statePSPACE$ in~\cite{rosenthal2022interactive,metger2023stateqip} are also defined similarly with respect to some error parameter $\delta$). 

In this paper we define state (and unitary) complexity classes slightly differently when it comes to the error parameterization. We now insist that an algorithm or verifier that synthesizes a state family $(\rho_x)_x$ must, in addition to a string $x \in \bits^*$, also take in a parameter $r \in \N$ as input, which controls how close the output is to $\rho_x$. That is, the algorithm/verifier must be able to approximate the state family $(\rho_x)_x$ arbitrarily well, potentially at the cost of more running time (or more space). For $\stateBQP$ (resp. $\statePSPACE$) specifically, the running time (resp. space usage) of the algorithm must scale as $\poly(|x|,\frac{1}{\epsilon})$ where $\epsilon$ is the desired approximation error. 

The definitions given in~\cite{rosenthal2022interactive,metger2023stateqip} technically allow a different algorithm for every error function $\delta$, and only require algorithms for inverse polynomial error (but not necessarily for inverse exponential error). In these revised definitions we insist on a \emph{single uniform algorithm} that can handle all error ranges, which is arguably more natural, as we discuss below \cref{def:worst_case_error}.
\end{remark} \newpage
\part{Unitary Complexity Theory}
\label{part:unitary_complexity_theory}

\section{Unitary Synthesis Problems and Unitary Complexity Classes}
\label{sec:defs}

To be able to make formal statements about the complexity of quantum tasks, we present a framework for unitary complexity theory: we define unitary synthesis problems, algorithms for implementing them, unitary complexity classes, and reductions between unitary synthesis problems.

Defining these concepts for unitaries is more subtle than in the classical world. First, a unitary synthesis problem requires both a classical \emph{instance} specifying which unitary we are interested in and a quantum input to which that unitary is supposed to be applied. Secondly, we need to allow approximation errors when implementing unitaries. Finally, we will want to consider unitary versions of promise problems, which are most naturally formalized by partial isometries; these are not physical operations, so we will frequently talk about all physical instantiations (formally, channel completions) of these partial isometries.

The definitions presented in this section should be viewed as our current best understanding of conceptually clean yet practically useful definitions, but new results, e.g., on error amplification, may provide reasons to adapt these definitions in the future. In \cref{sec:definition_essay} we provide additional discussion of the choices we made in our definitions.

\subsection{Unitary synthesis problems} \label{sec:unitary_synth_problems_def}

In traditional complexity theory, decision problems are formalized as \emph{languages}, which are sets of binary strings. The analogue in our framework is the following formalization of unitary synthesis problems.

\begin{definition}[Unitary synthesis problem] \label{def:unitary_synth_problem}
A \emph{unitary synthesis problem} is a sequence $\mathscr{U} = (U_{x})_{x \in \{0,1\}^*}$ of partial isometries.\footnote{We note that while unitary synthesis problems are not necessarily sequences of unitaries, we believe that it is a better name than ``partial isometry synthesis problem''.} 
\end{definition}

We note that \cref{def:unitary_synth_problem} considers partial isometries, not only unitaries (which are of course the special case of partial isometries for which the projector in \cref{def:partial_iso} is $\Pi = \id$).
A partial isometry is essentially a unitary defined on some subspace. This is analogous to the idea of a ``promise'' on the inputs in traditional complexity theory: the partial isometry only specifies the state transformation task on input states coming from a ``promised subspace''. For inputs with support on the orthogonal complement to this promised subspace, the behavior of the state transformation is not specified. 

One should think of a unitary synthesis problem $\mathscr{U}$ as specifying, for each string $x \in \bits^*$ (called the \emph{instance}), a state transformation task $U_x$ to be performed on some quantum system (called the \emph{quantum input}). How the instance $x$ specifies such a task varies from problem to problem; we give some examples of unitary synthesis problems below. 

\begin{enumerate}
    \item (\emph{Hamiltonian time evolution}) Consider some natural string encoding of pairs $(H,t)$ where $H$ is a local Hamiltonian and $t$ is a real number: the encoding will specify the number of qubits that $H$ acts on as well as each local term of $H$. If $x$ is a valid encoding of such a pair $(H,t)$, then define $U_x = e^{-iHt}$. Otherwise, define $U_x = 0$ (i.e., to be the zero map, which is a partial isometry). Then we define $\textsc{TimeEvolution} = (U_x)_{x \in \{0,1\}^*}$.

    \item (\emph{Decision languages}) Let $L \subseteq \{0,1\}^*$ be a decision language. Define $\textsc{UnitaryDecider}_L = (U_x)_{x \in \{0,1\}^*}$ as follows: if $x = 1^n$ (i.e.~$x$ is the unary representation of an integer $n \in \N$), then the unitary $U_{1^n} = U_x$ acts on $n+1$ qubits and for all $y \in \{0,1\}^n,b \in \{0,1\}$, we define $U_{1^n} \ket{y}\ket{b} = \ket{y} \ket{b \oplus L(y)}$ where $L(y) = 1$ iff $y \in L$. If $x$ is not a unary encoding of an integer, then we define $U_x = 0$.

    \item (\emph{Ground state preparation}) Consider a natural string encoding local Hamiltonians. If $x$ is an encoding of a local Hamiltonian with a unique ground state $\ket{\psi_x}$, then let $U_x = \ketbra{\psi_x}{0 \cdots 0}$; otherwise let $U_x = 0$. The unitary synthesis problem $\mathscr{U} = (U_x)_{x \in \bits^*}$ corresponds to the task of preparing (starting with the all zeroes state) the unique ground state of the Hamiltonian encoded by $x$; the ``promised subspace'' here consists only of the all zeroes state.
    
\end{enumerate}

We now define what it means to \emph{implement} a unitary synthesis problem.

\begin{definition}[Worst-case implementation of unitary synthesis problems] \label{def:worst_case_error}
Let $\usynth{U} = (U_{x})_{x \in \bits^*}$ denote a unitary synthesis problem. Let $C = (C_{x,r})_{x \in \bits^*,r \in \N}$ denote a family of quantum circuits. We say that \emph{$C$ is a worst-case implementation of $\usynth{U}$} if for all $x \in \{0,1\}^*$, for all $r \in \N$, there exists a channel completion $\Phi_{x,r}$ of $U_x$ such that
\[
    \Big \| C_{x,r} - \Phi_{x,r} \Big \|_\diamond \leq \frac{1}{r}\,,
\]
where $\| \cdot \|_\diamond$ denotes the diamond norm\footnote{
Recall that a small diamond distance between two channels means that the channels are difficult to distinguish even if the channels are applied to an entangled state.}.
\end{definition}
This definition also clarifies the sense in which we can think of the subscript $r$ in $C_{x,r}$ as a desired approximation error.

Intuitively, an implementation of a unitary synthesis problem is just a sequence of general quantum circuits that approximately implement the corresponding partial isometries. 
We call \cref{def:worst_case_error} a \emph{worst-case implementation} because $C_{x,r}$ has to approximate the partial isometry (more precisely, any channel completion) in diamond distance, which means that $C_{x,r}$ and $\Phi_{x,r}$ have to behave the same \emph{on all quantum inputs}.
We note two subtleties:
\begin{enumerate}
    \item We require that the circuit $C_{x,r}$ is close to \emph{some} channel completion $\Phi_{x,r}$ of $U_x$ because $C_{x,r}$ is only required to behave like $U_x$ on the support of $U_x$, and can behave arbitrarily on states orthogonal to the support (in a way that may depend arbitrarily on $x$ and $r$). This is analogous to classical promise problems, where a Turing machine deciding the promise problem is allowed to behave arbitrarily on inputs violating the promise, i.e., the Turing machine is only required to implement \emph{some} function that agrees with the promise problem on inputs fulfilling the promise.

    \item This definition provides, for every instance $x \in \bits^*$ and $\eps > 0$, a circuit $C_{x,r}$ that implements $U_x$ with error $\eps$, for $r = O(1/\eps)$. Thus an algorithm for implementing a unitary synthesis problem has to be able to achieve arbitrarily small error. This captures the notion that the unitary synthesis problem $(U_x)_x$ must be \emph{algorithmically implementable}.

\end{enumerate}

\subsection{Distributional unitary synthesis problems}

We also define a notion of \emph{distributional (or average-case) unitary synthesis problems}.
Here, in addition to a partial isometry, we also specify a state and a register of this state on which the partial isometry is going to act; note, however, that this is very different from a state synthesis problem, as we discuss in \cref{rem:avg_vs_state}.
We first give the formal definition and then explain why this is a reasonable notion of a distributional unitary synthesis problem.

\begin{definition}[Distributional unitary synthesis problem]
\label{def:dist-usynth}
We say that a pair $(\usynth{U}, \Psi)$ is a \emph{distributional unitary synthesis problem} if $\usynth{U} = (U_x)_x$ is a unitary synthesis problem with $U_x \in \linear(\reg A_x, \reg B_x)$ for some registers $\reg A_x, \reg B_x$, and $\Psi = (\ket{\psi_x})_x$ is a family of bipartite pure states on registers $\reg A_x \reg R_x$. We call $\ket{\psi_x}$ the \emph{distribution state} with \emph{target register} $\reg{A}_x$ and \emph{ancilla register} $\reg{R}_x$. 
\end{definition}

\begin{definition}[Average-case implementation of distributional unitary synthesis problems] \label{def:avg_case_error}
Let $(\usynth{U},\Psi)$ denote a distributional unitary synthesis problem, where $\usynth{U} = (U_x)_x$ and $\Psi = (\ket{\psi_x})_x$. Let $C = (C_{x, r})_{x \in \bits^{*}, r \in \N}$ denote a family of quantum circuits, where $C_{x, r}$ implements a channel whose input and output registers are the same as those of $U_x$ for all $r \in \N$. We say that \emph{$C$ is an average-case implementation of $(\usynth{U},\Psi)$} if for all $x \in \bits^*$ and $r \in \N$, there exists a channel completion $\Phi_{x,r}$ of $U_x$ such that
\[
\td \Big ( (C_{x, r} \ot \id)(\psi_x), \, (\Phi_{x,r} \ot \id)(\psi_x) \Big ) \leq \frac{1}{r} \,,
\]
where the identity channel acts on the ancilla register of $\ket{\psi_x}$.
\end{definition}

\begin{remark} \label{rem:distributional_explainer}
The term ``distributional'' may seem a bit odd at first; for example, where is the distribution in \Cref{def:dist-usynth}? In classical average-case complexity theory, a distributional problem is one where the inputs are sampled from some probability distribution $\mathcal{D}$. The state family $\Psi = (\ket{\psi_x})_x$ in a distributional unitary synthesis problem $(\usynth{U},\Psi)$ can be viewed as a \emph{purification} of a distribution over pure states: by the Schmidt decomposition, we can always write 
\begin{align}
\ket{\psi_x} = \sum_{j} \sqrt{p_{x,j}} \ket{\phi_{x,j}} \ot \ket{j} \label{eqn:schmidt_decomp}
\end{align}
for orthonormal states $\{\ket{\phi_{x,j}}\}_j$ on $\reg{A}_x$ and $\{\ket{j}\}_j$ on $\reg{R}_x$.
The Schmidt coefficients $\{p_{x,j}\}_j$ form a probability distribution $\cD_x$, so $\ket{\psi_x}$ can be viewed as the purification of the distribution $\cD_x$ over pure states $\{\ket{\phi_{x,j}}\}_j$.
The condition of $C$ implementing $(\usynth{U},\Psi)$ with average-case error implies the following: for all $x \in \bits^*$ and $r \in \N$ there exists a channel completion $\Phi_{x,r}$ of $U_x$ such that
\begin{align}
\E_{j \sim \mathcal{D}_x} \left[\td(C_{x,r}(\phi_{x,j}), \, \Phi_{x,r}(\phi_{x,j}))\right] \leq \frac{1}{r}\,. \label{eqn:dist_cond}
\end{align}
This can be seen by applying a measurement in the Schmidt basis on the purifying register, which does not increase the trace distance in \Cref{def:avg_case_error}. 
Looking ahead, we will find it more convenient to simply specify (for each $x$) one pure state $\ket{\psi_x}_{\reg{A}_x \reg{R}_x}$ instead of a set of pure states on $\reg{A}_x$ and a distribution over them.
\end{remark}

\begin{remark}
Comparing \cref{def:worst_case_error} and \cref{def:avg_case_error}, we see that we can also define worst-case implementations in terms of average-case implementations:
a circuit family $C = (C_{x,r})_{x \in \bits^*,r \in \N}$ is a worst-case implementation for a unitary synthesis problem $\usynth{U} = (U_{x})_{x \in \bits^*}$ if and only if it is an average-case implementation of the distributional unitary synthesis problem $(\usynth{U},\Psi)$ for all state sequences $\Psi = (\ket{\psi_x})_x$. This is because the diamond norm distance between two channels is equivalent to the maximum trace distance over all inputs (which can be, without loss of generality, a pure state if we include the purifying system). 
\end{remark}

\begin{remark} \label{rem:avg_vs_state}
Since an average-case unitary synthesis problem specifies both an input state $\ket{\psi_x}$ and a unitary $U_x$ to be applied on that state, it may seem like this is equivalent to the state synthesis problem of preparing $U_x \ket{\psi_x}$.
This, however, is not the case: the state $\ket{\psi_x}$ can be entangled across registers $\reg{A}_x \reg{R}_x$, but the unitary $U_x$ is only allowed to act on $\reg{A}_x$. Thus even if one could prepare copies of the pure state $(U_x \ot \id) \ket{\psi_x}$ for free, one cannot necessarily efficiently transform $\ket{\psi_x}$ to $(U_x \ot \id) \ket{\psi_x}$ by acting locally on a subsystem -- there is the additional requirement that the output of the transformation task is correctly entangled with a reference register that the transformation cannot access. 
\end{remark}

\subsection{Unitary complexity classes}
\label{sec:unitary-complexity-classes}

A \emph{unitary complexity class} is a collection of unitary synthesis problems. We introduce some natural unitary complexity classes by defining the unitary synthesis analogues of $\BQP$ and $\PSPACE$, respectively.

\begin{definition}[$\unitaryBQP$, $\unitaryPSPACE$]\label{def:unitaryBQP_unitaryPSAPCE}
Define the unitary complexity class $\unitaryBQP$ (resp.~$\unitaryPSPACE$) to be the set of unitary synthesis problems $\usynth{U} = (U_x)_x$ for which there exists a time-uniform polynomial-time (resp.~polynomial-space) quantum algorithm $C = (C_{x,r})_{x,r}$ that is a worst-case implementation of $\usynth{U}$.
\end{definition}

Next we define classes of distributional unitary synthesis problems, the unitary complexity analogues of classical average case complexity classes.

\begin{definition}[$\avgUnitaryBQP$, $\avgUnitaryPSPACE$]\label{def:avgunitaryBQP_avgunitaryPSPACE}
Define the unitary complexity class $\avgUnitaryBQP$ (resp.~$\avgUnitaryPSPACE$) to be the set of distributional unitary synthesis problems $(\usynth{U},\Psi)$ where $\Psi \in \stateBQP$ (resp.~$\Psi \in \statePSPACE$) and there exists a polynomial-time (resp.~polynomial-space) quantum algorithm $C$ that implements $(\usynth{U},\Psi)$ with average-case error.
\end{definition}

In our definition of $\avgUnitaryBQP$ and $\avgUnitaryPSPACE$, we require that the state sequence with respect to which the average case unitary synthesis problem is defined be in the corresponding state complexity class (i.e.~$\stateBQP$ and $\statePSPACE$, respectively).
We will follow this general pattern throughout the paper: whenever we define an average case unitary complexity class, we will require that the state sequence is in the corresponding state class (see  e.g.~\cref{def:avgUnitaryQIP}).

\subsection{Reductions}

Notions of reductions are crucial in complexity theory and theoretical computer science. We introduce a basic notion of reduction that allows us to relate one unitary synthesis problem to another. First, we formalize the notion of circuits that can make queries to a unitary synthesis oracle. 
Intuitively, a quantum circuit with access to a unitary synthesis oracle is just like a normal quantum circuit, except that it can apply some set of partial isometries (or more precisely arbitrary channel completions of partial isometries) in a single computational step by using the unitary synthesis oracle.

\begin{definition}[Quantum query circuits]
A \emph{quantum query circuit} $C^*$ specifies a sequence of gates like those in a general quantum circuit (defined in \cref{sec:qcircuits}), except it may also include special ``oracle gates''.
An oracle gate is specified by a label $(y, s) \in \bits^* \times \N$; its action on its input qubits will be specified separately, i.e.~a quantum query circuit is not actually a quantum circuit, but rather a \emph{template} for a quantum circuit.
\end{definition}

The label $(y,s)$ of an oracle gate specifies both which instance $y$ of a unitary synthesis problem will be inserted into the oracle gate and what implementation error $1/s$ we allow for the oracle gate.

\begin{definition}[Instantiations of quantum query circuits] \label{def:instantiation}
An \emph{instantiation} of a quantum query circuit $C^* = (C^*_{x,r})_{x,r}$ with a unitary synthesis problem $\usynth U$, denoted by $C^\usynth{U}$, is a sequence of quantum channels obtained by first fixing a worst-case implementation $(D_{y,s})_{y,s}$ of $\usynth{U}$, and then replacing all the oracle gates in $C^*_{x,r}$ with labels $\ell = (y, s)$ by channels $D_{y,s}$. Whenever we write $C^\usynth{U}$, we implicitly require that $\usynth U$ is such that the input and output registers of $U_y$ match the input and output registers of any oracle gate with label $(y,s)$ in $C^*$.
\end{definition}

\begin{remark}
In \cref{def:instantiation} we leave implicit the dependence on which worst-case implementation of $U_y$ the instantiation uses.
This is we because we will always require that an oracle circuit works for all choices of channel completion: whenever we say that a statement holds for $C^{\usynth U}$, we mean that it holds for all possible choices of worst-case implementation $(D_{y,s})_{y,s}$ of $\cU$.
This is analogous to classical oracle machines that have access to promise problems: in the classical case, such an oracle machine must work no matter how the oracle behaves on inputs outside the promise.
\end{remark}

\begin{remark}
We note that our definition of quantum query circuit has the classical instances $y$ ``hardcoded'' into the description of the circuit. In particular, the query circuit cannot choose which oracles it queries depending on its quantum input.\footnote{Of course, for a family of query circuits $(C^*_{x,r})$, the labels $(y, s)$ used by $C_{x,r}^*$ can depend on the index $x$; the point here is that a given $C_{x,r}^*$ cannot compute $y$ as a function of the quantum input it is given.} To accommodate situations when the oracle circuit may want to query different oracles $\usynth{U} = (U_x)_x$ (perhaps even in superposition), one can define a ``controlled oracle'' $\tilde{U}_n = \sum_{x: |x| = n} \ketbra{x}{x} \otimes U_x$. In other words, $\tilde{U}_n$ applies the oracle $U_x$ conditioned on some $n$-qubit register being in the state $\ket{x}$. 
A quantum query circuit with access to this controlled oracle can then apply different $U_x$ coherently depending on its quantum input, i.e.~the controlled oracle gives a query circuit more power than the uncontrolled one.
\end{remark}

It will also be useful to define instantiations of quantum query circuits with distributional unitary synthesis problems. The only difference to \cref{def:instantiation} is that the channels $D_{y,s}$ are now only required to be average-case implementations of some distributional unitary synthesis problem. We include a formal definition for completeness:

\begin{definition}[Instantiations of query circuits with distributional unitary synthesis problems] \label{def:avg_instantiation}
An instantiation of a quantum query circuit $C^* = (C^*_{x,r})_{x,r}$ with an average-case unitary synthesis problem $(\usynth U, \Psi)$, denoted by $C^{(\usynth U,\Psi)}$, is a sequence of quantum channels obtained by first fixing an average-case implementation $(D_{y,s})_{y,s}$ of $(\usynth{U}, \Psi)$, and then replacing all the oracle gates in $C^*_{x,r}$ with labels $\ell = (y, s)$ by channels $D_{y,s}$. 
Whenever we write $C^{(\usynth U,\Psi)}$, we implicitly require that $(\usynth U, \Psi)$ is such that the input and output registers of $U_y$ match the input and output registers of any oracle gate with label $(y,s)$ in $C^*$.
\end{definition}

Basic notions from ``normal'' circuits, like polynomial-size, naturally extend to query circuits.
The one additional requirement is that a query circuit $C^*_{x, r}$ only uses oracle gates with labels $y, s$ that are polynomially related to $x, r$:

\begin{definition}[Polynomial-size query circuits] \label{def:poly_query_circuits}
A family $(C^*_{x, r})_{x \in \{0,1\}^*, r \in \N}$ of quantum query circuits has \emph{polynomial-size} if there exists a polynomial $p$ such that $C^*_{x, r}$ contains at most $p(|x|, r)$ gates and each oracle gate in $C^*_{x, r}$ is labelled by a tuple $(y,s)$ satisfying $|y|,s \leq p(|x|, r)$. 
\end{definition}

\begin{definition}[Time-uniform quantum query circuits]
\label{def:uniform_query_ckt}
A family $(C^*_{x,r})_{x \in \bits^*,r \in \N}$ of polynomial-size quantum query circuits is \emph{time-uniform} (or just \emph{uniform}) if there exists a classical polynomial-time Turing machine that on input $(x,1^r)$ outputs the description of $C^*_{x,r}$.
We call a time-uniform family of quantum query circuits a \emph{polynomial-time quantum query algorithm}. 
\end{definition}

Using quantum query circuits, we can define polynomial-time reductions between unitary synthesis problems.
Intuitively, ``$\mathscr{U}$ polynomial-time reduces to $\mathscr{V}$'' means that $\mathscr{U}$ can be implemented efficiently if we have access to oracle gates that (approximately) implement channel completions of instances of  $\mathscr{V}$.
\begin{definition}[Polynomial-time reductions between unitary synthesis problems] \label{def:reduction_worst_case}
Let $\mathscr{U} = (U_x)_x$ and $\mathscr{V} = (V_x)_x$ denote unitary synthesis problems. Then $\mathscr{U}$ \emph{polynomial-time reduces} to $\mathscr{V}$ if there exists a polynomial-time quantum query algorithm $C^* = (C_{x,r}^*)_{x,r}$ such that $C^{\mathscr{V}}$ is a worst-case implementation of $\usynth{U}$.
\end{definition}

This notion of reduction readily extends to distributional problems, too. The only thing that changes is whether the oracle gates are instantiated with worst-case or average-case implementations, and whether the instantiated query circuit is required to be a worst-case or an average-case implementation. We spell out the definition for completeness.

\begin{definition}[Polynomial-time reductions between distributional unitary synthesis problems] \label{def:reduction_distributional}
Let $(\mathscr{U}, \Psi)$ and $(\mathscr{V}, \Phi)$ denote distributional unitary synthesis problems. Then $(\mathscr{U}, \Psi)$ \emph{polynomial-time reduces} to $(\mathscr{V}, \Phi)$ if there exists a polynomial-time quantum query algorithm $C^* = (C_{x,r}^*)_{x,r}$ such that $C^{(\mathscr{V}, \Phi)}$ is an average-case implementation of $(\usynth{U}, \Psi)$.
\end{definition}

Note that one can also define reductions between ``normal'' and distributional unitary synthesis problems in the same way.

Just like one can define oracle complexity classes like $\mathsf{P}^{\textsc{3SAT}}$ (i.e., polynomial-time computation with oracle access to a 3SAT oracle), we can now also define oracle complexity classes for unitary synthesis problems:

\begin{definition}[Oracle unitary complexity classes]
    We define the oracle class $\unitaryBQP^{\usynth{V}}$ to be the set of all unitary synthesis problems that are polynomial-time reducible to a unitary synthesis problem $\usynth V$.
\end{definition}

We can similarly define the oracle class $\avgUnitaryBQP^{(\usynth V,\Omega)}$.
However, there is a subtlety: we will have to specify a state complexity class which the distributional states are  required to be from.
For $\avgUnitaryBQP$, we required that the distributional states be from $\stateBQP$.
However, if we give the $\avgUnitaryBQP$ oracle access to $(\usynth V,\Omega)$, it is natural to allow the same oracle access for the preparation of the distributional states, too (corresponding to the principle that the complexity of the distributional state should match the complexity of the unitary class). 
The resulting ``oracle state complexity class'' $\stateBQP^{(\usynth V,\Omega)}$ is defined completely analogously to the unitary setting by allowing oracle gates in the state preparation circuits from \cref{def:stateclasses}.
With this, we can define:

\begin{definition}[Average-case oracle unitary complexity classes]
We define the oracle class $\avgUnitaryBQP^{(\usynth V,\Omega)}$ to be the set of all distributional problems $(\usynth U, \Psi)$ that are polynomial-time reducible to the distributional unitary synthesis problem $(\usynth V,\Omega)$ and for which $\Psi \in \stateBQP^{(\usynth V,\Omega)}$.
\end{definition}

Just like for classical complexity classes, we can use this notion of reduction to define hard and complete problems for (average-case) unitary complexity classes.
\begin{definition}[Hard and complete problems]
We call a unitary synthesis problem $\usynth{U}$ \emph{hard} (under polynomial-time reductions) for a unitary complexity class $\class{unitaryC}$ if $\class{unitaryC} \subseteq \unitaryBQP^\usynth{U}$.
If additionally $\usynth{U} \in \class{unitaryC}$, we call $\usynth{U}$ \emph{complete} for the class $\class{unitaryC}$.

Analogously, we call a distributional unitary synthesis problem $(\usynth{U}, \Psi)$ \emph{hard} (under polynomial-time reductions) for an average-case unitary complexity class $\class{avgUnitaryC}$ if $\class{avgUnitaryC} \subseteq \avgUnitaryBQP^{(\usynth{U}, \Psi)}$.
If additionally $(\usynth{U}, \Psi) \in \class{avgUnitaryC}$, we call $(\usynth{U}, \Psi)$ \emph{complete} for the class $\class{avgUnitaryC}$.
\end{definition}

As would be expected, $\unitaryBQP$  and $\avgUnitaryBQP$ are closed under polynomial-time reductions.
\begin{lemma} \label{lem:unitary_bqp_closed}
$\unitaryBQP$ is closed under polynomial-time reductions, i.e.,~for all $\usynth V \in \unitaryBQP$, we have that $\unitaryBQP^\usynth{V} \subseteq \unitaryBQP$. Similarly, $\avgUnitaryBQP$ is closed under polynomial-time reductions, i.e.~for all $(\usynth V,\Omega) \in \avgUnitaryBQP$, we have that $\avgUnitaryBQP^{(\usynth{V},\Omega)} \subseteq \avgUnitaryBQP$.
\end{lemma}

This follows almost trivially from  the definitions, but we spell out the proof to illustrate how to use the definitions.

\begin{proof}
Consider a unitary synthesis problem $\usynth U = (U_x)_x \in \unitaryBQP^\usynth{V}$.
By definition, there exists a polynomial-time quantum query algorithm $C^* = (C^*_{x,r})_{x,r}$ such that $C^{\usynth{V}}$ is a worst-case implementation of $\usynth{U}$.
Since $\usynth V \in \unitaryBQP$, there exists a polynomial-time quantum algorithm $(D_{y,s})_{y,s}$ that is a worst-case implementation of $\usynth{V}$.
For every oracle gate in $C^{\usynth{V}}_{x,r}$ with some label $(y,s)$, we can insert the explicit circuit $D_{y,s}$.
This yields a polynomial-time quantum algorithm for $\usynth{U}$ because $|y|,s \leq \poly(|x|,r)$ by \cref{def:poly_query_circuits}.
\end{proof}

The same ``proof-by-plugging-in-explicit-circuits'' also shows that $\unitaryPSPACE$ and $\avgUnitaryPSPACE$ are closed under polynomial-time reductions.

\begin{lemma} \label{lem:unitary_pspace_closed}
$\unitaryPSPACE$ is closed under polynomial-time reductions, i.e.~for all $\usynth V \in \unitaryPSPACE$, we have that $\unitaryBQP^{\usynth V} \subseteq \unitaryPSPACE$. Similarly, $\avgUnitaryPSPACE$ is closed under polynomial-time reductions, i.e.~for all $(\usynth V,\Omega) \in \avgUnitaryPSPACE$, we have that $\avgUnitaryBQP^{(\usynth{V},\Omega)}\subseteq \avgUnitaryPSPACE$.
\end{lemma}

\subsection{Purification of space-bounded unitary synthesis}
\label{sec:unitary-space-properties}

By definition, unitary synthesis problems in $\unitaryPSPACE$ and $\avgUnitaryPSPACE$ are implementable by \emph{general} circuits using polynomial space. Thus, non-unitary operations such as measurements, tracing out, qubit reset, etc., are allowed. Although one can simulate such non-unitary operations using unitary ones (for example, performing measurements coherently), such direct simulations result in an expensive blow-up in space. For example, applying the principle of deferred measurement to a polynomial-space circuit with exponentially many intermediate measurements would require exponentially many ancilla qubits. 

Whether intermediate measurements and other non-unitary operations can be removed in space-bounded quantum computation has been studied by Fefferman and Remscrim~\cite{fefferman2021eliminating} and Girish and Raz~\cite{girish2021eliminating}. Both works show that, up to a small blow-up in space (but perhaps a large blow-up in time), general circuits and unitary circuits decide the same languages. Girish and Raz prove a stronger result and show that the \emph{channel} computed by a circuit with intermediate measurements can be approximated by a unitary circuit with small blow-up in space:

\begin{theorem}[Eliminating intermediate measurements~\cite{girish2021eliminating}]
\label{thm:eliminating-intermediate-measurements}
Let $C$ be a time $T$, space $S \geq \log T$ general quantum circuit whose output is $n$ qubits. Then for all $\eps > 0$ there exists a unitary circuit $D$ using $O((S + \log (1/\eps)) \log T)$ space and $T \cdot \poly(S,\log 1/\eps)$ time such that
\[
	\Big \| C - B \circ D \circ A \Big \|_{\diamond} \leq \eps
\]
Here, $C$ and $D$ are also used to denote the channels computed by the circuits, $A$ denotes the channel that appends a number of zeroes (so that the input size of $A$ matches the input size of $C$, and the output size of $A$ matches the input size of $D$), and $B$ denotes tracing out all but the first $n$ qubits. Furthermore, if the circuit description of $C$ is computable by a Turing machine $M$, then the circuit description of $D$ is computable by a Turing machine whose space and time usage is polynomial in that of $M$. 
\end{theorem}

We apply this result to show that, without loss of generality, all unitary synthesis problems in $\unitaryPSPACE$ and $\avgUnitaryPSPACE$ can be implemented by circuits that are entirely unitary except for some initial ancilla initialization, and a final trace out operation.
\begin{lemma}
    \label{lem:unitary-pspace-purification} 
   	All unitary synthesis problems $\usynth{U} \in \unitaryPSPACE$ and $(\usynth{U},\Psi) \in \avgUnitaryPSPACE$ have polynomial-space implementations that are entirely unitary except for some initial ancilla initialization, and a final trace out operation.
\end{lemma}
\begin{proof}
	We prove this for $\unitaryPSPACE$; the proof for $\avgUnitaryPSPACE$ is essentially the same. Let $\usynth{U} = (U_x)_x \in \unitaryPSPACE$. Then for all instances $x$ and precision parameter $r$, there exists a general quantum circuit $C$ on $S = \poly(n,r)$ qubits that implements $U_x$ with error $1/2r$. Furthermore, since the description of $C$ is computable by a $\poly(n,r)$-space Turing machine, it must be that the time complexity of the circuit is at most $T = 2^{\poly(n,r)}$. By \Cref{thm:eliminating-intermediate-measurements} and setting $\eps = 1/2r$, there exists a unitary circuit $D$ on $(S + \log (1/\eps)) \log T = \poly(n,r)$ qubits and $T \cdot \poly(S,\log 1/\eps) = 2^{\poly(n,r)}$ time that simulates $C$ with error $1/2r$, and therefore
	\[
		\Big \| U_x - B \circ D \circ A \Big  \|_{\diamond} \leq \frac{1}{r}~,
	\]
	as desired.
	Since the circuit $C$ is part of a space-uniform family of circuits, so is the circuit $D$ (this follows from the ``furthermore'' part of \Cref{thm:eliminating-intermediate-measurements}).
\end{proof}

\subsection{Discussion} \label{sec:definition_essay}

Compared to the classical setting, the definitions of complexity classes and reductions in the unitary world have additional subtleties.
Above, we have presented the definitions that we find most natural and useful for showing interesting reductions.
In this subsection, we explain the choices we made in our definitions.

\subsubsection{Classical vs quantum inputs} \label{sec:classical_vs_quantum_inputs}
Our definition of unitary synthesis problems distinguishes between the instance specified by a string $x \in \bits^*$ and the quantum input.
One might wonder about an alternative definition of unitary synthesis problems where the instance $x$ is ``folded'' into the quantum input; instead of having $U_x$ for every string $x$, we define a unitary $U_n$ that maps $\ket{x} \ket{\psi}$ to $\ket{x} U_x \ket{\psi}$ for all $x \in \bits^n$. That is, the unitaries/partial isometries are indexed by natural numbers indicating the length of the instance $x$. 

Our definition of unitary synthesis problem contains this definition as a special case (e.g.~we could only define $U_{n}$ for $n \in \N$ to be nontrivial, like in the definition of $\textsc{UnitaryDecider}_L$ in \cref{sec:unitary_synth_problems_def}). 
We find our more general definition helpful because it reinforces the conceptual separation between the instance $x$ (which specifies \emph{which} state transformation task we are meant to perform), and the input state to that state transformation task. In the aforementioned alternative definition, the instance $x$ is syntactically treated on equal footing with the quantum input, even though the instance string $\ket{x}$ is not meant to be transformed. 

Furthermore, the alternative definition suggests an unnatural requirement where implementing $U_n$ would mean behaving coherently on a superposition of instances, i.e.,
    \[
        \sum_x \alpha_x \ket{x} \ket{\psi_x} \mapsto \sum_x \alpha_x \ket{x} U_x \ket{\psi_x}~.
    \]
However, this is usually difficult to achieve as it requires the algorithm implementing $U_n$ to uncompute any junk that depends on $x$, which may not be efficiently performable (even though each individual map $\ket{\psi} \mapsto U_x \ket{\psi}$ may be efficiently performable). 

\subsubsection{Classical average-case complexity vs distributional unitary synthesis problems}

Distinguishing between the classical instance and the quantum input allows us to consider two different notions of average-case unitary synthesis problems.
The first, which we focus on in this paper, are distributional unitary synthesis problems. Here, the input to the unitary is fixed as part of a larger entangled state, and for every instance $x$, we require the existence of a circuit that approximates the unitary $U_x$ on this fixed entangled input state.

Another notion, which is closer to classical average-case complexity~\cite{bogdanov2006average}, are ``instance-average-case'' unitary synthesis problems: we could consider a distribution over classical instances $x$ of a unitary synthesis problem $\usynth{U} = (U_x)_x$, and e.g.~require that a polynomial-time quantum algorithm $\mathcal{C}_{x,r}$ approximates $U_x$ with high probability over the choice of $x$.
This parallels the way in which classical average-case complexity considers distributions over computational problems and demands that an algorithm can solve a randomly sampled instance with high probability. We do not explore this in depth here, but we do introduce some version of this when discussing the complexity of quantum commitment schemes (see \Cref{sec:qcrypto}).

\subsubsection{On the error dependence}

In our definition of (average-case) $\unitaryBQP$, the running time of the algorithm scales polynomially with the instance size, as well as the precision parameter $r$. Put another way, the algorithm achieves error $\eps$ by running in time $\poly(n,1/\eps)$. As mentioned earlier, the fact that the algorithm must work for all $\eps$ ensures that the given unitary synthesis problem is captured by the algorithm. 

Some may wonder why we choose the error dependence to be polynomial in $1/\eps$ rather than, say, $\log 1/\eps$. In traditional complexity theory we are accustomed to the fact that errors can be often reduced to $\eps$ by simply repeating the algorithm $O(\log 1/\eps)$ times and taking a majority vote. Being able to achieve negligible error while maintaining polynomial time complexity is a very convenient property in traditional complexity theory.

This feature is not universal, however. 
For example, the majority vote approach does not work in settings where the output is randomized and consists of many bits. Furthermore, in the unitary complexity setting it may not be possible to repeat an algorithm, because ithere is only one copy of the input. 

The $\poly(1/\eps)$-dependence in our definitions is analogous to several other areas in complexity theory:
\begin{enumerate}
    \item In the theory of approximation algorithms for $\mathsf{NP}$-optimization problems, a \emph{fully polynomial time approximation scheme (FPTAS)} is an algorithm that achieves approximation error $\eps$ in time $\poly(n,1/\eps)$. A wide variety of optimization problems are known to be $\mathsf{NP}$-hard to solve exactly but admit FPTASes, such as knapsack problems, some subset sum problems, and some restricted shortest path problems~\cite{vazirani2001approximation}. 

    \item In the theory of average-case complexity~\cite{bogdanov2006average}, a (decision) problem is efficiently solvable on average if there is an algorithm that, for all $\eps$, decides the problem with error $\eps$ and runs in time $\poly(n,1/\eps)$. 

    \item The definitions of relational complexity classes such as $\mathsf{FBPP}, \mathsf{FBQP}$ as well as the sampling classes $\mathsf{SampBPP}, \mathsf{SampBQP}$ all have error dependencies that scale as $\poly(1/\eps)$~\cite{aaronson2014equivalence,aaronson2023qubit}. 
\end{enumerate}
In these settings, the $\poly(1/\eps)$ error dependence in the definition of efficient computation is justified in two ways. First, it captures natural classes of algorithms in the setting of interest (e.g., approximation, average-case, heuristic, or sampling algorithms). This definitional choice enables one to make useful statements about the various computational phenomena being studied (see, for example, the brief discussion about this choice in~\cite{aaronson2023qubit}). 

Second, in certain cases, it is highly implausible to have better error dependence: for example, problems like knapsack have FPTASes (with running time $\poly(n,1/\eps)$), but if there was $\poly(n,\log 1/\eps)$-time approximation algorithm, then $\mathsf{P} = \mathsf{NP}$. 

As we will see, our definitions of efficient computation (with $\poly(1/\eps)$ error dependency) capture the error-dependency of most algorithms and protocols for ``fully quantum'' tasks. Just to name a few assorted examples: interactive protocols for state synthesis~\cite{rosenthal2022interactive}, tomography algorithms~\cite{o2016efficient,huang2020predicting}, Uhlmann transformations~\cite{metger2023stateqip}, and density matrix exponentiation algorithms~\cite{Lloyd_2014,Kimmel_2017} all have $\poly(1/\eps)$ dependency on accuracy. Furthermore, there is some evidence that this error dependency is necessary in the unitary complexity setting: there is a (black-box) unitary synthesis task that can be solved in time $\poly(n,1/\eps)$ but not $\poly(n,\log1/\eps)$~\cite{chia2021on}.\footnote{The task studied is in the context of post-quantum zero-knowledge proofs: given a malicious verifier's code and its auxiliary input as input, output an $\eps$-approximation of the verifier's view after its interaction with the honest prover.}

We find that it is a fascinating question of whether there are exponential error reduction techniques in state and unitary synthesis, or whether it is impossible. Finding additional evidence for it one way or another would help us understand ``fully quantum'' complexity theory better.

\begin{openproblem}
    Is exponential error reduction generically possible for state and unitary synthesis problems, or can we find additional evidence that this is impossible?
\end{openproblem}

Even if generic error reduction is impossible, one can still hope to construct exponentially precise implementations of \emph{specific} unitary synthesis problems by other means.
We leave it as an interesting open problem to come up with such exponentially precise algorithms for e.g.\ the Uhlmann Transformation Problem.

\begin{openproblem}
    Are there interesting examples of unitary synthesis problems, e.g.~the Uhlmann Transformation Problem, which can be efficiently implemented with inverse exponential error?
\end{openproblem}

\subsubsection{On errors in reductions} \label{sec:poly_space_reductions_discussion}
In traditional complexity theory, when defining an oracle class like $\mathsf{PSPACE}^{\mathsf{BQP}}$, we mean that the $\mathsf{PSPACE}$ query algorithm gets to query an oracle that decides a $\mathsf{BQP}$ language \emph{without error}. This is justified by the fact that we can assume without loss of generality that a $\mathsf{BQP}$ algorithm has exponentially small error. The lack of exponential error reduction for unitary synthesis problems, however, compels us to take additional care when defining reductions. 

In our notion of reduction, the query circuit is required to specify an error bound for each query. This captures the operational idea of a reduction: when running a reduction algorithm in ``real life'', each query is made to some actual algorithm which will makes errors. Of course, when composing polynomial-time algorithms for unitary synthesis problems, we can assume the queries are made to the \emph{ideal} unitary: the error of each query can be made an arbitrarily small inverse polynomial at the cost of a polynomial blow-up in running time, and the overall error is still only an inverse polynomial. 

On the other hand, if the query algorithm can make a superpolynomial number of queries (for example, when considering $\mathsf{unitaryPSPACE}$ query algorithms), it can matter whether we assume the query algorithm gets access to the ideal unitary or an approximate version of it. For example, one can construct examples\footnote{An example of this would be unitaries of the form $U = V^{1/\exp(n)}$ where $V$ requires exponential time and space to implement. $U$ can be implemented in $\unitaryBQP$ because it is exponentially close to the identity and because of the $\poly(n,1/\eps)$ definition of $\unitaryBQP$, but powering it exponentially many times recovers $V$.} of unitary synthesis problems $\usynth{U} \in \unitaryBQP$ such that, given the ability to make exponentially many queries to the ideal unitary, one can in polynomial-space implement a unitary that would require exponential space -- leading to the puzzling conclusion that $\unitaryPSPACE^{\usynth{U}}$ requires exponential space to compute. On the other hand, if we insist that an error bound has to be specified for each query (and the space complexity of the query algorithm scales with the desired precision), then we obtain the seemingly more reasonable conclusion that $\unitaryPSPACE = \unitaryPSPACE^\unitaryBQP$. 

Granted, these peculiarities only manifest themselves when considering reductions that make a superpolynomial number of queries. The rest of this paper only considers polynomial-time reductions, but we believe it is an interesting future direction to explore the nature of inefficient reductions in unitary complexity theory.

\subsection{Summary and open problems}

In this section, we introduced a formal framework for studying the complexity of unitary synthesis problems.
We have already seen the unitary complexity classes $\unitaryBQP$ and $\unitaryPSPACE$, as well as their average-case versions.
In the next section, we consider interactive proofs for unitary synthesis problems, which will naturally lead us to define the classes $\unitaryQIP$ and $\unitary{SZK}$.
This, however, is by no means a full list of all unitary complexity classes that might be of interest --- our aim here is to introduce the classes relevant to the Uhlmann transformation problem, not to provide a complete account.
As such, it is natural to consider the following question. 
\begin{openproblem}\label{prob:unitaryQMA}
    What are other unitary complexity classes that naturally relate to physically interesting problems? For example, is there a useful notion of $\class{unitaryQMA}$? 
\end{openproblem}

The core goal of complexity theory is to organize computational problems into computational classes and understand the relationships between these classes.
Later in this paper, we will prove some results relating unitary complexity classes to one another.
However, one would naturally conjecture that certain unitary complexity classes are in fact different, e.g.~one would expect $\unitaryBQP \neq \unitaryPSPACE$.
For decision languages, proving such separations unconditionally is significantly out of reach of current techniques.
Intriguingly, it is not clear whether this necessarily constitutes a barrier for proving similar results in the unitary setting, as it might for example be possible that $\unitaryBQP \neq \unitaryPSPACE$, but $\BQP = \PSPACE$.
Therefore, another interesting question is the following:
\begin{openproblem}
    Are there barriers from traditional complexity theory to proving unitary complexity class separations? Might it be feasible to prove $\unitaryBQP \neq \unitaryPSPACE$ unconditionally?
\end{openproblem}

Finally, as mentioned at the beginning of this section, the unitary complexity framework presented here (and in the rest of the paper) should not be treated as being set in stone; it is our best attempt at a starting point for a ``fully quantum'' complexity theory. The definitions and notions were chosen to balance both conceptual clarity as well as practical usefulness for capturing phenomena in ``fully quantum'' computational tasks. However, we anticipate the framework will evolve in tandem with our understanding of the complexity of quantum input, quantum output tasks.

 \section{Interactive Proofs for Unitary Synthesis}
\label{sec:protocols}

In this section we introduce the notion of interactive proofs for unitary synthesis problems. Intuitively, a unitary synthesis problem $\usynth{U} = (U_x)_x$ admits an interactive proof if a polynomial-time verifier, who receives an instance $x$ and a quantum register $\reg{B}$, can interact with an all-powerful but untrusted prover, and at the end -- conditioned on accepting -- implements $U_x$ on the register $\reg{B}$. This is inspired by the notion of interactive proofs for decision languages, except in addition to accepting/rejecting at the end, the verifier has to implement a unitary transformation on a given register. 

We introduce the (worst-case) unitary synthesis class $\unitaryQIP$, and give an example of a unitary synthesis problem in it that is plausibly outside $\unitaryBQP$. We then introduce the average-case interactive proof classes $\avgUnitaryQIP$ and $\avgUnitary{HVSZK}$; this latter class captures a notion of \emph{zero-knowledge} interactive unitary synthesis. As we will see in \Cref{sec:structural-uhlmann,sec:structural-succinct-uhlmann}, the complexity of such average-case interactive proof classes are deeply related to the Uhlmann Transformation Problem.

\subsection{Interactive proofs for unitary synthesis}

We present our notion of interactive protocols for unitary synthesis. (For a refresher on how quantum interactive protocols are formalized, we refer the reader to \Cref{sec:interactive-protocols}. It is also useful to compare this definition to the anologous one for state complexity in \cref{def:stateQIP}). 

\begin{definition}[{$\unitaryQIP$}]
	\label{def:unitaryQIP}
        Let $c,s:\N \times \N \to [0,1]$ be functions. The class $\unitaryQIP_{c,s}$ is the set of unitary synthesis problems $\usynth{U} = (U_x)_x$ where there exists a time-uniform polynomial-time quantum verifier $V = (V_{x, r})_{x \in \{0,1\}^*, r \in \N}$ satisfying, for all $x \in \{0,1\}^*$ and $r \in \N$:
	\begin{itemize}
		\item \emph{Completeness:} There exists a quantum prover $P^*$ (called an \emph{honest prover}) such that for all input states $\ket{\psi}$ in the support of $U_x$,
		\begin{equation*}
			\pr {V_{x, r}(\ket{\psi}) \interact P^* \text{ accepts}} \geq c(|x|, r) \,.
		\end{equation*}
		\item \emph{Soundness:} For all input states $\ket{\psi}$ (which consists of an input register $\reg{A}$ given to the verifier, and a purifying register $\reg{R}$ not accessed by the verifier or prover), and for all quantum provers $P$, there exists a channel completion $\Phi_{x,r}$ of $U_x$ such that
		\[
		   \text{if } \quad \pr { V_{x, r}(\ket{\psi}) \interact P \text{ accepts}} \geq s(|x|, r) \qquad \text{then} \qquad \td(\sigma_{x,r}, (\Phi_{x,r} \otimes \id)(\psi)) \leq \frac{1}{r}\,,
		\]
		where $\sigma_{x,r}$ denotes the output of $V_{x, r}(\ket{\psi}) \interact P$ conditioned on accepting.
	\end{itemize}
	Here the probabilities are over the randomness of the interaction. 
    \end{definition}

\begin{remark}
    The concept of interactive unitary synthesis and the class $\unitaryQIP$ was first introduced by Rosenthal and Yuen~\cite{rosenthal2022interactive}, albeit with a slightly different formulation. Here we have adapted it to be compatible with the framework of unitary synthesis problems established in \Cref{sec:defs}. 
\end{remark}

We note that in the soundness condition, the malicious prover $P$ can depend on the input state $\ket{\psi}$ of the verifier, on which the verifier wants to implement the unitary $U_x$. This poses a major challenge in designing any $\unitaryQIP$ protocol. The prover is untrusted, so it seems dangerous for the verifier to send $\ket{\psi}$ to the prover ``in the clear'', as the prover could potentially recognize the state and alter it without being detected. But then, how can the verifier enlist the prover's help to implement the desired unitary on the input state?

Rosenthal and Yuen demonstrated that there are nontrivial protocols for interactive unitary synthesis~\cite{rosenthal2022interactive}; in particular, they showed that a class of unitaries with \emph{polynomial action} -- unitary operators that act nontrivially on a subspace of polynomial dimension -- can be interactively synthesized. More formally:

\begin{definition}[Polynomial-action unitary synthesis problems~\cite{rosenthal2022interactive}]
    A unitary synthesis problem $\usynth{U} = (U_x)_x$ has \emph{polynomial action} if all $U_x$ are unitary operators, and there exists a polynomial $p(n)$ such that for all $x$, the unitary $U_x$ acts nontrivially on a subspace of dimension at most $p(|x|)$. 
\end{definition}

An example of a unitary synthesis problem with polynomial action is a family of reflections $\id - 2\proj{\psi}$ for some state $\ket{\psi}$. The reflection only acts nontrivially on a subspace of dimension $1$. 

\begin{theorem}[Interactive proofs for unitaries with polynomial action~\cite{rosenthal2022interactive}]
\label{thm:interactive-synthesis-polynomial-action}
Let $\usynth{U} \in \unitaryPSPACE$ have polynomial action. Then $\usynth{U} \in \unitaryQIP_{1,\exp(-\poly(n,r))}$.  
\end{theorem}

The high-level ideal behind \Cref{thm:interactive-synthesis-polynomial-action} is as follows: if $U$ is an $n$-qubit unitary that acts nontrivially only on a $\poly(n)$-dimensional subspace, then there exists an $n$-qubit density matrix $\rho$ and a parameter $t = \poly(n)$ such that $U = e^{-i \rho t}$. Then, the unitary $U$ can be approximately implemented via the so-called \emph{density matrix exponentiation} algorithm of~\cite{Lloyd_2014,Kimmel_2017} using only $\poly(n)$-copies of the state $\rho$. Furthermore, if $U$ can be implemented in $\unitaryPSPACE$, then copies of the corresponding density matrices $\rho$ can be efficiently synthesized via an interactive proof; this utilizes the $\statePSPACE \subseteq \stateQIP$ result of~\cite{rosenthal2022interactive}. Once these copies have been synthesized, the density matrix exponentiation algorithm is used to apply the Hamiltonian evolution $e^{-i \rho t}$ to the desired register. Note that in this protocol, the prover is just used to synthesize a (potentially) complex state, which in turn helps the verifier apply the desired unitary -- the prover never touches the input register given to the verifier.

While unitaries with polynomial action are a rather restrictive class of transformations, they can still be quite complex. For example, suppose $\ket{\psi}$ is a state that requires exponential time (but polynomial space) to synthesize. The reflection operator $\id - 2\proj{\psi}$ is unlikely to be efficiently implementable in polynomial time; finding formal evidence for this is an interesting direction that we leave for future work. 

Polynomial action unitaries provide evidence of problems in $\unitaryQIP$ that lie beyond $\unitaryBQP$, but it would be even better to find a broader class of unitary synthesis problems in $\unitaryQIP$, perhaps even a complete problem. Next, we define an \emph{average-case} version of $\unitaryQIP$ that ends up capturing a much wider class of (distributional) unitary synthesis problems (as we will see in \Cref{sec:structural-uhlmann}).

\subsection{Average-case interactive unitary synthesis}

Analogously to $\avgUnitaryBQP$, in the average-case complexity version of $\unitaryQIP$, the verifier only has to synthesize the desired unitary accurately on a specified distribution state. We define $\avgUnitaryQIP$ as follows. 

\begin{definition}[{$\avgUnitaryQIP$}]
	\label{def:avgUnitaryQIP}
        Let $c,s:\N \times \N \to [0,1]$ be functions. The class $\avgUnitaryQIP_{c,s}$ is the set of distributional unitary synthesis problems $(\usynth{U} = (U_x)_x,\Psi = (\ket{\psi_x})_x)$ such that $\Psi \in \stateQIP$ and there exists a polynomial-time quantum verifier $V = (V_{x, r})_{x \in \{0,1\}^*, r \in \N}$ satisfying, for all $x \in \{0,1\}^*$ and $r \in \N$,
	\begin{itemize}
		\item \emph{Completeness:} There exists a quantum prover $P$ (called an \emph{honest prover}) such that 
		\begin{equation*}
			\pr {V_{x, r}(\ket{\psi_x}) \interact P \text{ accepts}} \geq c(|x|, r)
		\end{equation*}
  where $V_{x, r}$ does not have access to the ancilla register of $\ket{\psi_x}$. 
		\item \emph{Soundness:} For all quantum provers $P$, there exists a channel completion $\Phi_{x,r}$ of $U_x$ such that
		\[
		   \text{if } \quad \pr {V_{x, r}(\ket{\psi_x}) \interact P \text{ accepts}} \geq s(|x|, r) \qquad \text{then} \qquad \td(\sigma_{x,r}, (\Phi_{x,r} \ot \id)(\psi_x)) \leq \frac{1}{r} \,,
		\]
		where $\sigma_{x,r}$ denotes the output of $V_{x, r}(\ket{\psi_x}) \interact P$ conditioned on accepting and $V_{x, r}$ does not have access to the ancilla register of $\ket{\psi_x}$. 
	\end{itemize}
	Here the probabilities are over the randomness of the interaction. \end{definition}

\begin{remark}
    The distributional state family $\Psi$ of a problem in $\avgUnitaryQIP$ is defined to come from $\stateQIP$ (corresponding to the principle that the complexity of the distributional state family should match the complexity of the unitary class). However, the result that $\stateQIP = \statePSPACE$~\cite{rosenthal2022interactive,metger2023stateqip} implies that the class is equivalently defined with the distributional state family coming from $\statePSPACE$. In fact, this will be the definition we work with in the future sections.  
\end{remark}

Note the difference between the soundness condition of $\avgUnitaryQIP$ and $\unitaryQIP$: in the average-case setting, the verifier $V$ receives half of an an entangled distributional state, and at the beginning of the protocol the prover $P$ is guaranteed to be unentangled with this distributional state. The prover's uncertainty about the verifier's quantum input turns out to be a very useful in designing $\avgUnitaryQIP$ protocols, as will be illustrated in \Cref{sec:structural-succinct-uhlmann} when we show that the Succinct Uhlmann Transformation Problem is complete for $\avgUnitaryQIP$.

\subsection{Zero-knowledge protocols for unitary synthesis}
\label{sec:zk}

In this section we present a notion of \emph{zero knowledge} for unitary synthesis problems. In traditional complexity theory and cryptography, a zero knowledge proof allows a polynomial-time verifier, interacting with an all-powerful prover, to decide the truth of a statement (e.g., whether a graph is $3$-colorable) without learning anything else. This is formalized via the notion of a polynomial-time simulator that can reproduce the view of the verifier without any interaction with the prover~\cite{goldwasser1989knowledge,watrous1999space}. We now explore an analogous notion for interactive unitary synthesis.

\emph{A priori}, it is unclear how to reasonably define zero knowledge in the unitary synthesis setting. First, defining zero-knowledge quantum protocols for decision languages is already challenging, as the notion of ``view'' in the quantum setting is less straightforward than with classical protocols~\cite{watrous2002limits,watrous2006zero}.
Second, in the unitary synthesis setting the verifier additionally gets one copy of an unknown state $\ket{\psi}$ for the quantum part of its input.

We first explore several attempts to define zero knowledge for unitary synthesis, and highlight their shortcomings.
For simplicity, for this discussion we will ignore the precision parameter $r$.
A first attempt is to require that the view of the verifier, when given instance $x$ and a quantum input $\ket{\psi}$ and interacts with the honest prover, can be efficiently output by the simulator $\Sim$ that only receives instance $x$ and state $\ket{\psi}$ as input and does not interact with the prover. However, since the verifier is supposed to end up with $U_x \ket{\psi}$ at the end of the protocol, this means that the simulator can output $U_x \ket{\psi}$ from $x$ and $\ket{\psi}$ in polynomial time, meaning that $\usynth{U} \in \unitaryBQP$. This would lead to an uninteresting definition of zero knowledge.

A second attempt to define zero knowledge is inspired by simulation-based security, where we allow the simulator to query the ideal Uhlmann transformation $U_x$ once.
In particular, the simulator gets as input the honest verifier's input $\ket\psi$, and gets a single query to $U_x$, before being asked to output the verifier's view.
This still seems problematic in the honest verifier setting, since the simulator might decide to query $U_x$ on a state other than $\ket\psi$. If it does that, it seems tricky to argue that the verifier does not learn anything from the interaction since it could potentially learn the target unitary transformation applied to a state that is completely unrelated to the input. 

These difficulties point to the core issue with devising a notion of zero knowledge in the unitary synthesis setting. With the standard definition of zero knowledge for decision problems, the input and outputs of the verifier are fully specified for the simulator: in particular, the simulator only has to reproduce the interaction in the accepting case. In the unitary synthesis setting, the verifier does not have a full classical description of what state it is supposed to output: the classical string $x$ provides the simulator with a complete classical description of the partial isometry $U_x$, but it only gets the input state $\ket{\psi}$ in quantum form.

This motivates us to define a notion of \emph{honest-verifier, average-case} zero knowledge for unitary synthesis, where we consider verifiers that get a classical input $x$ and a subsystem of a fixed state $\ket{\psi_x}$ (like the fixed distribution state in \cref{def:avgunitaryBQP_avgunitaryPSPACE}). We assume the distribution state $\ket{\psi_x}$ has an efficient classical description (i.e.~it comes from a $\stateBQP$ state family). Thus, the input/output behavior of the unitary synthesis protocol when both the verifier and prover are honest is completely specified, which then allows for the possibility of a simulator. Although the honest-verifier and average-case conditions may appear restrictive, we believe that this definition captures a reasonable notion of zero-knowledge in the unitary synthesis setting. Furthermore, as we will see in \Cref{sec:structural-uhlmann}, it is deeply related to the complexity of the Uhlmann Transformation Problem.

\begin{definition}[Honest-verifier, zero-knowledge unitary synthesis] \label{def:avgUnitarySZK}
    Let $c,s,\eps:\N \times \N \to [0,1]$ be functions. The class $\avgUnitary{HVSZK}_{c,s,\eps}$ is the set of distributional unitary synthesis problems $(\usynth{U} = (U_x)_x,\Psi = (\psi_x)_x) \in \avgUnitaryQIP_{c,s}$ such that
    \begin{enumerate}
        \item The distributional state family $\Psi \in \stateBQP$.
        \item There is an \emph{honest verifier} $V^* = (V_{x,r})_{x,r}$ and an \emph{honest prover} $P^*$ satisfying the $\avgUnitaryQIP_{c,s}$ completeness and soundness conditions for $(\usynth{U},\Psi)$, and
        \item There exists a polynomial-time quantum algorithm $\Sim$ (called the \emph{simulator}) ssuch that on input $(x,r,j)$ (for
        $j \in \N$), outputs a state $\rho$ satisfying
        \[
            \td(\rho,\sigma_{x, r,j}) \leq \eps(|x|, r)
        \]
        where $\sigma_{x, r,j}$
        is the joint density matrix of both the honest verifier $V^*_{x, r}$'s private register \emph{and} the ancilla register of the input $\ket{\psi_x}$, immediately after the $j$'th round of interaction with the honest prover $P^*$.
    \end{enumerate}

\end{definition}

\begin{remark}
Note that the distribution state sequence $\Psi$ associated with a distributional unitary synthesis problem in $\avgUnitary{HVSZK}$ is required to be in $\stateBQP$, instead of some notion of ``zero-knowledge state synthesis''. This may appear to violate our principle that the complexity of the distributional state corresponds to the complexity of the unitary synthesis class, but note that the simulator in a $\avgUnitary{HVSZK}$ protocol can synthesize the distributional state (because the simulator can produce the ancilla register of $\ket{\psi_x}$ and the verifier's view at the beginning of the protocol). This implies that the distributional state is in $\stateBQP$.
\end{remark}

\paragraph{Perfect zero knowledge unitary synthesis.} We also define a special subclass of $\avgUnitary{HVSZK}$ where the simulator can \emph{perfectly} reproduce the view of the verifier. We call this class $\avgUnitary{HVPZK}$, which is analogous to the decision complexity class \emph{perfect zero knowledge}, or $\mathsf{PZK}$, which is a subclass of $\mathsf{SZK}$. 

\begin{definition}[$\avgUnitary{HVPZK}$]
\label{def:avgUnitaryPZK}
    The class $\avgUnitary{HVPZK}_{s}$ is defined to be $\avgUnitary{HVSZK}_{1, s,0}$, i.e., with completeness $1$, soundness $s$, and zero simulator error. When we don't specify the soundness parameter $s$, we set it to be $1/2$ by default. \end{definition}
In \Cref{sec:structural-uhlmann} we will show that $\DistUhlmann_1$ is a complete problem for $\avgUnitary{HVPZK}$.

\subsection{Discussion}

\subsubsection{Completeness and soundness parameters}

For classical interactive protocols, the completeness and soundness parameters can be amplified in a black box fashion. As a result, there one typically chooses canonical parameters (e.g., completeness $=2/3$ and soundness $=1/3$).
In the state synthesis setting, the soundness parameter can also be generically amplified via sequential repetition (see~\cite{rosenthal2022interactive} for a proof). 

However, it is not clear whether soundness amplification is possible in the unitary synthesis setting. Conceptually, this is because the verifier only gets one copy of the input state, and if a verifier does not accept the interaction it is unclear how to recover the input state for another repetition of the protocol. 
This is why we keep the subscripts for completeness and soundness in our definitions of $\avgUnitaryQIP$ and $\avgUnitary{HVSZK}$ explicit.

\begin{openproblem}
    Can completeness/soundness amplification be performed for $\avgUnitaryQIP$, or is there evidence that this is impossible?
\end{openproblem}

\subsubsection{Definition of average-case interactive unitary synthesis} \label{sec:interactive_average_discussion}

In \cref{sec:classical_vs_quantum_inputs}, we have explained how distinguishing between classical and quantum inputs allows us to consider two different average-case notions for unitary synthesis: distributional unitary synthesis problems, where the input state to the unitary is part of a larger entangled state, and ``instance-average-case'' unitary synthesis problems, where we consider a distribution over classical instances $x$.

This distinction between distributional and instance-average-case unitary synthesis problems is particularly important in the context of interactive unitary synthesis. Recall from \cref{rem:distributional_explainer} that we can view the quantum input as being a pure state $\ket{\theta}$ sampled from some distribution (because the verifier sees half of some entangled state $\ket{\psi}$, of which $\ket{\theta}$ is a Schmidt vector). In our definition, while the prover can know the classical instance $x$ and the distribution from which $\ket{\theta}$ is sampled, the prover is \emph{not} allowed to know which state $\ket{\theta}$ was sampled from the distribution. The reason for this is because distributional unitary synthesis problems are meant to capture the task of transforming an entangled state $\ket{\psi}$ (for example, breaking the security of a commitment scheme or decoding the radiation of a black hole). Allowing the prover to ``know'' what marginal state (more precisely, Schmidt vector) $\ket{\theta}$ was received by the verifier would lead to unphysical operations on the state $\ket{\psi}$.

\subsubsection{Closure under polynomial-time reductions}
\label{sec:qip-closure-under-ptime}

It is not immediately clear from the definitions of $\avgUnitaryQIP$ and $\avgUnitary{HVSZK}$ whether they are closed under polynomial-time reductions. Let us see why this is not so straightforward. Let $(\usynth{U} = (U_x),\Psi = (\psi_x)) \in \avgUnitaryQIP$, and let $(\usynth{V} = (V_x),\Phi = (\phi_x)) \in \avgUnitary{BQP}^{(\usynth{U},\Psi)}$, meaning that there exists a polynomial time query algorithm $A$ with oracle access to $(\usynth{U},\Psi)$ that implements $(\usynth{V},\Phi)$.  

We would like to argue that $(\usynth{V},\Phi) \in \avgUnitaryQIP$ by giving an interactive synthesis protocol for it. A natural verifier $V$ would be as follows: it would simulate the query algorithm $A$; every time $A$ makes a query to $\usynth{U}$, the verifier $V$ would run the interactive protocol for $\usynth{U}$. Then $V$ would accept only if all of the invocations of the $\usynth{U}$ protocols accepted. The completeness of this protocol is straightforward; the overall protocol would accept with high probability, provided that each of the sub-protocols accepted with high enough probability. 

The soundness of the protocol is less clear. Supposing the overall protocol succeeded with probability at least $1/2$ (say), each of the sub-protocols must have accepted with probability at least $1/2$. One would like to invoke the soundness of the sub-protocol, except it only holds if the input state was the appropriate distribution state from $\Phi$. Unfortunately, the query algorithm $A$ can query the oracle $\usynth{V}$ on any input state it likes. For example, one can imagine algorithms $A$ that, for some reason, query $\usynth{U}$ on a completely random input state, unrelated to the distribution state $\Phi$. This may be fine in the query setting, but because there is an adversarial prover involved, the soundness of the protocol for $\usynth{U}$ does not directly imply that (conditioned on acceptance) the protocol has implemented an average-case implementation of $\usynth{U}$. 

As it turns out, $\avgUnitaryQIP$ \emph{is} closed under polynomial-time reductions, but this follows from our proof of $\avgUnitaryQIP = \avgUnitaryPSPACE$ in \Cref{sec:structural-succinct-uhlmann}. This still leaves a couple open questions:

\begin{openproblem}
    Is the protocol based on simulating the query algorithm sound? In other words, is there a more direct way to show that $\avgUnitaryQIP$ is closed under polynomial-time reductions? 
\end{openproblem}

\begin{openproblem}
    Is $\avgUnitary{HVSZK}$ closed under polynomial-time reductions?
\end{openproblem}

We note that even if one shows that the protocol which simulates the query algorithm is sound, there is still a question of whether the protocol is zero-knowledge (i.e., there is a simulator).

\subsubsection{Other interactive unitary complexity classes}

In this section, we have introduced the unitary complexity classes $\mathsf{(avg)UnitaryQIP}$, $\avgUnitary{HVSZK}$, and $\avgUnitary{HVPZK}$.
We will explore these in detail in \cref{part:uhlmann_general} and show that they are closely related to the Uhlmann Transformation Problem. 
One can of course introduce and study additional classes for interactive unitary synthesis; the most obvious is $\mathsf{unitaryQMA}$, which we have already mentioned in \cref{prob:unitaryQMA}.

Beyond this, a natural question is whether our notion of zero knowledge, which we have introduced for the honest verifier setting, can be meaningfully generalized to the \emph{malicious verifier} setting.
In that setting, the interaction between the honest prover and verifier can be efficiently simulated even if the verifier deviates from the protocol. This is typically the notion of zero knowledge that is useful in the cryptographic setting. It is known that in both the classical and quantum settings, the malicious verifier and honest verifier definitions of statistical zero knowledge proofs yield the same complexity classes (i.e., $\class{HVSZK} = \class{HVSZK}$ and $\class{QSZK} = \class{QSZK}$)~\cite{okamoto1996relationships,goldreich1998honest,watrous2006zero}. We leave studying stronger notions of zero knowledge protocols for unitary synthesis to future work:
\begin{openproblem}
    Is there a meaningful notion of malicious verifier zero knowledge for unitary synthesis problems, and how is that related to the honest verifier setting that we considered here?
\end{openproblem}

Finally, in this section we have only considered single-prover interactive protocols.
However, in traditional (classical and quantum) complexity theory, multi-prover protocols have been shown to be surprisingly powerful~\cite{babai1991non,ji2021mip}.
It is natural to ask whether multi-prover models might also provide additional power (and insights) in the unitary synthesis setting:
\begin{openproblem}
    Is there a meaningful notion of multi-prover unitary synthesis protocols, and what is their power?
\end{openproblem}

 \newpage
\part{Uhlmann Transformation Problem: Definitions and Complexity}
 \label{part:uhlmann_general}
\section{The Uhlmann Transformation Problem}
\label{sec:uhlmann}

In this section we formally define the Uhlmann Transformation Problem as a unitary synthesis problem. We also define a ``succinct'' version of it, in which the two states $\ket{C},\ket{D}$ specifying an instance of the Uhlmann Transformation Problem may have exponential circuit complexity, but have a succinct polynomial-size classical description.

\subsection{Uhlmann's theorem and canonical Uhlmann transformations} 
\label{sec:uhlmanns_theorem}
We begin by recalling Uhlmann's theorem.
\begin{theorem}[{Uhlmann's theorem \cite{uhlmann1976transition}}] \label{thm:uhlmann_std}
Let $\ket{C}_{\reg{AB}}$ and $\ket{D}_{\reg{AB}}$ be pure states on registers $\reg{AB}$ and denote their reduced states on register $\reg{A}$ by $\rho$ and $\sigma$, respectively.
Then, there exists a unitary $U$ acting only on register $\reg{B}$ such that
\[
\fidelity(\rho, \sigma) = |\bra{D} (\Id_\reg{A} \ot U_\reg{B}) \ket{C}|^2 \,.
\]
\end{theorem}

We would like to define a unitary synthesis problem $(U_x)_x$ corresponding to Uhlmann's theorem. Intuitively, whenever the string $x$ represents a pair of bipartite states $\ket{C}, \ket{D}$ (by specifying circuits for them, for example), the unitary $U_x$ should satisfy the conclusion of Uhlmann's theorem. However a subtlety that arises: the unitary $U$ in \Cref{thm:uhlmann_std} is not unique; outside of the support of $\rho = \ptr{\reg{A}}{\proj{C}}$, $U$ can act arbitrarily. This motivates defining a \emph{canonical} Uhlmann transformation $W$ corresponding to a pair of bipartite states $\ket{C},\ket{D}$.

\begin{definition}[Canonical Uhlmann transformation]
\label{def:canonical-uhlmann}
The \emph{canonical Uhlmann transformation} corresponding to a pair of pure states $(\ket{C}_{\reg{AB}},\ket{D}_{\reg{AB}})$ is defined as
\begin{equation}
    \label{eq:def-canonical-uhlmann-isometry}
    W = \sgn(\Tr_{\reg{A}}(\ketbra{D}{C}))~.
\end{equation}
\end{definition}
For any linear operator $K$ with singular value decomposition $U \Sigma V^\dagger$, we define $\sgn(K) = U \sgn(\Sigma) V^\dagger$ with $\sgn(\Sigma)$ denoting replacing all the nonzero entries of $\Sigma$ with $1$ (which is the same as the usual sign function since all singular values are non-negative).

The next proposition justifies why \Cref{def:canonical-uhlmann} is canonical:
\begin{proposition}
\label{prop:canonical-uhlmann-properties}
    The map $W$ defined in~\cref{eq:def-canonical-uhlmann-isometry} is a partial isometry, and satisfies the following. Let $\rho, \sigma$ denote the reduced density matrices of $\ket{C},\ket{D}$, respectively, on register $\reg{A}$. Then a channel $\Phi$ acting on register $\reg{B}$ satisfies 
\[
    \fidelity \Big( (\id_{\reg{A}} \otimes \Phi_{\reg{B}})(\proj{C}{}) \, , \, \proj{D} \Big) = 
\fidelity(\rho,\sigma)
\]
if and only if $\Phi$ is a channel completion of $W$.

\end{proposition}

\begin{proof}
Let $X, Y$ be unitary operators acting on register $\reg{B}$ such that
\begin{gather*}
    \ket{C} = \sqrt{\rho} \otimes X \ket{\Omega} \\
    \ket{D} = \sqrt{\sigma} \otimes Y \ket{\Omega}
\end{gather*}
where $\ket{\Omega} = \sum_i \ket{i}_{\reg{A}} \ket{i}_{\reg{B}}$ is the unnormalized maximally entangled state in the standard basis. 
Let $U \Sigma V^\dagger$ denote the singular value decomposition of $(\sqrt{\rho} \sqrt{\sigma})^\top$, the transpose of $\sqrt{\rho} \sqrt{\sigma}$ with respect to the standard basis. Then the proof of~\cite[Lemma 7.6]{metger2023stateqip} shows that
\begin{equation}
    W = Y U \sgn(\Sigma) V^\dagger X^\dagger\,.
\end{equation}

    The fact that $W$ is a partial isometry is clear: since the matrices $X, U, V, Y$ are unitary and $\sgn(\Sigma)$ is a projection, it can be written in the form $W = \Pi F$ where $\Pi = X U \sgn(\Sigma) U^\dagger X^\dagger$ is a projection and $F = X U V^\dagger Y^\dagger$ is a unitary. This means that $W$ is a parital isometry.

    We now prove the ``if'' statement (if $\Phi$ is a channel completion of $W$, then it achieves the optimal Uhlmann fidelity); first we assume that $\Phi$ is in fact a unitary channel $X \mapsto RXR^\dagger$ for a unitary completion $R = cW + W^\perp$ for some constant $c \in \C$.  Assume without loss of generality that $c = 1$. 

    Suppose for sake of contradiction that $\fidelity((\id \otimes \Phi) \proj{C}, \proj{D}) = |\bra{D} \id \otimes R \ket{C}
    |^2 \neq \fidelity(\rho,\sigma)$. The proof of~\cite[Lemma 7.6]{metger2023stateqip} shows that $\bra{D} (\id \otimes W) \ket{C} = \sqrt{\fidelity(\rho,\sigma)}$; this then implies that $\bra{D} \id \otimes W^\perp \ket{C} \neq 0$. Let $e^{i \theta}$ be a complex phase such that $e^{i \theta} \bra{D} \id \otimes W^\perp \ket{C}$ is a strictly positive number. Then consider the unitary $R' = W + e^{i \theta} W^\perp$. Then
\[
    |\bra{D} \id \otimes R' \ket{C}|^2 = |\sqrt{\fidelity(\rho,\sigma)} + e^{i \theta} \bra{D} \id \otimes W^\perp \ket{C}|^2 > \fidelity(\rho,\sigma)
\]
which contradicts Uhlmann's theorem. 

    Now suppose that $\Phi$ is a general channel completion of $W$. Let $V$ denote a unitary Stinespring dilation of $\Phi$ that maps registers $\reg{BE}$ to $\reg{BE}$. Note that
    \begin{equation}
        \label{eq:canonical-uhlmann-properties-0}
        \fidelity \Big( (\id \otimes \Phi) \proj{C}, \proj{D} \Big) = \max_{\ket{\theta}} \Big | (\bra{D} \otimes \bra{\theta})(\id_{\reg{A}} \otimes V)( \ket{C} \otimes \ket{0}) \Big|^2
    \end{equation}
    where the maximization is over pure states $\ket{\theta}$ on register $\reg{E}$.  
    
    Let $P = W^\dagger W$ denote the projection onto the domain of $W$, respectively. Since $\Phi(\ketbra{a}{b}) = W\ketbra{a}{b}W^\dagger$ for all $\ket{a},\ket{b}$ in the support of $P$, we have that $V \ket{a} \ket{0} = (W \ket{a})_{\reg{B}} \otimes \ket{\varphi}$ for some state $\ket{\varphi}$ on register $\reg{E}$. This implies that $V = W \otimes \ketbra{\varphi}{0} + V^\perp$ for some partial isometry $V^\perp$. On the other hand, this means that $V$ is a unitary completion of the canonical Uhlmann transformation between the pair of states $(\ket{C} \otimes \ket{0},\ket{D} \otimes \ket{\varphi})$, which is $W \otimes \ketbra{\varphi}{0}$. Using reasoning analogous to that about unitary completions of $W$, we obtain that ~\eqref{eq:canonical-uhlmann-properties-0} is at least $\fidelity(\rho,\sigma)$, but on the other hand by Uhlmann's theorem is at most $\fidelity(\rho,\sigma)$ which concludes the proof of the ``if'' direction.

    We now prove the ``only if'' direction (if the channel $\Phi$ achieves the optimal Uhlmann fidelity, it must be a channel completion of $W$). Similarly to the proof of the ``if'' direction we prove this first for unitary channels. Let $R$ be a unitary such that $|\bra{D} (\id \otimes R) \ket{C}|^2 = \fidelity(\rho,\sigma)$. We note that the proof of Uhlmann's theorem~\cite[Theorem 9.2.1]{wilde2013quantum} shows that
    \[
         |\bra{D} (\id_\reg{A} \ot R_\reg{B}) \ket{C}|^2 =  |\Tr(Y^\dagger R X (\sqrt{\sigma} \sqrt{\rho})^\top) |^2~.
    \]
    Note that the singular value decomposition of $(\sqrt{\sigma} \sqrt{\rho})^\top$ is $V \Sigma U^\dagger$. By assumption, 
    \[
        |\Tr(Y^\dagger R X (\sqrt{\sigma} \sqrt{\rho})^\top) |^2 = \fidelity(\rho,\sigma) = \Tr(\Sigma)^2~.
    \]
    The last equality follows from the fact that $\fidelity(\rho,\sigma) = \Tr(|\sqrt{\rho} \sqrt{\sigma}|)^2 = \Tr(\Sigma)^2$. 
Let $c$ be a phase such that $c^\dagger \cdot\Tr(Y^\dagger R X (\sqrt{\sigma} \sqrt{\rho})^\top)$ is a nonnegative real number. Substituting in the singular value decomposition of $(\sqrt{\sigma} \sqrt{\rho})^\top$, we get
    \[
        \Tr(\Sigma) = \Tr(\Sigma M)
\]
    where $M := c^\dagger U^\dagger Y^\dagger R XV$. The only way that equality is achieved is if $M$ acts as identity on the support of $\Sigma$. Therefore
$M = P + M^\perp$ where $P = \sgn(\Sigma)$ is the projector onto the nonzero entries of $\Sigma$ and $M^\perp$ is some unitary acting on $I-P$, the orthogonal complement of $P$.

    Note that 
    \[
        R = cYU (P + M^\perp) V^\dagger X^\dagger = cW + cYU M^\perp V^\dagger X^\dagger~.
    \]
    Letting $W^\perp := cYU M^\perp V^\dagger X^\dagger$, we observe that it is a partial isometry with support and range disjoint from the support and range respectively of $W$. This concludes the proof of the ``only if'' direction for unitary channels $\Phi(X) = RXR^\dagger$. 

    Now suppose that $\Phi$ were a general channel achieving the optimal Uhlmann fidelity. Just like with the proof of the ``if'' direction, consider the unitary Stinespring dilation $V$ of $\Phi$. This satisfies 
    \begin{equation}
        \label{eq:canonical-uhlmann-properties-1}
        \fidelity(\rho,\sigma) = \fidelity \Big( (\id \otimes \Phi) \proj{C}, \proj{D} \Big) = \Big | (\bra{D} \otimes \bra{\varphi})(\id_{\reg{A}} \otimes V_{\reg{BE}})( \ket{C} \otimes \ket{0}) \Big|^2
    \end{equation}
    for some state $\ket{\varphi}$. This implies that $V$ is a unitary completion of the canonical Uhlmann transformation $W \otimes \ketbra{\varphi}{0}$ for the pair of states $(\ket{C} \otimes \ket{0},\ket{D} \otimes \ket{\varphi})$. By the above reasoning, this implies that $V = c W \otimes \ketbra{\varphi}{0} + V^\perp$ for some partial isometry $V^\perp$ with disjoint support and range and some constant $c \in C$. It can be verified that the corresponding channel $\Phi(X) = \Tr_{\reg{E}}( V(X \otimes \ketbra{0}{0} )V^\dagger)$ is a channel completion of $W$. This completes the proof of the ``only if'' direction.
    
\end{proof}

\subsection{Uhlmann Transformation Problem}
\label{subsec:worst_case_uhlmann}

We now formulate unitary synthesis problems corresponding to Uhlmann transformations. First, we define explicit and succinct descriptions of quantum circuits. 
\begin{definition}[Explicit and succinct descriptions of quantum circuits]
An \emph{explicit description} of a unitary quantum circuit $C$ is a sequence $(1^n, g_1,g_2,\ldots)$ where $1^n$ represents in unary the number of qubits that $C$ acts on, and $g_1,g_2,g_3,\ldots$ is a sequence of unitary gates. 

A \emph{succinct description} of a quantum circuit $C$ is a pair $(1^n,\hat{C})$ where $\hat{C}$ is a description of a classical circuit\footnote{Here, we think of $\hat{C}$ as being a list of AND, OR, and NOT gates.}  that takes as input an integer $t$ in binary and outputs the description a unitary gate $g_t$ coming from some universal gate set, as well as the (constant-sized) set of qubits that $g_t$ acts on. Together, the gates $g_1,\ldots,g_T$ describe a circuit $C$ acting on $n$ qubits; we will always denote the classical circuit with a hat (e.g.~$\hat C$) and use the same letter without a hat (e.g.~$C$) for the associated quantum circuit.
\end{definition}

We make a few remarks about the definitions of explicit and succinct descriptions of quantum circuits:
\begin{enumerate}[label=(\roman*)]
    \item The length of an explicit description of a quantum circuit is polynomial in the number of gates in the circuit as well as the number of qubits it acts on.
    \item In a succinct description of a quantum circuit $C$, the size of the circuit may be exponentially larger than the length of the description $(1^n,\hat{C})$.
    However, the number of qubits that $C$ acts on is polynomial (in fact, at most linear) in the description length.
    \item For a succinct description, we provide the number of qubits $n$ in the quantum circuit explicitly in unary because given only the classical circuit $\hat C$ it may be difficult to compute the the number of qubits that the quantum circuit $C$ acts on.
\end{enumerate}

We now define two variants of the Uhlmann Transformation Problem. In the first, the two bipartite states are described by explicit circuit descriptions, and in the second they are described by succinct circuit descriptions. 

\begin{definition}[Valid Uhlmann instances] \label{def:valid_uhlmann_string}
We say that a string $x \in \bits^*$ is a \emph{valid Uhlmann instance} if it encodes a tuple $(1^n,C,D)$ where $C,D$ are explicit descriptions of \emph{unitary} circuits that each act on $2n$ qubits. We say that $x$ is a \emph{valid succinct Uhlmann instance} if $x = (1^n, \hat C, \hat D)$ is a succinct description of a pair $(C,D)$ of unitary circuits that each act on $2n$ qubits for some $n$. 

We further say that a valid (possibly succinct) Uhlmann instance $x$ is a \emph{fidelity-$\kappa$ instance} if the reduced states $\rho,\sigma$ of the states $\ket{C} = C \ket{0^{2n}}$, $\ket{D} = D \ket{0^{2n}}$ on the first $n$ qubits satisfy $\fidelity(\rho,\sigma) \geq \kappa$.
\end{definition}

\begin{definition}[Uhlmann Transformation Problem] \label{def:utp}
Let $\kappa:\N \to [0,1]$ be a function. The \emph{$\kappa$-fidelity Uhlmann Transformation Problem} is the unitary synthesis problem $\Uhlmann_{\kappa} = (U_x)_{x \in \bits^*}$ where whenever $x = (1^n,C,D)$ is a fidelity-$\kappa(n)$ Uhlmann instance specifying a pair $(C,D)$ of unitary circuits that each act on $2n$ qubits for some $n$, then $U_x$ is the canonical Uhlmann transformation for the pair of states $(\ket{C},\ket{D}$). Otherwise if $x$ is not a valid Uhlmann instance, then we define $U_x = 0$ (i.e., a partial isometry with zero-dimensional support). 

The \emph{$\kappa$-fidelity Succinct Uhlmann Transformation Problem}, denoted by $\SuccinctUhlmann_{\kappa}$, is the sequence $(U_x)_x$ where whenever $x = (1^n,\hat{C},\hat{D})$ is a valid fidelity-$\kappa(n)$ succinct Uhlmann instance specifying a pair $(C,D)$ of unitary circuits that each act on $2n$ qubits for some $n$, then $U_x$ is the canonical Uhlmann transformation for the pair of states $(\ket{C},\ket{D})$; if $x$ is not a valid succinct Uhlmann instance, then we define $U_x = 0$. \end{definition}

One should think of the fidelity parameter $\kappa$ as a promise: an algorithm for $\Uhlmann_\kappa$ only has to work for state pairs $\ket{C}, \ket{D}$ whose reduced states have fidelity at least $\kappa$.
The smaller $\kappa$, the more difficult $\Uhlmann_\kappa$ can become, in the sense that an algorithm for $\Uhlmann_\kappa$ also works for all $\Uhlmann_{\kappa'}$ for $\kappa' \geq \kappa$.

\subsection{Distributional Uhlmann Transformation Problem}
\label{subsec:dist_uhlmann}

To define average case versions of the Uhlmann Transformation Problems we specify a distribution state $\ket{\psi_x}$ for every valid (succinct or non-succinct) Uhlmann instance $x$.

\begin{definition}[Distributional Uhlmann Transformation Problems] \label{def:canonical_dist_UTP}
We define a state sequence $\Psi_{\Uhlmann} = (\ket{\psi_x})_{x \in \bits^*}$ as follows: for all $x \in \{0,1\}^*,$
\[
    \ket{\psi_x} = \begin{cases}
        \ket{C} & \text{if $x = (1^n,C,D)$ is valid Uhlmann instance,} \\
        0 & \text{otherwise.}
    \end{cases}
\]
Then, the \emph{distributional $\kappa$-fidelity Uhlmann Transformation Problem} is the distributional unitary synthesis problem $\DistUhlmann_{\kappa} = (\Uhlmann_{\kappa}, \Psi_{\Uhlmann})$.

The state sequence $\Psi_{\SuccinctUhlmann}$ and the distributional unitary synthesis problem $\DistSuccinctUhlmann_\kappa$ are defined analogously.

\end{definition}

We now argue that this choice of distribution state is natural for the Uhlmann Transformation Problems: being able to solve the distributional Uhlmann Transformation Problems in the average-case essentially coincides with being able to perform the Uhlmann transformation corresponding to a pair of (succinctly or non-succinctly described) states. The next proposition captures this equivalence in the \emph{high $\kappa$ regime}, where the promised fidelity $\kappa$ is close to $1$. It can also be viewed as a robust version of \Cref{prop:canonical-uhlmann-properties}. 
\begin{proposition}
\label{prop:operational-avg-case-uhlmann}
Let $\ket{C}_{\reg{AB}},\ket{D}_{\reg{AB}}$ denote two bipartite states with reduced density matrices $\rho,\sigma$ respectively such that $\fidelity(\rho,\sigma) = \kappa$. Let $M$ is a quantum algorithm acting on register $\reg{B}$
such that 
\begin{equation}
    \label{eq:canonical-distribution-0-0}
     \fidelity \Big( (\id_{\reg{A}} \otimes M_{\reg{B}})(\ketbra{C}{C}), \ketbra{D}{D} \Big) \geq \kappa - \delta
\end{equation}
for some $\delta$. 
Then there exists a channel completion $\Phi$ of the canonical Uhlmann transformation $W$ for $(\ket{C},\ket{D})$ such that 
\[
    \td \Big( (\id \otimes M)(\proj{C}) \, , \, (\id \otimes \Phi)(\proj{C}) \Big) \leq 2\sqrt{1 - \kappa} + \sqrt{\delta}~.
\]

Conversely, suppose that there exists a channel completion $\Phi$ of $W$ such that
\[
    \td \Big( (\id \otimes M)(\proj{C}) \, , \, (\id \otimes \Phi)(\proj{C}) \Big) \leq \delta~.
\]
Then 
\[
    \td \Big( (\id \otimes M)(\proj{C}) \, , \, \proj{D} \Big) \leq \delta + \sqrt{1 - \kappa}~. \]

\end{proposition}
\begin{proof}
We will prove this proposition for the case of Uhlmann instances; the case of succinct Uhlmann instances is entirely analogous.

We begin with the first part of the proposition.
Let $W$ denote the canonical Uhlmann transformation corresponding to $(\ket{C},\ket{D})$. Let $\Phi$ denote a channel completion of $W$; by \Cref{prop:canonical-uhlmann-properties} we have that $\fidelity((\id \otimes \Phi)\proj{C},\proj{D}) = \kappa$. 
By the triangle inequality, we have
\begin{align*}
    &\td \Big( (\id \otimes M)(\ketbra{C}{C}), (\id \otimes \Phi)(\ketbra{C}{C}) \Big) \\
    & \qquad \leq \td \Big( (\id \otimes M)(\ketbra{C}{C}), \ketbra{D}{D} \Big) + \td \Big(\ketbra{D}{D}, (\id \otimes \Phi)(\ketbra{C}{C}) \Big) \\
    &\qquad \leq \sqrt{1 - \kappa + \delta} \, + \sqrt{1 - \kappa} \\
    &\qquad \leq 2 \sqrt{1 - \kappa} + \sqrt{\delta}
\end{align*}
where in the third line we applied the Fuchs-van de Graaf inequality to \cref{eq:canonical-distribution-0-0}. 
This shows that one the state $\ket{C}$, $M_x$ behaves (approximately) like a channel completion of the canonical Uhlmann transformation. 
By \cref{def:avg_case_error}, this means that $M_x$ (approximately) implements the $\DistUhlmann$ problem as claimed in the first part of the proposition.

    We now prove the ``Conversely'' part of the proposition. By the triangle inequality
    \begin{align}
        &\td \Big( (\id \otimes M)(\ketbra{C}{C}), \ketbra{D}{D} \Big) \notag \\
        &\qquad \leq \td \Big( (\id \otimes M)(\ketbra{C}{C}), (\id \otimes \Phi) \, \ketbra{C}{C} \Big) + \td \Big ((\id \ot \Phi) \ketbra{C}{C}, \proj{D} \Big ) \notag \\
        &\qquad \leq \delta + \sqrt{1 - \kappa}
 \label{eq:canonical-distribution-6}
    \end{align}
    where the last line follows from our assumption on $M$, \Cref{prop:canonical-uhlmann-properties}, and Fuchs-van de Graaf. 

\end{proof}

 \section{Complexity of the Uhlmann Transformation Problem}
\label{sec:structural-uhlmann}

We show that $\DistUhlmann_1$ (the distributional Uhlmann Transformation Problem with fidelity promise $\kappa = 1$) is complete for the unitary complexity class $\avgUnitary{HVPZK}$ (honest-verifier perfect zero knowledge unitary synthesis) defined in \Cref{sec:protocols}. We also discuss an approach to showing that $\DistUhlmann_\kappa$ for fidelity promise $\kappa < 1$ is complete for the class $\avgUnitary{HVSZK}$ (i.e., where the simulator can make some error). This would be analogous to the famous result of Sahai and Vadhan~\cite{10.1145/636865.636868} showing that deciding whether the statistical distance between two efficiently sampleable distributions is large or small is complete for the decision class $\mathsf{SZK}$.

\subsection{A complete problem for \titleavgUnitaryHVPZK}
\label{sec:pzk-completeness}

In this section, we show that $\DistUhlmann_{1}$ is complete for the unitary complexity class $\avgUnitary{HVPZK}$; recall that this is the class of unitary synthesis problems implementable via a zero-knowledge interactive proof with perfect completeness, soundness $\frac{1}{2}$, and zero simulation error (see \Cref{def:avgUnitaryPZK}). First we show containment.

\begin{theorem}
    \label{thm:distuhlmann-in-pzk}
    $\DistUhlmann_1 \in \avgUnitary{HVPZK}$.
\end{theorem}
\begin{proof}
    In order to prove the claim, we need to exhibit an interactive verifier and simulator that satisfy the conditions of \Cref{def:avgUnitaryPZK}. Consider the following protocol (\Cref{prot:pzk_for_uhlmann}).

    \begin{longfbox}[breakable=false, padding=1em, margin-top=1em, margin-bottom=1em]
    \begin{protocol} {\bf Perfect zero-knowledge protocol for $\DistUhlmann_1$} \label{prot:pzk_for_uhlmann}
    \end{protocol}
    \noindent \textbf{Instance: } A valid $\Uhlmann_{1}$ instance $x = (1^n,C,D)$ and precision $r \in \mathbb{N}$. \\
    \noindent \textbf{Input: } An $n$ qubit quantum register $\reg{B}_0$.  
    
    \begin{enumerate}
        \item Let $m = 8r^2$.  Sample $i^* \in [m]$ uniformly at random.

        \item For $i = 1$ though $m$:
        \begin{enumerate}
            \item If $i \neq i^*$:
            \begin{enumerate}
                \item Starting with all zeroes in registers $\reg{A}' \reg{B}'$, prepare the state $\ket{C}_{\reg{A}' \reg{B}'}$, send $\reg{B}'$ to the prover.
                \item After receiving $\reg{B}'$ back from the prover, apply $D^{\dagger}$ to $\reg{A}'\reg{B}'$, and measure all qubits. If it is not equal to all zeroes, reject.
            \end{enumerate}
            \item If $i = i^*$:
            \begin{enumerate}
                \item Send $\reg{B}_0$ to the prover and receive $\reg{B}_0$ back.
            \end{enumerate}
        \end{enumerate}
        \item If the verifier has not rejected yet, accept and output $\reg{B}_0$.
    \end{enumerate}
\end{longfbox}

    We show that the honest verifier and prover specified in \Cref{prot:pzk_for_uhlmann} satisfy the properties of \Cref{def:avgUnitaryPZK}. We first prove that the honest prover $P^*$ is accepted with probability $1$.

\begin{lemma}[Completeness]
    For all valid $\Uhlmann_{1}$ instances $x = (1^n, C, D)$ and error parameters $r \in \mathbb{N}$, for sufficiently large $n$ the honest prover $P^*$ satisfies
    \begin{equation*}
        \Pr[V_{x, r}(\ket{C}_{\reg{A}_0\reg{B}_0}) \interact P^*] = 1\,.
    \end{equation*}
\end{lemma}
\begin{proof}
    Since $x$ is an $\Uhlmann_1$ instance, the canonical Uhlmann transformation exactly maps $\ket{C}$ to $\ket{D}$.  We define the honest prover as follows: in every round the honest prover receives the $\reg{B}$ register of $\ket{C}$ and applies the canonical Uhlmann transformation, mapping the pure state to $\ket{D}$. The verifier always measures the all $0$ string.  Thus, the honest prover is accepted with probability $1$.
\end{proof}

We now prove the soundness property, i.e., if the verifier accepts with probability at least $1/2$ when interacting with a prover $P$, then conditioned in accepting, the verifier outputs a state $\frac{1}{r}$ close to the ideal output.  
\begin{lemma}[Soundness]
\label{lem:Uhlmann_pzk_qip_soundness}
    For all $\Uhlmann_1$ instances $x = (1^n, C, D)$ and error parameters $r \in \mathbb{N}$, for sufficiently large $n$, for all quantum interactive provers $P$, there exists a channel completion $\Phi_x$ of $U_x$ such that 
    \begin{equation*}
        \text{if } \quad \pr {V_{x, r}(\ket{C}) \interact P \text{ accepts}} \geq \frac{1}{2} \qquad \text{then} \qquad \td(\sigma, (\Phi_x \ot \id) \ketbra{C}{C}) \leq \frac{1}{r}\,,
    \end{equation*}
    where $\sigma$ denotes the output of $V_{x, r}(\ket{C}) \interact P$, conditioned on $V_{x, r}$ accepting. 
\end{lemma}
\begin{proof}
     The verifier in \Cref{prot:pzk_for_uhlmann} is described as sequentially checking, for the ``decoy rounds'' $i \neq i^*$, whether the $i$'th copy of the state $\ket{C}$ was transformed to $\ket{D}$. However its messages to the prover are nonadaptive; thus we can equivalently analyze its soundness by considering the following process: 
     \begin{enumerate}
         \item Generate $m$ copies of $\ket{C}$ in registers $\reg{A}_1 \reg{B}_1,\ldots,\reg{A}_m \reg{B}_m$;
         \item Send and receive the registers $\reg{B}_i$ to the prover one at a time;
         \item Choose a random index $i^* \in [m]$;
         \item Measure all registers $\reg{A}_i \reg{B}_i$ for $i \neq i^*$ using the measurement $\{ \Pi, \id - \Pi\}$ where $\Pi = \ketbra{D}{D}$. Accept if all measurements return the $\Pi$ outcome, and output the $\reg{A}_{i^*} \reg{B}_{i^*}$ registers. Otherwise, reject.
     \end{enumerate}
    This can be seen to be equivalent to the verifier in \Cref{prot:pzk_for_uhlmann}; the ``special'' index $i^*$ corresponds to where we embed the ``true'' copy of $\ket{C}_{\reg{A}_0 \reg{B}_0}$. 

     In this process, the state of the system before the measurement can be written as
     \[
        \rho = (\id \otimes \Lambda)(\ketbra{C}{C}^{\otimes m})
     \]
     where $\Lambda(\cdot)$ denotes the prover's quantum channel acting on the registers $\reg{B}_1 \cdots \reg{B}_m$. Note that the prover's action is completely independent of the embedding index $i^*$. 

     Consider sampling a random $m$-bit string $X$ in the following way: measure each register $\reg{A}_i \reg{B}_i$ of $\rho$ with the $\{ \Pi, \id - \Pi \}$ measurement, and set $X_i = 0$ if the $\Pi$ outcome occurs, otherwise set $X_i = 1$. For each $i$ let $E_i$ denote the event that $X_j = 1$ for all $j \neq i$. The acceptance probability of the verifier is equal to $\frac{1}{m} \sum_{i^* \in [m]} \Pr[E_{i^*}]$. 
     
     On the other hand, the probability of measuring registers $\reg{A}_{i^*} \reg{B}_{i^*}$ and getting the $\id - \Pi$ outcome, conditioned on the verifier accepting, is given by
     \[
        \frac{\frac{1}{m} \sum_{i^*} \Pr[X_{i^*} = 0 \wedge E_{i^*}]}{\frac{1}{m} \sum_{i^*} \Pr[E_{i^*}]} \leq \frac{2}{m} \sum_{i^*} \Pr[X_{i^*} = 0 \wedge E_{i^*}] \leq \frac{2}{m}
     \]
     where we used the assumption that the acceptance probability of the verifier is at least $1/2$, and that all of the events $X_{i^*} = 0\wedge E_{i^*}$ are mutually exclusive. By the Gentle Measurement Lemma~\cite{wilde2013quantum}, we have
     \[
        \td(\sigma, \proj{D}) \leq 2\sqrt{\frac{2}{m}} \leq \sqrt{\frac{8}{m}}
     \]
     where we let $\sigma$ denote the state of the registers $\reg{A}_{i^*} \reg{B}_{i^*}$ (or in \Cref{prot:pzk_for_uhlmann}, registers $\reg{A}_0 \reg{B}_0$) conditioned on the verifier accepting. By our choice of $m$, this is at most $1/r$. 

     Finally, let $\Phi$ denote a channel completion of the canonical Uhlmann transformation $U$ for $x$. By \Cref{prop:canonical-uhlmann-properties}, $\ketbra{D}{D} = (\id \otimes \Phi)(\ketbra{C}{C})$, and thus 
     \[
        \td(\sigma, (\id \otimes \Phi)(\proj{C})) \leq 1/r
     \]
     as desired.

\end{proof}

We now show the protocol staisfies the perfect zero-knowledge condition.
\begin{lemma}[Zero-knowledge]
\label{lem:pzk_zero_knowledge}
There exists a polynomial-time simulator that, on input $(x, r, t)$, outputs a state equal to $\sigma_{x, r, t}$
which is the reduced density matrix of $V^*_{x, r}$'s private register and the purifying register $\reg{A}_0$ of the quantum input, immediately after the $t$'th round of interaction with the honest prover $P^*$.  
\end{lemma}
\begin{proof}
We assume that the simulator samples the random index $i^* \in [m]$ by preparing a uniform superposition in some register $\reg{C}$. Then the state of the verifier at time $t$ is equal to 
\[
    \frac{1}{\sqrt{m}} \Big [ \sum_{i^* < t} \ket{i^*}_{\reg{C}} \otimes \ket{D}_{\reg{A}_0 \reg{B}_0} \otimes \ket{D}_{\reg{A}' \reg{B}'} \Big ] + \frac{1}{\sqrt{m}} \ket{t}_{\reg{C}} \otimes \ket{D}_{\reg{A}_0 \reg{B}_0} \otimes \ket{0}_{\reg{A}' \reg{B}'} + \frac{1}{\sqrt{m}} \Big [ \sum_{i^* > t} \ket{i^*}_{\reg{C}} \otimes \ket{C}_{\reg{A}_0 \reg{B}_0} \otimes \ket{D}_{\reg{A}' \reg{B}'}  \Big ]
\]
because when $i^* < t$, the prover will have already transformed both the true copy as well as the decoy copy of $\ket{C}$; when $i^* = t$, the prover will have already transformed the true copy, and the decoy copy is still the all zeroes state; and when $i^* > t$ the true copy will not have been transformed yet, but the decoy copy has. It is easy to see that this state can be prepared in polynomial time.

\end{proof}
\noindent This completes the proof of \Cref{thm:distuhlmann-in-pzk}.
\end{proof}

\noindent We now show that all problems in $\avgUnitary{HVPZK}$ reduce to $\DistUhlmann_1$.  

\begin{theorem}
\label{thm:dist_uhlmann_pzk_hard}
$\DistUhlmann_1$ is $\avgUnitary{HVPZK}$-hard. 
\end{theorem}
\begin{proof}
    The main idea is as follows: the honest prover for an $\avgUnitary{HVPZK}$ protocol can be efficiently implemented using a $\DistUhlmann_1$ oracle. The oracle is used to perform Uhlmann transformations between the consecutive ``snapshots'' of the verifier's state, which can be efficiently produced from the simulator. 

    Let $(\usynth{U}, \Psi) \in \avgUnitary{HVPZK}$. Let $V = (V_{x,r})$ denote the corresponding $m(|x|,r)$-round verifier. For notational simplicity we fix an instance $x \in \bits^*$, a precision parameter $r \in \N$, and write $m = m(|x|,r)$, $V = V_{x,r}$, $\ket{\psi} = \ket{\psi_x}$, $U = U_x$, and $\Sim(t) = \Sim(x,r,t)$.  Let $\reg{A}$ and $\reg{B}$ denote the target and ancilla registers of $\ket{\psi}$, respectively (i.e., register $\reg{B}$ is never touched during the protocol). 

    Fix a round $1 \leq t \leq m$. Consider the purified circuit of the simulator $\Sim(t)$. On the all zeroes input it outputs a pure state $\ket{\varphi_t}$ on registers $\reg{BFQP}$ where $\reg{B}$ represents the ancilla register of the input state $\ket{\psi}$,  $\reg{F}$ represents the verifier's workspace register, $\reg{Q}$ represents the verifier's message register, and $\reg{P}$ represents the rest of the purification. By definition of perfect zero-knowledge, the reduced density matrix of $\ket{\varphi_t}$ on registers $\reg{BFQ}$ is identical to the same registers of the protocol after the verifier has received the $t$'th message from the honest prover $P^*$. Thus, $\ket{\varphi_t}$ denotes the state of the protocol up to a unitary on register $\reg{P}$. Let $L$ denote the total number of qubits.

    At the beginning of the $(t+1)$'st round of the protocol, the verifier applies the unitary $V_{t+1}$ to registers $\reg{FQ}$ of $\ket{\varphi_t}$, before sending the register $\reg{Q}$ to the prover (who applies a unitary to $\reg{Q}$ as well as its private memory register $\reg{P}$). 
    Note that 
\[
        \Tr_{\reg{QP}}(V_{t+1} \ketbra{\varphi_t}{\varphi_t} V_{t+1}^\dagger) = \Tr_{\reg{QP}}(\ketbra{\varphi_{t+1}}{\varphi_{t+1}})
    \]
    where $\ket{\varphi_{t+1}}$ denotes the (purified) output of the simulator $\Sim(t+1)$. This is because the registers $\reg{BF}$ do not change between $V_{t+1} \ket{\varphi_t}$ and $\ket{\varphi_{t+1}}$ in the actual protocol; the only difference is the prover's action on registers $\reg{QP}$. 
    
    Define $\ket{C_{t+1}} = V_{t+1} \ket{\varphi_t}$ and $\ket{D_{t+1}} = \ket{\varphi_{t+1}}$. Thus the prover's action is an Uhlmann transformation on the registers $\reg{QP}$ between the pair of states $(\ket{C_{t+1}},\ket{D_{t+1}})$.  Since $\Sim$ and $V$ have polynomial-size circuits, so do the pair of states $(\ket{C_{t+1}},\ket{D_{t+1}})$. Call this pair of circuits $(C_{t+1},D_{t+1})$.

    Define $\ket{\varphi_0} = \ket{\psi}_{\reg{BA}} \ket{0 \cdots 0}$. Since the state family $\Psi$ is in $\stateBQP$ this implies that $\ket{\varphi_0}$ also has a polynomial-size circuit. 
    
    Consider the following query algorithm:
    \begin{enumerate}
\item Given input register $\reg{A}$, prepare the rest of the verifier's workspace $\reg{F}$,  message register $\reg{Q}$, and the prover's workspace $\reg{P}$ in the state $\ket{0 \cdots 0}$. 
        \item For $1 \leq t \leq m$:
        \begin{enumerate}
            \item Apply the verifier unitary $V_t$ to register $\reg{FQ}$.
            \item Query the $\DistUhlmann_1$ oracle with instance $(1^L,C_t,D_t)$, precision parameter $s = rm$, and register $\reg{QP}$ of $\reg{C}_t$ as the quantum input. 
        \end{enumerate}
        \item Apply the final verifier unitary $V_{m+1}$ to register $\reg{FQ}$. 
        \item Output register $\reg{A}$ of the verifier's workspace register, which denotes the output of the verifier. 
    \end{enumerate}
    This is a polynomial-time query algorithm since it makes only polynomially-many queries and the circuits $V = (V_t)$ are all polynomial-size. Suppose each of the $\DistUhlmann_1$ oracle calls were implemented without error. By induction, the state at the end of step 2(b) in the query algorithm is $\ket{D_t} = \ket{\varphi_t}$ (this is because by \Cref{prop:canonical-uhlmann-properties}, any channel completion of the canonical Uhlmann transformation maps $\ket{C_t}$ to $\ket{D_t}$). By the completeness property of the zero-knowledge protocol for $(\usynth{U},\Psi)$ implies that the state $\ket{\psi}$ has been exactly transformed to $(U \otimes \id) \ket{\psi}$. However, since each prover action is simulated up to error $1/s$, the total error is at most $m/s = 1/r$. This shows that $(\usynth{U},\Psi)$ can be implemented in polynomial-time, given oracle access to $\DistUhlmann_1$.

\end{proof}

\subsection{The imperfect fidelity case}\label{subsec:padding_trick}

We now turn to the complexity of $\DistUhlmann_{\kappa}$ when $\kappa < 1$. Note that for all $0 \leq \kappa_1 \leq \kappa_2 \leq 1$, we have that all valid instances of $\Uhlmann_{\kappa_2}$ are valid instances of $\Uhlmann_{\kappa_1}$ but not vice versa (a similar statement holds for $\SuccinctUhlmann$). Thus, implementing general $\Uhlmann_{\kappa_1}$ transformations may potentially be more difficult than implementing $\Uhlmann_{\kappa_2}$ transformations. Furthermore, it is no longer apparent that there is a zero-knowledge protocol for, say, $\DistUhlmann_{1/2}$. Thus it is not clear how the complexities of $\Uhlmann_{\kappa_1}$ and $\Uhlmann_{\kappa_2}$ relate to each other for different $\kappa_1,\kappa_2$. 

We present a simple padding trick which shows that as long as $\kappa_1,\kappa_2$ are bounded by at least some inverse polynomial from either $0$ or $1$, the complexities of $\DistUhlmann_{\kappa_1}$ and $\DistUhlmann_{\kappa_2}$ are equivalent under polynomial-time reductions.  

\begin{lemma}[Padding trick]
\label{lem:padding-trick}
    Let $0 \leq \kappa_1 \leq \kappa_2 \leq 1$ and let $C,D$ be circuits on $2n$ qubits such that $\fidelity(\rho,\sigma) 
 \geq \kappa_1$ where $\rho,\sigma$ are the reduced density matrices of $\ket{C} = C\ket{0^{2n}},\ket{D}=D\ket{0^{2n}}$, respectively, on the first $n$ qubits. Let $0 < \alpha \leq (1 - \kappa_2)/(1 - \kappa_1)$. Define the following states $\ket{E},\ket{F}$ on $2(n+1)$ qubits where
    \begin{align*}
        \ket{E} &= \sqrt{\alpha} \ket{0} \ket{C} \ket{0} + \sqrt{1 - \alpha} \ket{1^{2(n+1)}} \\
        \ket{F} &= \sqrt{\alpha} \ket{0} \ket{D} \ket{0} + \sqrt{1 - \alpha} \ket{1^{2(n+1)}}\,.
    \end{align*}
    Suppose that the state $\sqrt{\alpha} \ket{0} + \sqrt{1 - \alpha} \ket{1}$ can be prepared using a circuit of size $s$. Then the following hold:
    \begin{enumerate}
        \item $\ket{E}$, $\ket{F}$ can be computed by circuits $E,F$ of size $O(|C| + |D| + s)$; 
        \item $\fidelity(\tau,\mu) \geq \kappa_2$ where $\tau,\mu$ are the reduced density matrices of $\ket{E},\ket{F}$ on the first $n+1$ qubits;
        \item The canonical $(n+1)$-qubit Uhlmann isometry $V$ for $(\ket{E},\ket{F})$ can be written as
        \[
            V = U \otimes \ketbra{0}{0} + \id \otimes \ketbra{1}{1}
        \]
        where $U$ is the $n$-qubit canonical Uhlmann isometry for $(\ket{C},\ket{D})$.
    \end{enumerate}
\end{lemma}
\begin{proof}
    We prove the first item. To compute the state $\ket{E}$, consider the circuit $E$ on $2(n+1)$ qubits that does the following:
    \begin{enumerate}
        \item  Initialize the first qubit in the state $\sqrt{\alpha} \ket{0} + \sqrt{1 - \alpha} \ket{1}$.
        \item Apply a CNOT from the first qubit to the last qubit. 
        \item Controlled on the first qubit being $\ket{0}$, run the $n$-qubit circuit $C$ on qubits $2$ through $n+1$.
        \item Controlled on the first qubit being $\ket{1}$, apply a bitflip operator to qubits $2$ through $n+1$.
    \end{enumerate}
    Clearly the size of $E$ is $O(|C| + s)$ where $|C|$ denotes the size of circuit $C$ where by assumption there is a circuit of size $s$ to initialize the first qubit. An analogous construction holds for $\ket{F}$. 
    
    For the second item, we have
    \begin{gather*}
        \tau = \alpha \ketbra{0}{0} \otimes \rho + (1 - \alpha) \ketbra{1}{1} \ot \ketbra{1^n}{1^n} \\
        \mu = \alpha \ketbra{0}{0} \otimes \sigma + (1 - \alpha) \ketbra{1}{1} \ot \ketbra{1^n}{1^n}\,.
    \end{gather*}
    The fidelity between $\tau$ and $\mu$ can be bounded as $\fidelity(\tau,\mu) = \alpha \fidelity(\rho,\sigma) + 1 - \alpha \geq \alpha \kappa_1 + 1 - \alpha \geq \kappa_2$.

    For the third item, recall that the canonical Uhlmann isometry (where we have set the cutoff $\eta$ to $0$) for $(\ket{E},\ket{F})$ is defined as
    \[
        V = \sgn( \Tr_{\reg{A'}} ( \ketbra{E}{F}))
    \]
    where $\reg{A'}$ denotes the first $n+1$ qubits of $\ket{E},\ket{F}$. This is equal to
    \begin{align*}
        \sgn \Big ( \alpha \Tr_{\reg{A}} ( \ketbra{C}{D}) \ot \ketbra{0}{0} + (1 - \alpha) \ketbra{1^n}{1^n} \ot \ketbra{1}{1} \Big) = \sgn(\Tr_{\reg{A}} ( \ketbra{C}{D})) \ot \ketbra{0}{0} + \ketbra{1^n}{1^n} \ot \ketbra{1}{1}
    \end{align*}
    where $\reg{A}$ denotes the first $n$ qubits of $\ket{C},\ket{D}$. To conclude, note that $\sgn(\Tr_{\reg{A}} ( \ketbra{C}{D}))$ is the canonical Uhlmann isometry for $(\ket{C},\ket{D})$.
\end{proof}

\begin{lemma}[Reductions for $\DistUhlmann_\kappa$ for different fidelities $\kappa$]
    Let $\kappa: \N \to [0,1]$ be such that $1/p(n) \leq \kappa(n) \leq 1 - 1/p(n)$ for all $n$ for some polynomial $p(n)$. Then $\DistUhlmann_{\kappa}$ polynomial-time reduces to $\DistUhlmann_{1 - 1/p}$. 
\end{lemma}
\begin{proof}
    For every valid $\Uhlmann_\kappa$ instance $x = (1^n,C,D)$, let $y = (1^{2(n+1)},E,F)$ denote the valid $\Uhlmann_{1 - 1/p}$ instance given by the padding trick (\Cref{lem:padding-trick}), where $\alpha(n) = 1/p(n)$. The state $\sqrt{\alpha(n)} \ket{0} + \sqrt{1 - \alpha(n)} \ket{1}$ can be prepared with circuits of size $O(\log n)$ by the Solovay-Kitaev theorem, so by \cref{lem:padding-trick} $E$ and $F$ are also polynomial-sized (in $n$) circuits. Furthermore, given explicit descriptions of $C, D$ one can efficiently compute explicit descriptions of $E,F$.

    In order to prove the theorem, by \cref{def:reduction_distributional}, for every precision $r$ we need to find another precision $r'$ (which can be polynomial in $r$ and $n$) and a polynomial-time quantum query algorithm $A^*$ such that any $1/r'$-error average case instantiation (see \cref{def:avg_instantiation}) of $A^\DistUhlmann_{1-1/p}$ implements $\DistUhlmann_{1/p}$ with average-case error $1/r$.
    
    We define $A^* = (A^*_x)_x$ as follows. The circuit $A^*_x$ takes as input an $n$-qubit register $\reg{B}$ and initializes a single-qubit register $\reg{F}$ in the state $\ket{0}$.
    It then applies the $\DistUhlmann_{1-1/p}$ oracle for instance $y$ (whose description can be efficiently computed from $x$) on registers $\reg{F} \reg{B}$ and outputs the result.

    To show that this implements $\DistUhlmann_{1/p}$, let $r' = p(n)r$, and let $A^{\DistUhlmann_{1 - 1/p}}$ denote a $1/r'$-error average-case instantiation. 
    Concretely, let $V_y$ denote the (exact) Uhlmann partial isometry for instance $y$ and let $H = (H_y)_y$ denote a quantum algorithm that implements $\DistUhlmann_{1 - 1/p}$ with average-case error $1/r'$ and is used to instantiate the $\DistUhlmann_{1-1/p}$-oracle. This means there is a channel completion $\Phi_y$ of $V_y$ such that 
    \[
        \td \Big ((\id \otimes H_y)(\ketbra{E}{E}), (\id \otimes \Phi_y)(\ketbra{E}{E}) \Big) \leq \frac{1}{r'}\,.
    \]

    By the third item of \Cref{lem:padding-trick}, any channel completion $\Phi_y$ of $V_y$ can be turned into a channel completion of $\Xi_x$ of $U_x$, the $\Uhlmann_\kappa$ transformation corresponding to $(\ket{C},\ket{D})$. Define $\Xi_x(\rho) \deq \Tr_{\reg{G}}(\Phi_x(\rho \ot \ketbra{0}{0}_{\reg{G}}))$ where $\reg{G}$ denotes the last qubit. Let $\Pi$ denote the support onto $U_x$. Then $\Xi_x(\Pi \rho \Pi) = \Tr_{\reg{G}}(\Phi_x(\Pi \rho \Pi \ot \ketbra{0}{0}_{\reg{G}}))$. But notice that the state $\Pi \rho \Pi \ot \ketbra{0}{0}$ is contained in the support of $V_y$; therefore 
    \begin{align*}
        \Tr_{\reg{G}}(\Phi_x(\Pi \rho \Pi \ot \ketbra{0}{0})) = \Tr_{\reg{G}} \Big( V_y (\Pi \rho \Pi \ot \ketbra{0}{0}) V_y^\dagger \Big) = U_x \rho U_x^\dagger
    \end{align*}
    where we used the expression for $V_y$ given by \Cref{lem:padding-trick}. 
    Thus we can evaluate the performance of the instantiation $A^{\DistUhlmann_{1 - 1/p}}$ on the input $\ket{C}$:
    \begin{align*}
        &\td \Big ((\id \otimes A^{\DistUhlmann_{1 - 1/p}}_x)(\ketbra{C}{C}), \, (\id \otimes \Xi_x)(\ketbra{C}{C}) \Big) \\
        &= \td \Big ((\id \otimes H_y)(\ketbra{0}{0} \ot \ketbra{C}{C} \ot \ketbra{0}{0}), \, (\id \otimes \Phi_y)(\ketbra{0}{0} \ot \ketbra{C}{C}) \ot \ketbra{0}{0}\Big) \\
        &= \frac{1}{\alpha(n)}  \td \Big( (\id \otimes H_y)(P \ketbra{E}{E} P^\dagger) , \, (\id \otimes \Phi_y)(P \ketbra{E}{E} P^\dagger) \Big) \\
        &\leq \frac{1}{\alpha(n)}  \td \Big( (\id \otimes H_y)(\ketbra{E}{E}) , \, (\id \otimes \Phi_y)(\ketbra{E}{E}) \Big)\\
        &\leq \frac{1}{\alpha(n) r'} = \frac{1}{r}\,.
    \end{align*}
    In the second line, we expanded the definitions of the query circuit $A_x$ and the channel completion $\Xi_x$. In the third line, we define the projector $P = \ketbra{0}{0}$ which acts on the first qubit so that $\ket{0} \ket{C} \ket{0} = \frac{1}{\sqrt{\alpha(n)}} P \ket{E}$. In the fifth line we used the guarantees about the algorithm $H_y$ and our definitions of $\alpha(n), r'$.
\end{proof}

The padding trick shows that $\Uhlmann_\kappa$ and $\DistUhlmann_\kappa$ are equivalent in complexity whenever $\kappa$ is bounded away from $0$ or $1$ by an inverse polynomial. It could be that there are (at least) three different complexities between $\DistUhlmann_{\negl}$, $\DistUhlmann_{1/2}$, and $\DistUhlmann_{1-\negl}$. This leads to the following natural questions:
\begin{enumerate}
    \item What is the complexity of $\DistUhlmann_\kappa$ for negligibly small $\kappa$?
    \item What is the complexity of $\DistUhlmann_\kappa$ for $\kappa = 1 - \negl(n)$?
    \item Is the complexity of $\DistUhlmann_{1/2}$ the same as the complexity of $\DistUhlmann_{1 - \negl}$?
\end{enumerate}
In the next section we explore an approach to answering the last two questions.

\subsection{A polarization lemma for Uhlmann transformations?}
\label{sec:polarize}

The padding trick is reminiscent of a result in complexity theory and cryptography called the \emph{polarization lemma}. This is used to show complete problems for $\class{SZK}$ and $\class{QSZK}$~\cite{10.1145/636865.636868,watrous2002limits}, the analogous decision classes to $\avgUnitary{HVSZK}$. The complete problem for $\mathsf{QSZK}$ is $\textsc{QuantumStateDistinguishing}$~\cite{watrous2002limits}, where one is given a pair of quantum circuits $(C,D)$ that generate mixed states $(\rho,\sigma)$, and one has to decide whether $\td(\rho,\sigma) \leq 1/3$ (the ``yes'' case) or $\td(\rho,\sigma) \geq 2/3$ (the ``no'' case), promised that one is the case. The polarization lemma yields an efficient transformation from $(C,D)$ that produces two circuits $(C',D')$ generating mixed states $(\rho',\sigma')$ such that in the ``yes'' case, $\td(\rho',\sigma') \leq 2^{-n}$, whereas in the ``no'' case, $\td(\rho',\sigma') \geq 1 - 2^{-n}$. Thus, if one can distinguish between circuit descriptions in the highly polarized case (i.e., the two density matrices are either exponentially close or exponentially far), then one can efficiently distinguish between the mildly polarized case. 

The analogous statement in the unitary complexity setting would be to have an efficient polynomial-time reduction from (say) $\DistUhlmann_{1/2}$ to (say) $\DistUhlmann_{1 - 2^{-n}}$. That is, implementing the Uhlmann transformation for an instance $(1^n,C,D)$ where the reduced density matrices of $\ket{C},\ket{D}$ have fidelity $1/2$ can be efficiently reduced to implementing Uhlmann transformations for instances with high fidelity $1 - 2^{-n}$. We conjecture that such a transformation is possible.

\begin{conjecture}[Polarization for the Uhlmann Transformation Problem]
\label{conj:polarization}
    For all polynomials $p(n)$, there exists a polynomial-time reduction from $\DistUhlmann_{1/2}$ to $\DistUhlmann_{1 - 2^{-p(n)}}$.
\end{conjecture}

One might hope that with such a polarization lemma, one can show completeness of $\DistUhlmann_{1/2}$ for $\avgUnitary{HVSZK}$ (for some choice of completeness, soundness, and simulation error), which is like $\avgUnitary{HVPZK}$ except the completeness and simulation errors may not be zero. Showing \emph{hardness} of $\DistUhlmann_{1/2}$ for $\avgUnitary{HVSZK}$ is straightforward, using similar ideas to \Cref{thm:dist_uhlmann_pzk_hard}. However it is unclear how to use the conjectured polarization lemma to argue that $\DistUhlmann_{1/2}$ has a statistical zero-knowledge protocol; the trouble is that we don't know if $\avgUnitary{HVSZK}$ is closed under polynomial-time reductions (as discussed in \Cref{sec:qip-closure-under-ptime}). 

However, the conjectured polarization lemma \emph{does} imply a slightly weaker notion of completeness:

\begin{lemma}
    Let $c = 1 - 2^{-\poly(n,r)}$ and $s = \frac{1}{2}$. Assuming \Cref{conj:polarization}, for all inverse polynomials $\kappa(n)$, 
    \begin{enumerate}
        \item All problems in $\avgUnitary{HVSZK}_{c,s}$ are polynomial-time reducible to $\DistUhlmann_{\kappa}$, and
        \item $\DistUhlmann_{\kappa}$ is polynomial-time reducible to a problem in $\avgUnitary{HVSZK}_{c,s}$.
    \end{enumerate}
\end{lemma}
\begin{proof}
    We sketch the proof. The first item follows from the fact that $\avgUnitary{HVSZK}_{c,s}$ is polynomial-time reducible to $\DistUhlmann_{1 - 2^{-(n+r)}}$ using a very similar proof to \Cref{thm:dist_uhlmann_pzk_hard}, and $\DistUhlmann_{1 - 2^{-(n+r)}}$ is trivially reducible to $\DistUhlmann_{\kappa}$ for any inverse polynomial $\kappa(n)$. 

    The second item directly follows from the conjectured polarization lemma, and the fact that $\DistUhlmann_{1 - 2^{-(n+r)}}$ is contained in $\avgUnitary{HVSZK}$ using \Cref{prot:pzk_for_uhlmann}, and a similar analysis.
\end{proof}

We give evidence for \Cref{conj:polarization} by proving a \emph{weak} polarization lemma for Uhlmann transformations.

\begin{restatable}{theorem}{weakpolarization}\label{thm:polarization}
Let $p(n)$ be a polynomial. Suppose there is a polynomial-time algorithm $Q$ that implements $\DistUhlmann_{1 - 2^{-p(n)}}$ with average-case error at most $1/32$. Then for all $0 < \eps < 1$ there exists a quantum algorithm $A = (A_x)_{x \in \bits^*}$ that runs in $n^{O(1/\eps)}$ time, makes queries to the unitary purification of $Q$ and its inverse, such that for all valid instances $x = (1^n,C,D)$ of $\Uhlmann_{1/2}$, 
    \[
        \fidelity((\id \otimes A_x)(\ketbra{C}{C}),\ketbra{D}{D}) \geq \frac{1}{2} - \eps~.
    \]
\end{restatable}
We prove \Cref{thm:polarization} in \Cref{sec:polarization}. We compare \Cref{thm:polarization} with \Cref{conj:polarization}. The main difference is that the algorithm $A$ for $\Uhlmann_{1/2}$ instances runs in time that scales \emph{exponentially} with the (inverse) precision parameter $1/\eps$, whereas \Cref{conj:polarization} posits that there is an implementation of $\DistUhlmann_{1/2}$ that runs in time $\poly(n,1/r)$ where $r$ is the precision parameter (here we think of $\eps = 1/r$). Finally, one might also wonder why the conclusion of \Cref{thm:polarization} is not expressed as ``$A$ implements $\DistUhlmann_{1/2}$ with error $\eps$''. The reason is that while $A$ maps $\ket{C}$ to having near-optimal overlap with $\ket{D}$, it is no longer clear that $A$ must be close to the canonical Uhlmann transformation corresponding to $(C,D)$; \Cref{prop:operational-avg-case-uhlmann} which connects the two notions only works in the \emph{high $\kappa$ regime}.\footnote{This is related to a question of whether Uhlmann transforms are \emph{rigid}; must near-optimal Uhlmann transforms for a pair of states be close to the corresponding canonical Uhlmann transform?}

Nonetheless, \Cref{thm:polarization} does achieve a nontrivial guarantee: being able to perform Uhlmann transforms for $\Uhlmann_{1 - 2^{-\poly(n)}}$ instances with some fixed constant error (i.e., $1/32$) can be efficiently converted into Uhlmann transformations for $\Uhlmann_{1/2}$ instances with arbitrarily small constant error.

 \section{Complexity of the Succinct Uhlmann Transformation Problem}
\label{sec:structural-succinct-uhlmann}

In this section, we show that the $\DistSuccinctUhlmann_1$ problem is complete for both $\avgUnitaryPSPACE$ and $\avgUnitaryQIP$. We do this by establishing the following statements:
\begin{enumerate}
    \item $\DistSuccinctUhlmann_1 \in \avgUnitaryQIP$ (\Cref{lem:succuhl_in_qip}).
    \item $\avgUnitaryPSPACE$ polynomial-time reduces to $\DistSuccinctUhlmann_1$ (\Cref{lem:succuhl_hard_for_pspace}).
    \item $\avgUnitaryQIP \subseteq \avgUnitaryPSPACE$ (\Cref{lem:qip_in_pspace}).
\end{enumerate}
These statements together imply the desired completeness results, as well as the equality $\avgUnitaryQIP = \avgUnitaryPSPACE$ (\Cref{thm:avg-qip-pspace-equality}). This also allows us to show that $\avgUnitaryQIP$ is closed under polynomial-time reductions, which addresses a question raised in \Cref{sec:poly_space_reductions_discussion}.

\subsection{Interactive synthesis of succinct Uhlmann transformations}
\label{sec:suc_uhlmann-qip-complete}

We first show that $\DistSuccinctUhlmann_1 \in \avgUnitaryQIP_{c,s}$ for $c = 1 - 2^{-\poly(n,r)}$ and $s = 1/2$.
The protocol mirrors that of \cref{prot:pzk_for_uhlmann} (the $\avgUnitary{HVPZK}$ protocol for $\DistUhlmann_1$), except that the circuits $C$ and $D$ are no longer polynomial size since now they are specified succinctly.
As a result, the polynomial time verifier can no longer efficiently prepare copies of the state $\ket{C}$ and can no longer directly implement the unitary $D^{\dagger}$ to check that the Uhlmann transformation was applied correctly; these were critical steps in \cref{prot:pzk_for_uhlmann}. 

To address the first issue, we leverage the $\statePSPACE \subseteq \stateQIP$ result of~\cite{rosenthal2022interactive,rosenthal2024efficient}, so by interacting with the prover, the verifier approximately synthesizes the input state $\ket{C}$.
To solve the second issue, we show that the verifier can approximately perform the measurement $\{\proj{D}, \id - \proj{D} \}$ using copies of $\ket{D}$ as a resource; again we use $\statePSPACE = \stateQIP$ to show that the verifier can prepare copies of $\ket{D}$ with the help of the prover. We describe these solutions in more detail next.

\paragraph{Interactive state synthesis.} First we recall Rosenthal's interactive state synthesis protocol~\cite{rosenthal2024efficient} (improving upon the state synthesis protocol of~\cite{rosenthal2022interactive}), which can synthesize any state sequence $(\ket{\psi_x})_x \in \statePSPACE$. We describe this result at a high level (for formal details see~\cite{rosenthal2024efficient}): for every $\statePSPACE$ state sequence $(\ket{\psi_x})_x$ there exists a polynomial-time quantum verifier $V = (V_{x, q})_{x, q}$ such (a) there exists an honest prover $P^*$ that is accepted by the verifier with probability $1$ (\emph{perfect completeness}), (b) after interacting with the honest prover $P^*$, conditioned on accepting, the verifier outputs a state that is at most $2^{-(n+q)}$-close to $\ket{\psi_x}$ (\emph{honest closeness})\footnote{We note that this honest closeness property is not part of our definition of $\stateQIP$ in \Cref{def:stateQIP}, but it is an additional property guaranteed by~\cite{rosenthal2024efficient}.}, and (c) for all provers $P$ that are accepted with probability at least $\frac{1}{2}$, the output register of the verifier is close to $\ket{\psi_x}$ within $1/q$ in trace distance (\emph{soundness}). 

In what follows we utilize as a subroutine the interactive state synthesis protocol for the sequence $(\ket{C})_{\hat C}$  which is indexed by all succinct descriptions $\hat C$ of a unitary circuit $C$ and $\ket{C}$ is the corresponding output state of the circuit (given all zeroes). It is straightforward to see that $(\ket{C}) \in \statePSPACE$, and therefore there is a $\stateQIP$ protocol to synthesize the state family.

\paragraph{Approximate reflections from copies of a state.} To perform the measurement $\{ \proj{D}, \id - \proj{D} \}$ on a state $\rho$, it suffices to initialize an ancilla qubit $\ket{+}$, and controlled on the ancilla qubit perform the reflection $\id - 2 \proj{D}$ on $\rho$. By performing a Hadamard on the ancilla qubit and measuring it, the desired measurement $\{ \proj{D}, \id - \proj{D} \}$  is performed. The (controlled) reflection $\id - 2\proj{\psi}$ can be approximately performed given copies of the state $\proj{D}$, via the following result from Ji, Liu, Song~\cite[Theorem 5]{JLS18}.

\begin{theorem}[Approximate reflections from copies of a state~\cite{JLS18}]
\label{thm:appx-reflections}
Let $\ket{\psi}$ be a state and let $A^R$ be an algorithm that takes as input a register $\reg{A}$ and makes $q$ (possibly controlled) queries to the reflection oracle $R = \id - 2\proj{\psi}$. Then for all integers $\ell > 0$ there exists an algorithm $B$ that takes as input a register $\reg{A}$ in addition to $\ket{\psi}^{\otimes \ell}$ and makes no queries to $R$, such that for all input states $\rho_{\reg{EA}}$ 
\[
    \td \Big ( (\id_{\reg{E}} \otimes A^R)(\rho_{\reg{E} \reg{A}}) \, , \, (\id_{\reg{E}} \otimes B)(\rho_{\reg{E} \reg{A}} \otimes \proj{\psi}^{\otimes \ell} ) \Big) \leq \frac{2q}{\sqrt{\ell + 1}}~.
\]
Furthermore, if $A^R$ is an algorithm such that performs the measurement $\{ \proj{\psi}, \id - \proj{\psi} \}$ on the input register $\reg{A}$, then $B(\proj{\psi} \otimes \proj{\psi}^{\otimes \ell})$ accepts with probability $1$.
\end{theorem}
The ``furthermore'' part of the theorem statement above comes from examining the algorithm $B$ described in~\cite{JLS18}, which for every query to $R$ performs a reflection about the symmetric subspace on $\ell+1$ registers. 

\medskip
\vspace{5pt}

Augmenting the protocol from \Cref{prot:pzk_for_uhlmann} with the interactive state synthesis and the approximate reflections protocols enables us to prove the following.
\begin{lemma} \label{lem:succuhl_in_qip}
     $\DistSuccinctUhlmann_1 \in \avgUnitaryQIP_{c,s}$ for $c = 1 - 2^{-(n+r)}$, $s = 1/2$.
\end{lemma}
\begin{proof}

We present in \Cref{prot:qip_for_uhlmann} an $\avgUnitaryQIP$ protocol for $\DistSuccinctUhlmann_1$ with completeness $1 - 2^{-(n+r)}$ and soundness $\frac{1}{2}$, where $r$ is the precision parameter.

\begin{longfbox}[breakable=false, padding=1em, margin-top=1em, margin-bottom=1em]
\begin{protocol} {\bf Interactive protocol for $\DistSuccinctUhlmann_1$} \label{prot:qip_for_uhlmann} 
\end{protocol}
\noindent \textbf{Instance: } A valid $\SuccinctUhlmann_{1}$ instance $x = (1^n,\hat{C},\hat{D})$ specifying a succinct description of a pair of circuits $(C,D)$ and precision $r$. \\
\noindent \textbf{Input: } Register $\reg{B}_0$. 

\begin{enumerate}
    \item (\emph{\textbf{State synthesis}}). Let $m = 64r^2 + 1$, $\ell = 2^{20} \cdot r^6$.  Run the $\stateQIP$ protocol from~\cite{rosenthal2024efficient} to synthesize $\ket{C}^{\otimes m} \otimes \ket{D}^{\otimes \ell}$ with precision parameter $q = 8rn$. Let $\reg{A}_1 \reg{B}_1 \cdots \reg{A}_m \reg{B}_m$ denote the registers for $\ket{C}^{\otimes m}$. If the $\stateQIP$ protocol rejects, reject.

    \item (\emph{\textbf{Uhlmann transformation testing}}). Sample $i^* \in [m]$ uniformly at random. \begin{enumerate}
                 \item For $i = 1$ though $m$:
        \begin{enumerate}
            \item If $i \neq i^*$, send register $\reg{B}_i$ to the prover and receive $\reg{B}_i$ back.
\item If $i = i^*$, send $\reg{B}_0$ to the prover and receive $\reg{B}_0$ back.
\end{enumerate}
        \item For $i \neq i^*$, use the approximate reflections protocol of \Cref{thm:appx-reflections} with the states $\ket{D}^{\otimes \ell}$ synthesized earlier to perform the measurement $\{ \proj{D},\id - \proj{D} \}$ on registers $\reg{A}_i \reg{B}_i$. If the $\id - \proj{D}$ outcome is obtained for any $i$, then reject. \item Otherwise, accept and output register $\reg{B}_0$.

    \end{enumerate}
\end{enumerate}

\end{longfbox}

We show that the verifier described in \cref{prot:qip_for_uhlmann} satisfies the required properties of $\avgUnitaryQIP$ protocols. First, the verifier runs in $\poly(n,r)$ time. This uses the fact that the $\stateQIP$ protocol and approximate reflection algorithm runs in $\poly(n,r)$ time, and the succinct descriptions of $\ket{C}^{\ot m}$ and $\ket{D}^{\ot \ell}$ are $\poly(n,r)$-sized in the lengths of the succinct descriptions $\hat C$ and $\hat D$. We now prove that the verifier satisfies the completeness and soundness properties. 

\end{proof}

\begin{lemma}[Completeness]\label{lem:sucuhl_qip_completeness}
    For all valid $\SuccinctUhlmann_1$ instances $x = (1^n,\hat C,\hat D)$, precision parameter $r \in \mathbb{N}$, and sufficiently large $n$, there exists an honest prover $P^*$ for \cref{prot:qip_for_uhlmann} satisfying
    \begin{equation*}
        \Pr[V_{x, r}(\ket{C}) \interact P^*] \geq 1 - 2^{-(n+r)}\,.
    \end{equation*}
\end{lemma}

\begin{proof}

    We define an honest prover $P^*$ for $\DistSuccinctUhlmann_{1}$ as follows. Let $x = (1^n,\hat{C},\hat{D})$ denote the $\SuccinctUhlmann_1$ instance. In the state synthesis stage, $P^*$  implements the honest prover in the $\stateQIP$ protocol for synthesizing $\ket{C}^{\otimes m} \otimes \ket{D}^{\otimes \ell}$, which has completeness $1$~\cite{rosenthal2024efficient}.  In the Uhlmann transformation testing stage, in each of the $m$ rounds $P^*$ applies the canonical Uhlmann transformation $U_x$ on the given register $\reg{B}$ (just like in \Cref{sec:structural-uhlmann}). 

    We ignore the register $\reg{B}_0$ for now, because it doesn't affect the acceptance probability of the verifier. By the honest closeness property of the protocol described in~\cite{rosenthal2024efficient}, the state of the verifier after the first stage is $2^{-(n+q)}$-close to $\ket{C}^{\otimes m} \otimes \ket{D}^{\otimes \ell}$. Thus the state of the verifier after Step 2(a) of the protocol is $2^{-(n+q)}$-close to $\ket{D}^{\otimes m-1} \otimes \ket{C} \otimes \ket{D}^{\otimes \ell}$ where the first $m-1$ copies of $\ket{D}$ are from applying the honest prover to registers $\reg{A}_i \reg{B}_i$ for $i \neq i^*$, the lone copy of $\ket{C}$ is in register $\reg{A}_{i^*} \reg{B}_{i^*}$ (which was not sent over to the prover), and the last $\ell$ copies of $\ket{D}$ are from the state synthesis stage. By the ``furthermore'' part of \Cref{thm:appx-reflections}, each of the (approximate) $\{ \proj{D}, \id - \proj{D} \}$ measurements will succeed with probability at least $1 - 2^{-(n+q)}$. Thus the overall probability of acceptance is at least $1 - m2^{-(n+q)} \geq 1 - 2^{-(n+r)}$ for sufficiently large $n$ (this uses our choices of $m,q$).

\end{proof}

\begin{lemma}[Soundness]\label{lem:sucuhl_qip_soundness}
    For all valid $\SuccinctUhlmann_1$ instances $x = (1^n,\hat C,\hat D)$ and precision $r \in \mathbb{N}$, for sufficiently large $n$, for all quantum provers $P$, there exists a channel completion $\Phi_{x}$ of $U_{x}$ such that
    \begin{equation*}
        \text{if } \quad \pr {V_{x, r}(\ket{C}) \interact P \text{ accepts}} \geq \frac{1}{2} \qquad \text{then} \qquad \td(\sigma, (\Phi_x \ot \id)(\ketbra{C}{C})) \leq \frac{1}{r}\,,
    \end{equation*}
    where $\sigma$ denotes the output of $V_{x, r}(\ket{C})\interact P$ conditioned on the verifier $V_{x}$. \end{lemma}
\begin{proof}
Since $x = (1^n,\hat{C},\hat{D})$ is a $\SuccinctUhlmann_1$ instance, by \Cref{prop:canonical-uhlmann-properties} $(\Phi_{x} \otimes \id)(\ketbra{C}{C}) = \proj{D}$ for all channel completions $\Phi_{x}$ of $U_x$, so it suffices to show that conditioned on accepting, the verifier outputs a state within $1/r$ of $\ket{D}$ in trace distance. 

Let $\reg{A}_0 \reg{B}_0$ denote the registers for the input state $\ket{C}$, so that $\reg{A}_0$ is the reference register that is not accessed by either the prover or the verifier. Let $P$ be a prover that is accepted with probability at least $1/2$. This means that the state synthesis stage is accepted with probability at least $1/2$, and the Uhlmann transformation testing stage is accepted with probability at least $1/2$ (conditioned on the first stage accepting). By the soundness property of the $\stateQIP$ protocol, 
the state of the verifier (plus reference system $\reg{A}_0$), conditioned on acceptance of the state synthesis stage, is $1/q$-close to $\ket{C}^{\otimes (m+1)} \otimes \ket{D}^{\otimes \ell}$. 

Let $\cal{P}$ denote the prover channel acting on registers $\reg{B}_0 \cdots \reg{B}_m$. Let $\cal{B}$ denote the channel corresponding to the approximate reflections protocol from \Cref{thm:appx-reflections} used in Step 2(b), that traces out the copies of $\ket{D}^{\otimes \ell}$ at the end. Let $\cal{A}$ denote the channel that performs the \emph{exact} reflections instead of the approximate ones, which \Cref{thm:appx-reflections} guarantees is $2(m-1)/\sqrt{\ell+1}$-close to $\cal{B}$. 

We consider a sequence of hybrids. Define $\rho_0$ to be the state of the verifier and registers $\reg{A}_0 \reg{B}_0$ conditioned on the first stage accepting. Define $\rho_1 := \proj{C}^{\otimes m+1} \otimes \proj{D}^{\otimes \ell}$. Define
\[
    \sigma_b = \cal{B} \Big( \cal{P} (\rho_b) \Big)
\]
for $b \in \{0,1\}$. Note that $\sigma_0$ is the state of the verifier at the end of the protocol (but before conditioning on acceptance in the second stage). Define
\[
    \sigma_2 = \cal{A} \Big( \cal{P} (\proj{C}^{\otimes m+1}) \Big)~.
\]
Since $\td(\rho_0,\rho_1) \leq 1/q$, this means $\td(\sigma_0,\sigma_1) \leq 1/q$. 
By \Cref{thm:appx-reflections}, $\td(\sigma_1,\sigma_2) \leq 2(m-1)/\sqrt{\ell+1}$. Since the probability of acceptance in the state $\sigma_0$ is at least $1/2$, the probability of acceptance in the state $\sigma_2$ is at least $1/2 - \delta$ for
\[
    \delta := \frac{1}{q} + \frac{2(m-1)}{\sqrt{\ell+1}} \leq \frac{1}{4r}
\]
by our choices of $q,m,\ell$. 

Let $\tilde{\sigma}_b$ for $b \in \{0,1,2\}$ denote the states conditioned on acceptance of the second stage. The soundness analysis of the $\avgUnitary{HVPZK}$ protocol in \Cref{lem:Uhlmann_pzk_qip_soundness}  shows that the $\reg{A}_0 \reg{B}_0$ registers of $\tilde{\sigma}_2$ are $\sqrt{16/m}$-close to $\ket{D}$. Furthermore,~\cite[Lemma B.1]{coudron2013infinite} implies that 
\[
    \td(\tilde{\sigma}_0,\tilde{\sigma}_2) \leq \frac{\td(\sigma_0,\sigma_2)}{\max \{p_0,p_2 \}} \leq 2\td(\sigma_0,\sigma_2) \leq 2\delta~.
\]
where $p_b$ denotes the probability of acceptance in $\sigma_b$, respectively. Thus, by our choice of $m,\ell,q$, we have
\[
    \td \Big( (\tilde{\sigma}_0)_{\reg{A}_0 \reg{B}_0} \, , \, \proj{D} \Big) \leq 2\delta + \sqrt{\frac{16}{m}} \leq \frac{1}{r}~.
\]
Note that $(\tilde{\sigma}_0)_{\reg{A}_0 \reg{B}_0}$ is precisely the state of the reference register $\reg{A}_0$ and the register output by the verifier conditioned on accepting. This concludes the proof of soundness.

\end{proof}

\subsection{Hardness for unitary polynomial space}
\label{sec:suc_uhl_avgpspace_complete}

We now show that $\DistSuccinctUhlmann_1$ is hard for $\avgUnitaryPSPACE$.

\begin{lemma} \label{lem:succuhl_hard_for_pspace}
$\avgUnitaryPSPACE$ polynomial-time reduces to $\DistSuccinctUhlmann_1$.
\end{lemma}
\begin{proof}
We show that any distributional unitary synthesis problem $(\usynth{U} = (U_x)_x, \Psi = (\psi_x)_x) \in \avgUnitaryPSPACE$ can be reduced to implementing the Uhlmann transformation corresponding to a $\SuccinctUhlmann_1$ instance.
The idea for this is natural: to implement $U_x$ on $\psi_x$, we can simply perform the Uhlmann transformation corresponding to the pair of states $(\ket{\psi_x},(\id \otimes U_x) \ket{\psi_x})$. 

There are, however, some subtleties in making this to work. Instances of $\SuccinctUhlmann_1$ are tuples $y = (1^n,\hat C,\hat D)$ where $\hat C,\hat D$ are succinct descriptions of \emph{unitary} circuits $C,D$, and furthermore the states $\ket{C}, \ket{D}$ have \emph{identical} reduced density matrices on the first register. Even though $(\usynth{U},\Psi) \in \avgUnitaryPSPACE$, \emph{a priori} the unitary $U_x$ can only be \emph{approximately} implemented and the state $\psi_x$ can only be \emph{approximately} prepared. Therefore, it seems that we can only compute succinct descriptions of circuits that approximately prepare $\ket{\psi_x}$ and $(\id \otimes U_x) \ket{\psi_x}$. Additionally, these circuits are not necessarily unitary (since the circuits used to synthesize $\ket{\psi_x}$ and $U_x$ may have measurements, resets, and other non-unitary operations). Thus, these succinct descriptions do not directly constitute a $\SuccinctUhlmann_1$ instance. 

We give a sketch of how we handle these issues. From the instance $x$ of $(\usynth{U},\Psi)$ we compute an instance $y = (1^n, \hat C, \hat D)$ of $\SuccinctUhlmann_1$ where $C$ is a unitary circuit preparing a state close to $\ket{\psi_x}$, and $D$ is a unitary circuit preparing the state $(\id \otimes A) \ket{C}$, with $A$ being a another unitary circuit such that $(\id \otimes A) \ket{\psi_x}$ is close to $(\id \otimes U_x) \ket{\psi_x}$.  Let $W$ denote the canonical Uhlmann transformation corresponding to instance $y$. Then
\begin{align*}
    (\id \otimes W) \ket{\psi_x} &\approx (\id \otimes W) \ket{C} & \tag{$\ket{\psi_x} \approx \ket{C}$} \\
    &= \ket{D} = (\id \otimes A) \ket{C}& \tag{\Cref{prop:canonical-uhlmann-properties}} \\
    &\approx (\id \otimes A) \ket{\psi_x} & \tag{$\ket{\psi_x} \approx \ket{C}$} \\
    &\approx (\id \otimes U_x) \ket{\psi_x}~. & \tag{Definition of $A$}
\end{align*}
Thus, implementing the canonical Uhlmann transformation $W$ corresponding to $y$ is sufficient for approximately implementing $U_x$ on $\ket{\psi_x}$. 

The key behind making this sketch work is the following ``algorithmic Uhlmann's theorem'' of Metger and Yuen~\cite[Theorem 7.4]{metger2023stateqip}:

\begin{theorem}[Algorithmic Uhlmann's theorem]
\label{thm:algorithmic-uhlmann}
Let $\Psi_1 = (\ket{\psi^{(1)}_x})_x, \Psi_2 = (\ket{\psi^{(2)}_x})_x$ be in families of bipartite pure states in $\statePSPACE$, where for each $x$, the states $\ket{\psi^{(1)}_x}$ and $\ket{\psi^{(2)}_x}$ have the same number of qubits. Then there exists a $\poly(|x|,r)$-space family of unitary circuits $K = (K_{x,r})_{x,r}$ such that 
\[
    \Big \| K_{x,r} \ket{\psi^{(1)}_x} \ket{0 \cdots 0} - \ket{\psi^{(2)}_x} \ket{0 \cdots 0} \Big \|^2 \leq 2(1 - \sqrt{\fidelity(\sigma^{(1)}_x,\sigma^{(2)}_x)} ) + 1/r ~,
\]
where the $\ket{0 \cdots 0}$ denote the appropriate number of ancilla qubits, and $\sigma^{(b)}_x$ denotes the reduced density matrix of $\ket{\psi^{(b)}_x}$ on the first half of qubits.
\end{theorem}

A useful corollary of this is that every $\statePSPACE$ family of states can be prepared by a polynomial-space \emph{unitary} circuit that uses some ancillas. 

\begin{corollary}[Unitary preparation of $\statePSPACE$ families]
\label{cor:unitary-state-preparation}
Let $(\ket{\psi_x})_x \in \statePSPACE$. Then there exists a $\poly(|x|,r)$ unitary circuit family $K = \{K_{x,r} \}_{x,r}$ such that
\[
    \Big \| K_{x,r} \ket{0 \cdots 0} - \ket{\psi_x} \ket{0 \cdots 0} \Big \|^2 \leq \frac{1}{r}~.
\]
\end{corollary}
\begin{proof}
    This follows from \Cref{thm:algorithmic-uhlmann} by considering the Uhlmann transformation problem between the all zeroes state $\ket{0 \cdots 0}\ket{0 \cdots 0}$ and $\ket{0 \cdots 0} \ket{\psi_x}$. Clearly this pair of states is a fidelity-$1$ Uhlmann instance. 
\end{proof}

In what follows, fix a parameter $q$ and an instance $x$. Since $\Psi \in \statePSPACE$, by \Cref{cor:unitary-state-preparation} there exists a $\poly(|x|,q)$-space circuit $C$ such that
\[
    \Big \| C \ket{0 \cdots 0} - \ket{\psi_x} \ket{0 \cdots 0} \Big \|^2 \leq \frac{1}{q}~.
\]

Next, since $(\usynth{U},\Psi) \in \avgUnitaryPSPACE$, this means that for every $x$ there is a unitary $\tilde{U}_{x}$ (acting on possibly more qubits than $\ket{\psi_x}$) such that 
\[
    \ket{\varphi_x} := (\id \otimes \tilde{U}_x) \ket{\psi_x} \ket{0 \cdots 0}
\]
is such that the marginal of $\ket{\varphi_x}$ is $1/q$-close to a channel completion of $U_x$ applied to the second half of $\ket{\psi_x}$, and furthermore the state sequence $(\ket{\varphi_x})_x$ is in $\statePSPACE$. 

Notice that the pair of states $(\ket{\psi_x} \ket{0 \cdots 0},\ket{\varphi_x})$ forms a fidelity-$1$ instance of the Uhlmann transformation problem. Thus, by the algorithmic Uhlmann's theorem (\Cref{thm:algorithmic-uhlmann}), there exists a $\poly(|x|,q)$-space \emph{unitary} algorithm $A$ such that
\[
    \Big \| (\id \otimes A) \ket{\psi_x} \ket{0 \cdots 0} - \ket{\varphi_x} \Big \|^2 \leq \frac{1}{q}~.
\]  

Define the $\poly(|x|,q)$-space unitary circuit $D$ that, starting with all zeroes, does the following:
\begin{enumerate}
    \item Apply the circuit $C$, then
    \item Apply the circuit $A$ on the second half of the resulting state. 
\end{enumerate}
There are succinct descriptions $\hat C,\hat D$ of $C,D$, respectively, of $\poly(|x|,q)$-size. This follows from the fact that $C$ and $A$ are space-uniform\footnote{The succinct descriptions $\hat C, \hat D$ correspond to polynomial-\emph{time} Turing machines that output the descriptions of polynomial-\emph{space} circuits that compute the circuits $C,D$, and run them.}.

Therefore by construction the instance $y = (1^n, \hat C, \hat D)$ is a valid $\SuccinctUhlmann_1$ instance. By the above argument, the canonical Uhlmann transformation $W$ corresponding to $y$ will map $\ket{\psi_x} \ket{0 \cdots 0}$ to be within $O(1/q)$ of $(\id \otimes \tilde{U}_x) \ket{\psi_x} \ket{0 \cdots 0}$, as desired. 

Thus one can implement the distributional unitary synthesis problem $(\usynth{U},\Psi)$ as follows: given instance $x$, precision parameter $r$, and input register $\reg{B}$, set $q = O(r)$, compute the instance $y = (1^n, \hat C, \hat D)$ which depends on $x, q$, and call the $\DistSuccinctUhlmann_1$ oracle on instance $y$, precision $q$, and input register $\reg{B}$ appended with sufficiently many zeroes. Every average-case instantiation of the $\DistSuccinctUhlmann_1$ oracle will implement (a channel completion of) $W$ up to error $1/q$. Thus result has error $O(1/q)$, which is at most $1/r$ if $q$ is set large enough. This reduction can be computed in $\poly(|x|,r)$ time.

\end{proof}

\subsection{\titleavgUnitaryQIP $=$ \titleavgUnitaryPSPACE}

We first show that unitaries synthesizable using interactive protocols can also be synthesized in polynomial space, at least in the average-case setting. In what follows, recall that $\negl(n,r)$ means that this is a negligible function of $n+r$ (e.g., $2^{-(n+r)}$). 

\begin{lemma}
\label{lem:qip_in_pspace}
    For all $c,s:\N \times \N \to [0,1]$ such that $c = 1- \negl(n,r)$ and $c-s \geq \delta$ for some inverse polynomial function $\delta$, we have that $\avgUnitaryQIP_{c,s} \subseteq \avgUnitaryPSPACE$.
\end{lemma}
\begin{proof}

Let $(\usynth{U},\Psi) \in \avgUnitaryQIP$. There exists a verifier $V = (V_{x,r})_{x,r}$ satisfying the properties of \Cref{def:avgUnitaryQIP}. By the completeness property, there exists an honest prover $P^*$ that is accepted with probability at least $c(n,r) = 1 - \negl(n,r)$. Let $p(n,r)$ be such that $c(n,r) - s(n,r) \geq 1/p(n,r)$ for all sufficiently large $n,r$. 

If we can argue that there is a prover $P$ that can be computed in quantum polynomial space that is accepted with probability at least $1 - \negl(n,r) - \frac{1}{4r}$ by $V$, then the following is a $\avgUnitaryPSPACE$ algorithm for $(\usynth{U},\Psi)$: given an instance $x$, precision parameter $r$, and an input register $\reg{B}_0$, it runs the verifier $V$ on instance $x$, precision parameter $2r$, and input register $\reg{B}_0$. The $\avgUnitaryPSPACE$ algorithm simulates the prover as needed. At the end of the protocol, the verifier will accept with probability at least $1 - \negl(n,r) - \frac{1}{4r} \geq 1 - \frac{1}{2r} \geq \frac{1}{2}$ for large enough $n,r$. By the soundness of the verifier $V$, its output conditioned on acceptance is $1/2r$-close to the desired output. Since the acceptance probability is very high (at least $1 - 1/2r$ for large enough $n,r$), the output of the verifier (even without conditioning on acceptance) is $1/r$-close to the desired output. 

All that remains is to argue that there is a polynomial-space simulation of the prover $P^*$. This essentially follows from~\cite[Theorem 7.1]{metger2023stateqip}:

\begin{theorem}[Theorem 7.1 of~\cite{metger2023stateqip}]
Let $m(n)$ be a polynomial. Let $(V_x)_{x \in \bits^*}$ denote\footnote{Technically, Theorem 7.1 in~\cite{metger2023stateqip} considers verifiers indexed by integers $n$, rather than strings $x$ as we have written here, but inspection of the proof shows that the string-indexed version holds also.} a $m(|x|)$-round, $\poly(|x|)$-space verifier that starts with the all zeroes state. Let $\omega_x^*$ denote the optimal acceptance probability over all provers that interact with the verifier. Then for all polynomials $q(|x|)$ there exists a prover $P$ whose actions are computable in $\poly(|x|,q(|x|))$-space whose acceptance probability is at least $\omega_x^* - 1/q(|x|)$. 
\end{theorem}

Although this theorem doesn't have an explicit $r$ parameter, we can apply it to our verifier $V = (V_{x,r})_{x,r}$ by treating it as a sequence $(V_y)_{y \in \bits^*}$ where $y$ encodes $(x,1^r)$ and is of length $O(|x| + r)$, and furthermore the verifier initializes its input to the distributional state $\psi_x$. Since $\Psi \in \statePSPACE$, it can be prepared in polynomial space, and therefore $V$ corresponds to a $\poly(|y|)$-space verifier as required by the theorem. 
\end{proof}

We are now prove the main result of the section.
This answers an average-case version of an open problem raised in \cite{rosenthal2022interactive,metger2023stateqip}, namely whether $\unitary{QIP} = \unitaryPSPACE$, and is one of the first non-trivial results on relations between unitary complexity classes.

\begin{theorem}\label{thm:avg-qip-pspace-equality}
$\avgUnitaryPSPACE = \avgUnitaryQIP$.
\end{theorem}
\begin{proof}

We put everything together to prove $\avgUnitaryPSPACE = \avgUnitaryQIP$. First, \Cref{lem:qip_in_pspace} directly shows that  $\avgUnitaryQIP \subseteq \avgUnitaryPSPACE$. For the reverse inclusion, let $(\usynth{U},\Psi) \in \avgUnitaryPSPACE$. The proof of \Cref{lem:succuhl_hard_for_pspace} shows that every instance $x$ of the distributional unitary synthesis problem $(\usynth{U},\Psi)$ and precision parameter $r$ can be mapped in $\poly(|x|,r)$-time to an instance $y = (1^n,\hat C,\hat D)$ of $\SuccinctUhlmann_1$ such that $\ket{C} \approx_{1/6r} \ket{\psi_x} \ket{0 \cdots 0}$ and $\ket{D} \approx_{1/6r} (\id \otimes \tilde{U}_x) \ket{\psi_x} \ket{0 \cdots 0}$ where $\tilde{U}_x$ is a Stinespring dilation of a channel that is close to a channel completion of $U_x$. 

By \Cref{lem:succuhl_in_qip}, there is an interactive protocol for synthesizing the canonical Uhlmann transformation $W$ corresponding to $y$. This interactive protocol can be run with instance $y$, precision parameter $q = 6r$, and half of $\ket{\psi_x}$ as input rather than $\ket{C}$. Let $V = (V_{y,q})_{y,q}$ denote the verifier and let $P$ denote a prover that is accepted with probability at least $1/2$ when interacting with $V_{y,q}$ and input $\ket{\psi_x}$. Before conditioning on acceptance, we have that the output states of $V_{y,q}(\ket{\psi_x}) \interact P$ and $V_{y,q}(\ket{C}) \interact P$ are $1/6r$ close. Conditioning on acceptance only multiplies the closeness by at most $2$, so the outputs are $2/6r$-close. However, by the soundness of the protocol, $V_{y,q}(\ket{C}) \interact P$ conditioned on acceptance is $1/6r$-close to $\ket{D}$, which by definition is $1/6r$-close to $(\id \otimes \tilde{U}_x) \ket{\psi_x} \ket{0 \cdots 0}$. Adding up everything together, we have that conditioned on acceptance, the output of $V_{y,q}(\ket{\psi_x}) \interact P$ is at $5/6r$-close to $(\id \otimes \tilde{U}_x) \ket{\psi_x} \ket{0 \cdots 0}$. After tracing out some qubits, the resulting state is $1/r$-close to a channel completion of $U_x$ applied to $\ket{\psi_x}$. This concludes the proof of the soundness property.

The completeness of the protocol directly follows from \Cref{lem:sucuhl_qip_completeness}. Therefore this implies that $(\usynth{U},\Psi) \in \avgUnitaryQIP$, as desired.

\end{proof}

We record some easy corollaries of \Cref{thm:avg-qip-pspace-equality}. First, as discussed in \Cref{sec:poly_space_reductions_discussion}, the closure of the interactive unitary synthesis classes under polynomial-time reductions is not \emph{a priori} obvious, but follows easily from $\avgUnitaryPSPACE=\avgUnitaryQIP$:

\begin{corollary}
    $\avgUnitaryQIP$ is closed under polynomial-time reductions.
\end{corollary}
\begin{proof}
    This follows since $\avgUnitaryPSPACE$ is closed under polynomial-time reductions (\Cref{lem:unitary_pspace_closed}). 
\end{proof}

The next corollary records the centrality of the succinct Uhlmann transformation problem for these two unitary complexity classes.

\begin{corollary}
    $\DistSuccinctUhlmann_1$ is complete for $\avgUnitaryPSPACE$ and $\avgUnitaryQIP$.
\end{corollary}

Now, some questions for future work. The worst-case version of \Cref{thm:avg-qip-pspace-equality} remains open:
\begin{openproblem}
    Does it hold that $\unitaryQIP = \unitaryPSPACE$?
\end{openproblem}

Another interesting open question concerns the relationship between between traditional complexity theory and unitary complexity theory, and in particular the Uhlmann Transformation Problem:
\begin{openproblem}
    $\SuccinctUhlmann \in \unitaryBQP^{\PSPACE}$? This is closely related to the Unitary Synthesis Problem of~\cite{aaronson2007quantum} -- not to be confused with our notion of unitary synthesis problems -- which asks if there is a quantum algorithm $A$ and for every $n$-qubit unitary $U$ a boolean function $f:\{0,1\}^{\poly(n)} \to \{0,1\}$ such that the unitary $U$ can be implemented by $A^{f_U}$. 
\end{openproblem} \newpage
\part{Uhlmann Transformation Problem: Reductions} \label{part:applications}
\section{Quantum Cryptography}
\label{sec:qcrypto}

We show how our unitary complexity framework, and in particular the Uhlmann Transformation Problem, can be used to capture computational assumptions necessary for quantum cryptography. First, we show the security of quantum commitment schemes is \emph{equivalent} to the average-case hardness of the Uhlmann Transformation Problem. We also show that if $\avgUnitary{HVSZK}$ is hard on average, then secure quantum commitment schemes exist.

We then show that the hardness of the \emph{succinct} Uhlmann Transformation Problem is necessary for the security of a wide class of quantum cryptographic primitives. This is the class of primitives whose security is based on \emph{falsifiable quantum cryptographic assumptions}. Roughly speaking, a cryptographic assumption is falsifiable if there is an interactive security game with an efficient challenger that can check whether an assumption has been broken. In cryptography, generally speaking most assumptions are falsifiable, although there are some (e.g., knowledge assumptions) that are not. 

By the results of \Cref{sec:structural-uhlmann}, this means that $\avgUnitaryBQP \neq \avgUnitaryPSPACE$ is necessary for falsifiable quantum assumptions. This establishes the first general upper bound on the complexity of breaking (a large class of) computationally-secure quantum cryptography.

\subsection{Quantum bit commitments}\label{sec:commitments}

We first review the notion of quantum commitment schemes, and in particular the notion of a \emph{canonical quantum bit commitment scheme}, which is a non-interactive protocol for bit commitment involving quantum communication. Yan~\cite{yan2023general} showed that a general interactive quantum commitment scheme can always be compiled to a non-interactive commitment scheme with the same security properties. Thus without loss of generality we focus on such non-interactive schemes. 

\begin{definition}[Canonical quantum bit commitment~\cite{yan2023general}]\label{def:commitment_scheme}
A \emph{canonical non-interactive quantum bit commitment scheme} is given by a uniform family of unitary quantum circuits $\{ C_{\lambda,b} \}_{\lambda \in \N,b \in \{0,1\}}$ where for each $\lambda$, the circuits $C_{\lambda,0},C_{\lambda,1}$ act on $poly(\lambda)$ qubits and output two registers $\reg{C},\reg{R}$. The scheme has two phases:
\begin{enumerate}
    \item In the \emph{commit stage}, to commit to a bit $b \in \{0,1\}$, the sender prepares the state $\ket{\psi_{\lambda,b}}_{\reg{RC}} = C_{\lambda,b} \ket{0 \cdots 0}$, and then sends the ``commitment register'' $\reg{C}$ to the receiver. 

    \item In the \emph{reveal stage}, the sender announces the bit $b$ and sends the ``reveal register'' $\reg{R}$ to the receiver. The receiver then accepts if performing the inverse unitary $C_{\lambda,b}^\dagger$ on registers $\reg{C},\reg{R}$ and measuring in the computational basis yields the all zeroes state.
\end{enumerate}
\end{definition}

The security of a canonical commitment scheme consists of two parts, hiding and binding, which we define next.

\begin{definition}[Hiding property of commitment scheme]\label{def:hiding_commitment}
    Let $\eps(\lambda)$ denote a function. We say that a commitment scheme $\{ C_{\lambda,b} \}_{\lambda,b}$ satisfies \emph{$\eps$-computational (resp.~$\eps$-statistical) hiding} if for all polynomial-time algorithms  (resp.~for all algorithms) $A = (A_\lambda)_\lambda$ that take as input the commitment register $\reg{C}$ of the scheme $\{ C_{\lambda,b} \}_{\lambda,b}$, the following holds for sufficiently large $\lambda$:
    \begin{equation}\label{eqn:commitment_hiding_property}
        \Big | \Pr \Big [ A_\lambda(\rho_{\lambda,0}) = 1 \Big ]- \Pr \Big [ A_\lambda(\rho_{\lambda,1}) = 1\Big ] \Big | \leq \eps(\lambda) \,.
    \end{equation}
    Here, $\rho_{\lambda,b}$ denotes the reduced density matrix of $\ket{\psi_{\lambda,b}}$ on register $\reg{C}$. 
    If $\eps$ is a negligible function of $\lambda$ then we simply say that the scheme satisfies \emph{strong} computational (resp.~statistical) hiding. 

\end{definition}

\begin{definition}[Honest binding property of commitment scheme]\label{def:honest_binding}
    Let $\eps(\lambda)$ denote a function. We say that a commitment scheme $\{ C_{\lambda,b} \}_{\lambda,b}$ satisfies \emph{$\eps$-computational (resp.~$\eps$-statistical) honest binding} if for all polynomial-time algorithms  (resp.~for all algorithms) $A = (A_\lambda)_\lambda$ that take as input the reveal register $\reg{R}$ the following holds for sufficiently large $\lambda$:
    \begin{equation}
        \fidelity \left( \Big( A_{\lambda} \otimes \id_{\reg{C}} \Big)(\psi_{\lambda,0}) , \psi_{\lambda,1} \right ) \leq \eps(\lambda)\,,
    \end{equation}
    where $\psi_{\lambda,b} = \ketbra{\psi_{\lambda,b}}{\psi_{\lambda,b}}_{\reg{R} \reg{C}}$.
    
    If $\eps$ is a negligible function of $\lambda$ then we simply say that the scheme satisfies \emph{strong} computational (resp.~statistical) honest binding. \end{definition}

\begin{remark}
    \cref{def:honest_binding} is called \emph{honest} binding because it requires the binding property only for the states $\ket{\psi_{\lambda,b}}$ that are produced if the commit phase is executed honestly.
    We refer to \cite{yan2023general} for a discussion of this definition and stronger versions thereof.
    Throughout this paper, we will only consider the honest binding property, so we will just drop the term ``honest'' for brevity.
\end{remark}

\begin{remark}
The definitions of hiding and binding can easily be revised for non-uniform adversaries, perhaps even with quantum side information, but for simplicity we focus on uniform adversaries. The more general security notion would correspond to unitary complexity classes with \emph{advice}, e.g., $\class{avgUnitaryBQP/poly}$ (i.e., non-uniformity via classical advice) or $\class{avgUnitaryBQP/qpoly}$ (i.e., non-uniformity via quantum advice). We leave this for future work.
\end{remark}

Before discussing the connection between the Uhlmann Transformation Problem and commitment schemes, we review several basic facts about them. First, information-theoretically secure quantum commitments do not exist:

\begin{theorem}[Impossibility of unconditionally secure quantum commitments~\cite{mayers1997unconditionally,lo1998quantum}]
    There is no quantum commitment scheme that satisfies both strong statistical hiding and strong statistical binding. 
\end{theorem}

Thus at least one of the hiding or binding must be computationally secure. There are two commonly considered \emph{flavors} of quantum commitments: one with statistical hiding and computational binding, and the other with statistical binding and computational hiding. A remarkable fact about canonical quantum commitments is that there is a generic blackbox reduction between these two flavors~\cite{crepeau2001convert,yan2023general,gunn2022commitments,hhan2022hardness}.

\begin{proposition}[{\cite[Theorem 7]{hhan2022hardness}}]
\label{prop:flavor-switching}
Let $\eps(n),\delta(n)$ be functions. If $C$ is an $\eps$-computationally (resp.~statistical) hiding and $\delta$-statistical (resp.~computational) binding commitment scheme, then there is an efficient black-box transformation of $C$ into another commitment scheme $C'$ that is a $\sqrt{\delta}$-statistical (resp.~computational) hiding and $\eps$-computationally (resp.~statistical) binding commitment scheme. 
\end{proposition}

\paragraph{Commitments and the Uhlmann Transformation Problem.} The impossibility result of~\cite{mayers1997unconditionally,lo1998quantum} implies that information-theoretically secure quantum bit commitment schemes do not exist because Uhlmann transformations suffice to break the security of such schemes. We strengthen this to argue that the security of quantum commitments is \emph{equivalent} to the \emph{instance-average-case hardness} of the Uhlmann Transformation Problem. 

First we define what we mean by instance-average-case hardness of a (distributional) unitary synthesis problem. Roughly speaking, this notion of average-case hardness stipulates that there is an efficiently sampleable distribution over \emph{instances} $x$ such that all polynomial-time algorithms fail to implement the unitary transformation $U_x$ up to some inverse-polynomial error, with at least inverse-polynomial probability over instances sampled from the distribution. More precisely:

\begin{definition}[Hard unitary synthesis instances]
    Let $(\usynth{U} = (U_x)_x,\Psi = (\ket{\psi_x})_x)$ be a distributional unitary synthesis problem and let $D = (D_n)_{n \in \N}$ denote a uniform family of efficiently sampleable distributions over instances (i.e., $D_n$ is a distribution over $\{0,1\}^n$ that can be sampled in polynomial time). We say that $D$ is a \emph{distribution of hard instances for $(\usynth{U},\Psi)$} if there exists a polynomial $p(n)$ such that for all polynomial-time algorithms $A = (A_x)_x$, for all sufficiently large $n$, for all channel completions $\Phi_x$ of $U_x$ we have
    \[
        \Pr_{x \leftarrow D_n} \Big [ \fidelity( A_{x}(\psi_x),  \Phi_x(\psi_x)) \leq 1 - \frac{1}{p(n)} \Big ] \geq \frac{1}{p(n)}~.
    \]
    We say that \emph{$(\usynth{U},\Psi)$ is instance-average-case hard for $\avgUnitaryBQP$} if there exists an efficiently sampleable family of distributions of instances $D = (D_n)_n$ that is hard for it.
\end{definition}

\begin{remark}
    We emphasize that the ``average'' in the notion of ``instance-average-case hard'' is different from the ``average'' in the definition of $\avgUnitaryBQP$ or $\avgUnitary{HVSZK}$. The former refers to the distribution over \emph{instances} (i.e., the $x$ subscripting $U_x$), whereas the latter refers to the distribution over \emph{quantum inputs} (i.e., the mixed state that the unitary transformation $U_x$ is applied to).  
\end{remark}

We now state the main theorem of this subsection.

\begin{theorem}
\label{thm:uhlmann-hardness-implies-commmitments}  
$\DistUhlmann_{1 - \eps}$ is instance-average-case hard for $\avgUnitaryBQP$ for some negligible function $\eps(n)$ if and only if quantum commitments with strong statistical hiding and strong computational binding exist. 
\end{theorem}

\begin{proof}
    We first prove the ``if'' direction (i.e., secure commitments implies the hardness of $\DistUhlmann$). Let $C = \{C_{\lambda,b}\}$ denote a strong statistically hiding, strong computationally binding commitment scheme. Let $D$ denote the distribution that on input $1^n$, always outputs the pair $(C_{n,0},C_{n,1})$ (i.e., $D_n$ is a singleton distribution). The strong statistical hiding property of $C$ implies that $(C_{n,0},C_{n,1})$ (along with the distributional state $\ket{C_{n,0}}$) is a $\DistUhlmann_{1 - \eps}$ instance for some negligible function $\eps$. The strong computational binding property implies that $D$ is a distribution of hard instances for $\DistUhlmann_{1 - \eps}$. 

    Now we prove the ``only if'' direction (i.e., hardness of $\DistUhlmann$ implies secure commitments). Let $D$ denote an efficiently sampleable family of distributions of hard instances for $\DistUhlmann_{1 - \eps}$. These instances are strings $x$ that encode pairs $(C_0,C_1)$ of circuits outputting bipartite pure states $\ket{C_0},\ket{C_1}$ respectively. First we will show that this implies the existence of a commitment with weak binding security; then we will amplify the security via parallel repetition~\cite{BQSY24}. 

    For notational simplicity we fix the parameter $n$ and let it be implicit throughout this proof. Let $D_x$ denote the probability of sampling the $\DistUhlmann_{1 - \eps}$ instance $x = (C_0,C_1)$ from distribution $D$. Let $\sum_x \sqrt{D_x} \ket{x} \otimes \ket{\vartheta_x}$ denote the purification of the sampling algorithm for $D$. Thus the following states can be prepared in polynomial time: for $b \in \{0,1\}$,
    \begin{gather*}
        \ket{\psi_b} := \sum_{x \in \{0,1\}^n} \sqrt{D_x} \, \ket{x} \otimes \ket{\vartheta_x} \otimes \ket{C_b} \otimes \ket{x}~.
    \end{gather*}
    Here, the division of registers of $\ket{\psi_b}$ are as follows: the initial $\ket{x}$, $\ket{\vartheta_x}$ registers and register $\reg{A}$ of $\ket{C_b}$ are designated the commitment register, and the register $\reg{B}$ of $\ket{C_b}$ as well as the last $\ket{x}$ register are designated the reveal register. 

    Thus $(\ket{\psi_0},\ket{\psi_1})$ denotes a candidate commitment scheme. We first argue that it is strongly statistically hiding. The reduced state of $\ket{\psi_b}$ on the commitment register is
    \[
        \rho_b = \sum_x D_x \ketbra{x}{x} \otimes \ketbra{\vartheta_x}{\vartheta_x} \otimes \Tr_{\reg{B}}(\ketbra{C_b}{C_b})~.
    \]
    Then we have
    \[
    \fidelity(\rho_0,\rho_1) \geq \sum_x D_x \,\, \fidelity(\Tr_{\reg{B}}(\ketbra{C_0}{C_0}),\Tr_{\reg{B}}(\ketbra{C_1}{C_1})) \geq 1 - \eps(n)
    \]
    where we used that $(C_0,C_1)$ is a $\DistUhlmann_{1 - \eps}$ instance.
    By Fuchs-van de Graaf we have that the trace distance is at most $\sqrt{\eps(n)}$ which is still negligibly small, and thus the strong statistical hiding property holds.

    We now argue that it is weakly computationally binding. Let $A$ be a quantum polynomial-time algorithm acting only on register $\reg{R}$. Then
    \begin{align*}
        \fidelity((\id_{\reg{C}} \otimes A_{\reg{R}})(\psi_0),\psi_1) &= \sum_x D_x \, \, \fidelity( (\id_{\reg{A}} \otimes A_{\reg{R}})(\ketbra{C_0}{C_0} \otimes \ketbra{x}{x}), \ketbra{C_1}{C_1} \otimes \ketbra{x}{x}) \\
        &\leq (1 - \alpha) + \alpha \cdot \Big(1 - \frac{1}{p(n)} \Big) = 1 - \frac{\alpha}{p(n)}
    \end{align*}
    where $\alpha$ is the probability 
    \[
        \Pr_{x \leftarrow D} \Big [ \td(A(\ketbra{C_0}{C_0} \otimes \ketbra{x}{x}), \ketbra{C_1}{C_1}) \geq \frac{1}{p(n)} \Big]
    \]
    which is guaranteed to be at least $1/p(n)$ by the definition of hardness on average. Thus the above is at most $1 - 1/p(n)^2$.

    We now amplify the binding security. Consider the repeated commitment $(\ket{\psi_0}^{\otimes t},\ket{\psi_1}^{\otimes t})$ for $t = np(n)$. It was shown by~\cite[Corollary 1.2]{BQSY24} that any efficient algorithm for the amplified commitment that maps $\ket{\psi_0}^{\otimes t}$ to have fidelity non-negligibly greater than $(1 - 1/p(n))^t \leq \exp(-\Omega(n))$ can be converted into an efficient algorithm that maps $\ket{\psi_0}$ to have fidelity greater than $1 - 1/p(n)$ with $\ket{\psi_1}$, breaking the binding security of the original commitment, which is a contradiction. Therefore the amplified commitment scheme has strong computational binding security. Its hiding security is still strong, since we've only repeated the commitment a polynomial number of times. 
\end{proof}

By flavor switching (\Cref{prop:flavor-switching}), we also get the equivalence between the instance-average-case hardness of the Uhlmann Transformation Problem and the existence of commitments with strong computational hiding and strong statistical binding. 

\paragraph{Commitments and Unitary Zero-Knowledge.} Given the close relationship between the Uhlmann Transformation Problem and zero-knowledge protocols for unitary synthesis (as explored in \Cref{sec:structural-uhlmann}), one might also expect an equivalence between the existence of a hard problem in $\avgUnitary{HVSZK}$ and the existence of secure quantum commitments. 

We show one direction: the instance-average-case hardness of $\avgUnitary{HVSZK}$ also implies the existence of secure commitments. This can be viewed as a quantum analogue of the result of Ostrovsky~\cite{ostrovsky1991one}, who showed that a hard-on-average problem in (classical) $\mathsf{SZK}$ implies the existence of one-way functions. The intuition is that a hard distribution over instances of some unitary synthesis problem in $\avgUnitary{HVSZK}$ implies the existence of a hard distribution over instances of $\DistUhlmann$. 

\begin{theorem}
\label{thm:hvszk-hardness-implies-commitments}
    Let $c,s:\N \times \N \to [0,1]$ be such that $c-s \geq \Omega(1)$. Let $\nu = \negl(n,r)$. If $\avgUnitary{HVSZK}_{c,s,\nu}$ is instance-average-case hard for $\avgUnitaryBQP$, then quantum commitments with strong statistical hiding and strong computational binding exist. 
\end{theorem}
\begin{proof}
    Let $(\usynth{U} = (U_x),\Psi = (\psi_x))$ denote a distributional unitary synthesis problem in $\avgUnitary{HVSZK}_{c,s,\nu}$ with a distribution $D$ of hard instances where $p(n)$ is the associated polynomial. By definition, there exists a $m(n)$-round zero knowledge protocol (for some polynomial $m(n)$) for $(\usynth{U},\Psi)$ with an honest prover $P$, an (honest) verifier $V = (V_{x,r})_{x,r}$ and a simulator $S$ with negligible error $\nu(n,r)$. Let us recall the behavior of the simulator: on input $(x,r,j)$, the simulator outputs a purification $\ket{C_{x,j}}$ that is $\nu(n,r)$-close to the state of the honest verifier immediately after receiving the $j$'th message of the honest prover. We assume that $\ket{C_{x,0}} = \ket{\psi_x} \otimes \ket{0 \cdots 0}$. 

    In what follows we fix $r$ to be $4p(n)$, and omit its mention for notational convenience. Thus we write $V_x,\nu(n),(x,j)$ instead of $V_{x,r},\nu(n,r),(x,r,j)$, etc. 
    
    Define the uniform family of distributions $E = (E_n)_n$ where $E_n$ outputs a string $y$ sampled as follows:
    \begin{itemize}
        \item $x$ is sampled from $D_n$ and $j$ is a uniformly random integer between $0$ and $m(n)-1$. 
        \item $y = (B_{x,j},C_{x,j+1})$ encodes a pair of circuit descriptions where $B_{x,j} = (V_{x} \otimes \id) C_{x,j}$ with $C_{x,j}$ corresponding to the purified circuits (i.e., meaning ignoring measurements and partial traces) of the simulator $S$ on input $(x,j)$; and $C_{x,j+1}$ is the purified circuit of the simulator $S$ on input $(x,j+1)$. 
    \end{itemize}
    Since the simulator can be implemented by a polynomial-size circuit, the distribution family $(E_n)_n$ can be sampled in polynomial time. 

    Thus the purifications $\ket{B_{x,j}}$ and $\ket{C_{x,j+1}}$ are approximately consistent in the following sense: the reduced density matrix of $\ket{C_{x,j+1}}$ on the verifier's system has fidelity at least $1 - 4\nu(n)$ with the reduced density matrix of $\ket{C_{x,j}}$, evolved by the verifier's unitary $V_{x}$. This implies that the pair $y = (B_{x,j},C_{x,j+1})$ defines an instance of $\Uhlmann_{1 - 4\nu}$. 

    We now argue that $E$ is a distribution of hard instances for $\DistUhlmann_{1 - 4\nu}$. Suppose not. Then by definition of instances (and applying Fuchs-van de Graaf), for the polynomial $q(n) = 4m(n)p(n)^2$ there exists a quantum polynomial time algorithm $A = (A_y)_y$, an integer $n$, and a channel completion $\Phi_y$ of the canonical Uhlmann transformation for instance $y = (B_{x,j},C_{x,j+1})$, such that
\[
        \E_{\substack{x \leftarrow D \\ j \leftarrow \{0,1,\ldots,m(n)-1\}}} I_{x,j}  \leq \frac{1}{q(n)}
    \]  
    where $I_{x,j}$ denotes the indicator variable for the event $\td( (A_{y} \otimes \id) (\ketbra{B_{x,j}}{B_{x,j}}),  (\Phi_y \otimes \id)(\ketbra{B_{x,j}}{B_{x,j}})) \geq O(\frac{1}{q(n)})$. Multiplying both sides by $m(n)$ and applying Markov's inequality, we get that with probability at most $\sqrt{m(n)/q(n)}$ we have a \emph{bad $x$}, i.e., one where 
    \[
        \sum_{0 \leq j < m(n)} I_{x,j} > \sqrt{\frac{m(n)}{q(n)}}~.
    \]
    Otherwise, if $x$ is \emph{good}, then $\sum_{0 \leq j < m(n)} I_{x,j} \leq \sqrt{m(n)/q(n)} < 1$, but since $I_{x,j}$ is either $0$ or $1$, this implies that $I_{x,j} = 0$ for \emph{all} $j$. 

    Note that when $I_{x,j} = 0$, we have that for all $j$,
\begin{align*}
&\td( (A_{y} \otimes \id) (\ketbra{B_{x,j}}{B_{x,j}}),  \ketbra{C_{x,j+1}}{C_{x,j+1}}) \\\ & \qquad \leq \td( (A_{y} \otimes \id) (\ketbra{B_{x,j}}{B_{x,j}}),  (\Phi_y \otimes \id)(\ketbra{B_{x,j}}{B_{x,j}}))
+ \td( (\Phi_y \otimes \id)(\ketbra{B_{x,j}}{B_{x,j}}),\ketbra{C_{x,j+1}}{C_{x,j+1}}) \\
&\qquad \qquad \leq O(\frac{1}{q(n)}) + \sqrt{4\nu(n)} =: \xi(n)
\end{align*}
where we used that any channel completion $\Phi_y$ of the canonical Uhlmann transformation achieves the optimal fidelity, plus Fuchs-van de Graaf. Note that the error $\xi(n)$ is $O(1/q(n))$ because $\nu(n)$ is negligble. 

    We now describe an efficient procedure $F$ that, when $x$ is good, simulates the prover in the $\avgUnitary{HVSZK}$ protocol to implement the unitary $U_x$ on $\ket{\psi_x}$ with error smaller than $1/p(n)$. The idea is as follows: given input the target register of $\ket{\psi_x}$, evolve it forward by applying the verifier's circuit $V_x$ along with some ancilla zeroes. The pure global state is $\ket{B_{x,0}} = V_x \ket{\psi_x} \ket{0 \cdots 0}$. Applying $A_{y_0}$ for $y_0 = (B_{x,0},C_{x,1})$ to the prover and message registers of $\ket{B_{x,0}}$ we get a state that is $\xi(n)$-close to $\ket{C_{x,1}}$. Apply $V_x$ to the verifier and message register to get a state that is $\xi(n)$-close to $\ket{B_{x,1}}$, and apply $A_{y_1}$ for $y_1 = (B_{x,1},C_{x,2})$ to simulate the second prover action, and so on. After $m(n)$ iterations we have a state that is $m(n)\xi(n)$-close to the final state of the interaction between an honest verifier and honest prover. Thus the acceptance probability of this state is at least $c - m(n) \xi(n)$, which for large enough $n$, is greater than $s$. Therefore by the soundness of the protocol, this final state is $1/r(n)$-close to the desired output state $(\id \otimes U_x) \ket{\psi_x}$ (perhaps with some junk state attached). Thus overall this procedure implements the transformation $U_x$ on $\ket{\psi_x}$ up to trace distance error $m(n)\xi(n) + 1/r(n) < 1/2p(n)$, or (by Fuchs-van de Graaf) equivalent with fidelity at least $1 - 1/p(n)$. 

    Thus when sampling $x$ from $D$ we get a good $x$ with probability at least $1 - \sqrt{m(n)/q(n)} \geq 1 - 1/2p(n)$, and in this case the polynomial-time procedure $F$ implements $U_x$ on $\ket{\psi_x}$ with fidelity at least $1 - 1/p(n)$. This contradicts the definition of $D$ being a hard distribution for the unitary synthesis problem $(\usynth{U},\Psi)$. 

    Therefore $E$ is a hard distribution of instances for $\DistUhlmann_{1 - 4\nu}$. \Cref{thm:uhlmann-hardness-implies-commmitments} then implies the existence of commitments with the desired security properties.
\end{proof}

\begin{remark}
    We suspect that the converse of \Cref{thm:hvszk-hardness-implies-commitments} holds; however this is closely related to the question of whether $\DistUhlmann$ is a complete problem for $\avgUnitary{HVSZK}$, which we can show if a polarization lemma holds for the Uhlmann Transformation Problem (see \Cref{sec:polarize} for a discussion of this).
\end{remark}

\subsection{Falsifiable quantum cryptographic assumptions}\label{sec:falsifiable}

In this section, we show an $\avgUnitary{PSPACE}$ upper bound for breaking \emph{falsifiable quantum cryptographic assumptions}, which can be seen as a quantum analogue of the notion of falsifiable assumption considered by Naor~\cite{Naor2003} as well as Gentry and Wichs~\cite{cryptoeprint:2010/610}.
Morally having a falsifiable assumption means that the challenger in the security game must be efficient, so that if an adversary claims to break the security game, it is possible to verify that they have done so.
Roughly speaking, we show that a falsifiable assumption is either \emph{information-theoretically} secure (in which case, not even a computationally unbounded prover can win at the security experiment beyond a certain threshold), or it can be reduced to $\avgSuccinctUhlmann$, and hence it can broken in $\avgUnitary{PSPACE}$ (as shown in \Cref{sec:structural-succinct-uhlmann}).

Our notion of a \emph{falsifiable quantum cryptographic assumption} captures most cryptographic assumptions in both classical and quantum cryptography. The definition is essentially a $\QIP$ protocol, albeit cast in a cryptographic language. Instead of a \emph{verifier}, we have a \emph{challenger}; instead of a \emph{prover}, we have an \emph{adversary}. We formally define falsifiable quantum cryptographic assumptions as follows. We refer the reader to \Cref{sec:protocols} for the formal definitions of quantum verifiers and interactive protocols.

\begin{definition}[Falsifiable quantum cryptographic assumption]\label{def:falsifiable}
A \emph{falsifiable quantum cryptographic assumption} (or \emph{falsifiable assumption} for short) is a pair $(\cal C,c)$ consisting of a polynomial-time quantum verifier $\cal C = (\cal C_x)_x$ (which we call the \emph{challenger}) and a constant $c \in [0,1]$.
Given a string $x \in \bits^*$,\footnote{
    Here, $x$ should be taken as the security parameter in unary $1^\lambda$, and perhaps in addition expected format of the interaction.
    This includes for example, the number of queries that the adversary wishes to make (in a CCA security game for an encryption scheme as an example), or an upper bound on the message length sent by the adversary (in a collision finding security game as an example).
    The point of having $x$ is so that the overall running time of the challenger is upper bounded by a \emph{fixed} polynomial in $|x|$.
    Furthermore, since we allow arbitrary bitstrings, this should be regarded as auxiliary input to the cryptosystem.
}
the challenger $\cal C_x$ engages in an interaction with a prover $\cal A$ (which also gets the input $x$) called the \emph{adversary}. At the end of the protocol, the challenger accepts or rejects. If the challenger accepts, we say that the adversary \emph{wins}.
\end{definition}

See \Cref{fig:falsifiable_assumption_circuit} for a depiction of an interaction between a challenger and adversary. We now describe the security property corresponding to a falsifiable assumption.

\begin{definition}[Security of a falsifiable assumption] 
A falsifiable assumption $(\cal C,c)$ is \emph{computationally secure} (resp.~\emph{information-theoretically secure}) if for all polynomial-time (resp.\ computationally unbounded) adversaries $\cal A$, there exists a negligible function $\nu$ such that for all $x \in \bits^*$, the probability that the adversary is accepted is at most $c + \nu(|x|)$ over the randomness of the interaction $\cal C_x \interact \cal A$. We say that a (possibly inefficient) adversary $\cal A$ \emph{breaks instance $x$ of the assumption $(\cal C,c)$ with advantage $\delta$} if $\Pr\Big( \cal C_x \interact \cal A \text{ accepts} \Big) \geq c + \delta$. \end{definition}

Here are some (informally-described) examples of falsifiable quantum cryptographic assumptions.
\begin{enumerate}
    \item (\emph{Public-key quantum money}) Consider a candidate public-key quantum money scheme (see~\cite[Lectures 8 and 9]{aaronson2016complexity} for a longer discussion of quantum money). The assumption here is the pair $(\cal C^{\$},0)$. The challenger $\cal C^{\$}$ first generates a random money state along with the serial number and sends both to the adversary (while remembering the serial number). The adversary wins if it can send back two states (which may be entangled) that both pass the money scheme's verification procedure. 

    \item (\emph{Pseudorandom states}) 
Consider a candidate pseudorandom state generator $G$~\cite{JLS18}. The assumption here is $(\cal C^{\rm PRS},\frac{1}{2})$ where the instances $x$ specify the security parameter $\lambda$ as well as a positive integer $t$. The challenger $\cal C^{\rm PRS}$, given $x = (1^\lambda,1^t)$, either sends to the adversary $t$ copies of a pseudorandom state or $t$ copies of a Haar-random state (which can be done efficiently using, e.g., $t$-designs~\cite{ambainis2007quantum}). The adversary wins if it can guess whether it was given pseudorandom states or Haar-random states.

    \item (\emph{Quantum EFI pairs}) Consider a candidate ensemble of EFI pairs $\{(\rho_{\lambda,0},\rho_{\lambda,1})\}_\lambda$~\cite{brakerski2022computational}. The assumption here is $(\cal C^{\rm EFI},\frac{1}{2})$. The challenge $\cal C^{\rm EFI}$ picks a random bit $b\in \bits$ and sends $\rho_{\lambda,b}$ to the adversary. The adversary wins if it can guess the bit $b$.
\end{enumerate}

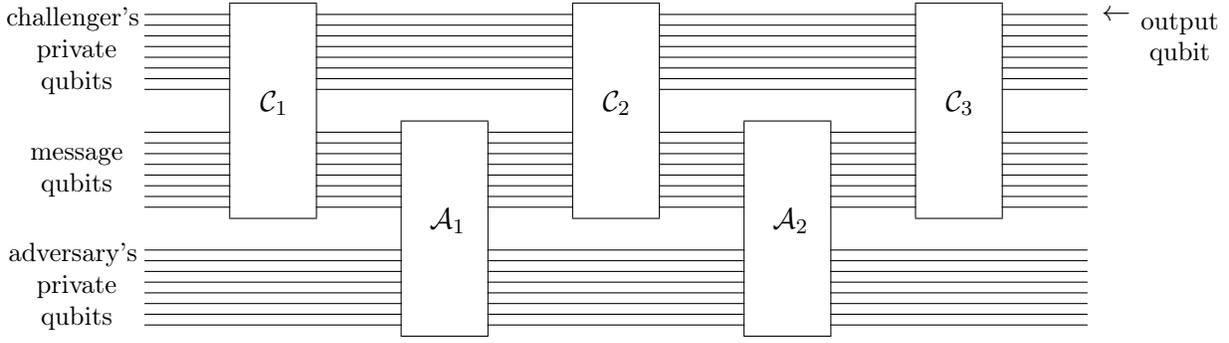
\begin{figure}[!t]
\begin{center}
\setlength{\unitlength}{1771sp}\begingroup\makeatletter\ifx\SetFigFont\undefined \gdef\SetFigFont#1#2#3#4#5{\reset@font\fontsize{#1}{#2pt}\fontfamily{#3}\fontseries{#4}\fontshape{#5}\selectfont}\fi\endgroup \begin{picture}(13224,5472)(589,-5623)
\thinlines
\put(1801,-3961){\framebox(1200,3000){}}
\put(6601,-3961){\framebox(1200,3000){}}
\put(11401,-3961){\framebox(1200,3000){}}
\put(601,-1261){\line( 1, 0){1200}}
\put(601,-1561){\line( 1, 0){1200}}
\put(601,-1861){\line( 1, 0){1200}}
\put(601,-1711){\line( 1, 0){1200}}
\put(601,-1411){\line( 1, 0){1200}}
\put(601,-2011){\line( 1, 0){1200}}
\put(3001,-2761){\line( 1, 0){1200}}
\put(3001,-2911){\line( 1, 0){1200}}
\put(3001,-3061){\line( 1, 0){1200}}
\put(3001,-3211){\line( 1, 0){1200}}
\put(3001,-3361){\line( 1, 0){1200}}
\put(3001,-3511){\line( 1, 0){1200}}
\put(601,-2161){\line( 1, 0){1200}}
\put(601,-1111){\line( 1, 0){1200}}
\put(3001,-3661){\line( 1, 0){1200}}
\put(3001,-3811){\line( 1, 0){1200}}
\put(4201,-5611){\framebox(1200,3000){}}
\put(9001,-5611){\framebox(1200,3000){}}
\put(3001,-1561){\line( 1, 0){3600}}
\put(3001,-1111){\line( 1, 0){3600}}
\put(3001,-1261){\line( 1, 0){3600}}
\put(3001,-1411){\line( 1, 0){3600}}
\put(3001,-1711){\line( 1, 0){3600}}
\put(3001,-1861){\line( 1, 0){3600}}
\put(3001,-2011){\line( 1, 0){3600}}
\put(3001,-2161){\line( 1, 0){3600}}
\put(7801,-1111){\line( 1, 0){3600}}
\put(7801,-1261){\line( 1, 0){3600}}
\put(7801,-1411){\line( 1, 0){3600}}
\put(7801,-1561){\line( 1, 0){3600}}
\put(7801,-1711){\line( 1, 0){3600}}
\put(7801,-1861){\line( 1, 0){3600}}
\put(7801,-2011){\line( 1, 0){3600}}
\put(7801,-2161){\line( 1, 0){3600}}
\put(601,-2761){\line( 1, 0){1200}}
\put(601,-2911){\line( 1, 0){1200}}
\put(601,-3061){\line( 1, 0){1200}}
\put(601,-3211){\line( 1, 0){1200}}
\put(601,-3361){\line( 1, 0){1200}}
\put(601,-3511){\line( 1, 0){1200}}
\put(601,-3661){\line( 1, 0){1200}}
\put(601,-3811){\line( 1, 0){1200}}
\put(5401,-2761){\line( 1, 0){1200}}
\put(5401,-2911){\line( 1, 0){1200}}
\put(5401,-3061){\line( 1, 0){1200}}
\put(5401,-3211){\line( 1, 0){1200}}
\put(5401,-3361){\line( 1, 0){1200}}
\put(5401,-3511){\line( 1, 0){1200}}
\put(5401,-3661){\line( 1, 0){1200}}
\put(5401,-3811){\line( 1, 0){1200}}
\put(7801,-2761){\line( 1, 0){1200}}
\put(7801,-2911){\line( 1, 0){1200}}
\put(7801,-3061){\line( 1, 0){1200}}
\put(7801,-3211){\line( 1, 0){1200}}
\put(7801,-3361){\line( 1, 0){1200}}
\put(7801,-3511){\line( 1, 0){1200}}
\put(7801,-3661){\line( 1, 0){1200}}
\put(7801,-3811){\line( 1, 0){1200}}
\put(10201,-2761){\line( 1, 0){1200}}
\put(10201,-2911){\line( 1, 0){1200}}
\put(10201,-3061){\line( 1, 0){1200}}
\put(10201,-3211){\line( 1, 0){1200}}
\put(10201,-3361){\line( 1, 0){1200}}
\put(10201,-3511){\line( 1, 0){1200}}
\put(10201,-3661){\line( 1, 0){1200}}
\put(10201,-3811){\line( 1, 0){1200}}
\put(12601,-1111){\line( 1, 0){1200}}
\put(12601,-1261){\line( 1, 0){1200}}
\put(12601,-1411){\line( 1, 0){1200}}
\put(12601,-1561){\line( 1, 0){1200}}
\put(12601,-1711){\line( 1, 0){1200}}
\put(12601,-1861){\line( 1, 0){1200}}
\put(12601,-2161){\line( 1, 0){1200}}
\put(12601,-2011){\line( 1, 0){1200}}
\put(12601,-2761){\line( 1, 0){1200}}
\put(12601,-2911){\line( 1, 0){1200}}
\put(12601,-3061){\line( 1, 0){1200}}
\put(12601,-3211){\line( 1, 0){1200}}
\put(12601,-3361){\line( 1, 0){1200}}
\put(12601,-3511){\line( 1, 0){1200}}
\put(12601,-3661){\line( 1, 0){1200}}
\put(12601,-3811){\line( 1, 0){1200}}
\put(601,-5461){\line( 1, 0){3600}}
\put(601,-5311){\line( 1, 0){3600}}
\put(601,-5161){\line( 1, 0){3600}}
\put(601,-5011){\line( 1, 0){3600}}
\put(601,-4861){\line( 1, 0){3600}}
\put(601,-4711){\line( 1, 0){3600}}
\put(601,-4561){\line( 1, 0){3600}}
\put(601,-4411){\line( 1, 0){3600}}
\put(5401,-4411){\line( 1, 0){3600}}
\put(5401,-4561){\line( 1, 0){3600}}
\put(5401,-4711){\line( 1, 0){3600}}
\put(5401,-4861){\line( 1, 0){3600}}
\put(5401,-5011){\line( 1, 0){3600}}
\put(5401,-5161){\line( 1, 0){3600}}
\put(5401,-5311){\line( 1, 0){3600}}
\put(5401,-5461){\line( 1, 0){3600}}
\put(10201,-4411){\line( 1, 0){3600}}
\put(10201,-4561){\line( 1, 0){3600}}
\put(10201,-4711){\line( 1, 0){3600}}
\put(10201,-4861){\line( 1, 0){3600}}
\put(10201,-5011){\line( 1, 0){3600}}
\put(10201,-5161){\line( 1, 0){3600}}
\put(10201,-5311){\line( 1, 0){3600}}
\put(10201,-5461){\line( 1, 0){3600}}
\put(1951,-2461){\,\,\,\,$\cal C_1$}
\put(6751,-2461){\,\,\,\,$\cal C_2$}
\put(11551,-2461){\,\,\,\,$\cal C_3$}
\put(4326,-4111){\,\,\,\,$\cal A_1$}
\put(9126,-4111){\,\,\,\,$\cal A_2$}

\put(-1400,-1725){\parbox{2cm}{\begin{center}\small \hspace{-1mm}challenger's\\private\\qubits
\end{center}}}

\put(-1400,-3425){\parbox{2cm}{\begin{center}\small message\\qubits
\end{center}}}

\put(-1400,-5025){\parbox{2cm}{\begin{center}\small \hspace{-1mm}adversary's\\private\\qubits
\end{center}}}

\put(14000,-1200){$\leftarrow$ \hspace{-6mm}
\parbox{2cm}{\ \\\begin{center}\small output\\
qubit \end{center}}}

\end{picture} \end{center}
\caption{Quantum circuit representation of a 4-message interaction between an efficient challenger and an adversary who seeks to falsify a cryptrographic assumption $(\mathcal{C},c)$.}
\label{fig:falsifiable_assumption_circuit}
\end{figure}

\begin{theorem}
\label{thm:falsifiable}
A falsifiable quantum cryptographic assumption $(\cal C,c)$ is either information-theoretically secure, or breaking the assumption $(\cal C,c)$ can be reduced to $\avgSuccinctUhlmann_1$.
\end{theorem}

Formally what we mean by ``breaking the assumption can be reduced to $\avgSuccinctUhlmann_1$'' is the following: there exists an adversary $A$ that is a polynomial time quantum query algorithm with access to a $\DistSuccinctUhlmann_1$ oracle and breaks infinitely many instances $x$ of the assumption $(\cal C,c)$ with advantage $1/p(|x|)$ for some polynomial $p$.

The proof of \cref{thm:falsifiable} is very similar to that of \Cref{thm:dist_uhlmann_pzk_hard}: again, the idea is that if we are considering a quantum interactive protocol, we can implement the prover's (or in this case adversary's) actions as Uhlmann unitaries.
Hence, if there is any adversary that can break the falsifiable assumption, we can implement that adversary using a $\DistSuccinctUhlmann_1$ oracle, so breaking  the assumption reduces to $\DistSuccinctUhlmann_1$.
To make the paper more modular, we nonetheless spell out the details.
\begin{proof}[Proof of \cref{thm:falsifiable}]
Suppose that $(\cal C,c)$ is not in fact information-theoretically secure and there exists a possibly inefficient adversary $\cal A$ with at most $r=\poly(n)$ many rounds of interaction and a polynomial $p(n)$ such that
$$
\Pr\Big( \cal C_x \interact \cal A \text{ accepts} \Big) \geq c + 1/p(n)\,,
$$
where $x \in \{0,1\}^*$ and $n = |x|$ for infinitely many $x$'s.
For each round $j \in \{1,\dots,r\}$, we let
\begin{itemize}
    \item $\rho_{\reg{M}_x^j \reg W_x^j}^{(j)}$ denote the state of the message register $\reg M_x^j$
and the private workspace $\reg W_x^j$ of the challenger $\mathcal{C}_x$
at the beginning of the challenger’s $j$’th turn
    \item $\sigma_{\reg M_x^j \reg W_x^j}^{(j)}$ denote the state of the message register
and the challenger's private workspace
at the end of the challenger's $j$’th turn.
\end{itemize}
We now argue that the intermediate states on the message and challenger register in the interaction of $\mathcal{C}_x$ with $\mathcal{A}$ have purifications in $\statePSPACE$.
From~\cite[Lemma 7.5]{metger2023stateqip}, it follows that, for all $x$, there exists a prover $\mathcal{P}_x$ that is accepted with probability at least $c + 1/2p(n)$ for which the following property holds: there are families of pure states
$$
( \ket{\psi_{x,j}}_{\reg M_x^j \reg W_x^j \reg P_x^j})_{x,j},\,
\ket{\varphi_{x,j}}_{\reg M_x^j \reg W_x^j \reg P_x^j})_{x,j} \in \statePSPACE
$$
for some purifying registers $\reg P_x^j$ that are purifications of intermediate states $\rho_{\reg M_x^j \reg W_x^j}^{(j)}$ and $\sigma_{\reg M_x^j \reg W_x^j}^{(j)}$ of the challenger $\mathcal{C}_x$ interacting with the prover $\mathcal{P}_x$.
Moreover, there are polynomial-time Turing machines that, given as input a description of the verifier's actions in the protocol, output succinct classical descriptions of the quantum polynomial-space circuits for preparing  $\ket{\psi_{x,j}}$ and $\ket{\varphi_{x,j}}$.
This holds because \cite[Lemma 7.5]{metger2023stateqip} only relies on the block-encoding transformations implemented in \cite[Theorems 5.5 and 6.1]{metger2023stateqip}, which have efficient (and explicit) descriptions.

This means that for each round $j$ of the protocol, there exist polynomial-space quantum circuits $C^j$ and $D^j$ with efficiently computable succinct classical descriptions $\hat C^j$ and $\hat D^j$ such that $\ket{\psi_{x,j}}_{\reg{M_x^j W_x^j P_x^j}} = C^j \ket{0 \dots 0}$ and $\ket{\varphi_{x,j}}_{\reg{M_x^j W_x^j P_x^j}} = D^j \ket{0 \dots 0}$ are purifications of the reduced state on the message register $\reg{M_x^j}$ and challenger register $\reg{W_x^j}$ of the interactive protocol right before and after the prover's action in round $j$.
Notice that because the challenger register in the interactive protocol is not acted upon by the prover, the reduced states on the challenger register are unchanged, i.e. 
\begin{align*}
\ptr{\reg{M_x^j P_x^j}}{\proj{\psi_{x,j}}_{\reg{M_x^j W_x^j P_x^j}}} = \ptr{\reg{M_x^j P_x^j}}{\proj{\varphi_{x,j}}_{\reg{M_x^j W_x^j P_x^j}}} \,.
\end{align*}
We can therefore interpret the circuit pair $(C^j, D^j)$ as an instance of the $\SuccinctUhlmann_1$ problem, with $\reg{W^j}$ taking the role of the register that cannot be acted upon by the Uhlmann unitary. Hence, with access to a $\avgSuccinctUhlmann_1$-oracle, we can apply an Uhlmann transformation mapping $\ket{\psi_{x,j}}_{\reg{M_x^j W_x^j P_x^j}}$ to $\ket{\varphi_{x,j}}_{\reg{M_x^j W_x^j P_x^j}}$ by acting only on registers $\reg{M_x^j P_x^j}$.
This means that with the $\avgSuccinctUhlmann_1$-oracle, we can efficiently implement the actions of a successful prover in the interactive protocol.
\end{proof}

 \section{Quantum Shannon Theory}
\label{sec:shannon}

We now relate the Uhlmann Transformation Problem to two fundamental tasks in quantum Shannon theory: decoding the output of quantum channels and compressing quantum information. We show that both of these tasks can be performed in polynomial time if the Uhlmann Transformation Problem can be implemented in polynomial time. We also prove that channel decoding is as hard as solving the Uhlmann transformation problem in the inverse polynomial error regime.

\subsection{Decodable channels}\label{ssec:Noisy_Channel_Decoding}

We discuss the task of decoding the output of a channel (i.e.~recovering the input to the channel from its output). We focus on channels that are \emph{decodable}:

\begin{definition}[Decodable channel]
\label{def:decodable-channel}
    Let $\eps > 0$. A channel $\cal{N}$ mapping register $\reg{A}$ to $\reg{B}$ is \emph{$\eps$-decodable} if there exists a (not necessarily efficient) quantum algorithm $D$ that takes as input register $\reg{B}$ and outputs register $\reg{A}'$ isomorphic to $\reg{A}$ such that 
    \[
        \fidelity \Big( (D_{\reg{B} \to \reg{A}'} \circ \cal{N}_{\reg{A} \to \reg{B}})( \ketbra{\Phi}{\Phi}_{\reg{AR}} ) , \, \ketbra{\Phi}{\Phi}_{\reg{A}'\reg{R}} \Big) \geq 1- \eps~,
    \]
    where $\ket{\Phi}_{\reg{AR}}$ is the maximally entangled state on registers $\reg{AR}$. 
\end{definition}

Decodable channels naturally arise in the context of error corrected communication. Consider a noisy channel $\cal{C}$, and a encoder $\cal{E}$ corresponding to a quantum error correcting code. There are general situations when the concatenated channel $\cal{C} \circ \cal{E}$ is decodable, e.g., for example if $\cal{C}$ is a tensor-product Pauli channel and $\cal{E}$ is a random stabilizer code~\cite[Theorem 24]{wilde2013quantum}. However, it is not clear whether the concatenated channel is efficiently decodable, even if the encoder $\cal{E}$ is efficient. In fact, it is known that decoding arbitrary classical linear codes and stabilizer codes is computationally intractable~\cite{vardy1997intractability,berlekamp2003inherent,hsieh2011np,iyer2015hardness}.

\begin{remark}
    We could also consider a generalization of \Cref{def:decodable-channel} where we consider states other than the maximally entangled state. However we focus on the maximally entangled state for simplicity, and it already illustrates the key ideas of our complexity result.
    It is known that using the maximally entangled state as the input to a coding scheme for a noisy channel is without loss of generality (up to small changes in capacity, see e.g.~\cite[Chapter 15]{renes2015quantum}).
\end{remark}

We first show a sufficient and necessary condition for a channel $\cal{N}: \reg{A} \to \reg{B}$ to be decodable. Recall the definition of 
a Stinespring dilation of a channel: this is an isometry $V: \reg{A} \to \reg{BC}$ such that $\cal{N}(X) = \Tr_{\reg{C}}(V X V^\dagger)$. We introduce a condition about the \emph{complementary channel} $\cal{N}^c(X) \deq \Tr_{\reg{B}}(V X V^\dagger)$, which maps register $\reg{A}$ to register $\reg{C}$, defined relative to a Stinespring dilation $V$: 

\begin{definition}[Decoupling condition for channels]\label{def:decouple_condition}
    We say a channel $\cal{N}_{\reg A \to \reg B}$ satisfies the \emph{decoupling condition with error $\eps$} if 
    \[ 
        \fidelity \Big ( \cal{N}^c_{\reg{A} \to \reg{C}}(\ketbra{\Phi}{\Phi}_{\reg{AR}}) , \, \cal{N}^c_{\reg{A} \to \reg{C}} \Big ( \frac{\id_{\reg{A}}}{\dim \reg{A}} \Big)  \otimes \frac{\id_{\reg{R}}}{\dim \reg{R}} \Big) \geq 1 - \eps \,,
    \]
where $\cal{N}^c$ is a complementary channel to $\cal{N}$ relative to any Stinespring dilation.
\end{definition}

\begin{proposition}[Necessary and sufficient conditions for decodability]
\label{prop:decodable-conditions}
If a channel $\cal{N}$ satisfies the decoupling condition with error $\eps$ then it is $\eps$-decodable. If it is $\eps$-decodable, then it satisfies the decoupling condition with error $2\sqrt{\eps}$.
\end{proposition}

In other words, a channel is decodable if and only if the output of the complementary channel is close to unentangled with the reference register $\reg{R}$ of the maximally entangled state that was input to channel. 

\begin{proof}
    The first direction we prove is the following: if a channel $\cal{N}$ satisfies the decoupling condition, then it is decodable. Let $V$ denote the Stinespring dilation of $\cal{N}$ which defines the complementary channel $\cal{N}^c$ satisfying the decoupling condition.

    Let registers $\reg{A}',\reg{R}'$ be isomorphic to $\reg{A}, \reg{R}$ respectively. Consider the following pure states:
    \begin{align*}
        &\ket{E}_{\reg{RBCA}' \reg{R}'} \deq V_{\reg{A} \to \reg{BC}} \ket{\Phi}_{\reg{RA}} \ot \ket{0}_{\reg{A}' \reg{R}'} \\
        &\ket{F}_{\reg{RA}'\reg{BCR}'} \deq \ket{\Phi}_{\reg{RA}'} \ot V_{\reg{A} \to \reg{BC}}  \ket{\Phi}_{\reg{A} \reg{R}'}~.
    \end{align*}
    Note that the reduced density matrices of $\ket{E}$ and $\ket{F}$ on registers $\reg{C}$ and $\reg{R}$ are, respectively, $\cal{N}^c_{\reg{A} \to \reg{C}}(\ketbra{\Phi}{\Phi}_{\reg{AR}})$ and $\cal{N}^c_{\reg{A} \to \reg{C}} \Big ( \frac{\id_{\reg{A}}}{\dim \reg{A}} \Big) \otimes \frac{\id_{\reg{R}}}{\dim \reg{R}}$. Therefore by the decoupling condition and Uhlmann's theorem there exists a unitary $U$ mapping registers $\reg{BA}' \reg{R}'$ to registers $\reg{A}' \reg{B} \reg{R}'$ such that
    \begin{equation}
        \label{eq:decodable-0}
        \fidelity \Big ( (\id \ot U)\ketbra{E}{E} (\id \ot U^\dagger) , \ketbra{F}{F} \Big) \geq 1 - \eps~.
\end{equation}
    Define the decoding procedure $D$ that maps register $\reg{B}$ to register $\reg{A}'$ and behaves as follows:
    it appends registers $\reg{A}' \reg{R}'$ in the $\ket{0}$ state, applies the isometry $U$ to registers $\reg{BA}' \reg{R}'$, and then traces out registers $\reg{B}\reg{R}'$ to obtain register $\reg{A}'$. Since $\ket{E}$ is the result of applying the Stinespring dilation of $\cal{N}$ to $\ket{\Phi}$ and appending $\ket{0}_{\reg{A}' \reg{R}'}$, and using the fact that tracing out registers $\reg{B} \reg{R}'$ does not reduce the fidelity, \Cref{eq:decodable-0} implies that
    \[
        \fidelity \Big( (D_{\reg{B} \to \reg{A}'} \circ \cal{N}_{\reg{A} \to \reg{B}})(\ketbra{\Phi}{\Phi}_{\reg{AR}}) , \ketbra{\Phi}{\Phi}_{\reg{A}' \reg{R}} \Big) \geq 1 -\eps \,,
    \]
    showing that $\cal{N}$ is $\eps$-decodable, as desired. 

    Now we argue the other direction (if $\cal{N}$ is decodable, then the decoupling condition holds). The fact that it is decodable is equivalent to 
    \[
        \Tr \Big( \ketbra{\Phi}{\Phi} \, (D_{\reg{B} \to \reg{A}'} \circ \cal{N}_{\reg{A} \to \reg{B}})( \ketbra{\Phi}{\Phi}_{\reg{AR}} ) \Big) \geq 1 - \eps~.
    \]
    Considering the Stinespring dilation $V: \reg{A} \to \reg{BC}$ of $\cal{N}$ this is equivalent to 
    \begin{equation}
        \label{eq:decodable-1}
        \Tr \Big( (\ketbra{\Phi}{\Phi}_{\reg{A}' \reg{R}} \ot \id_{\reg{C}}) \, D_{\reg{B} \to \reg{A}'}( V \, \ketbra{\Phi}{\Phi}_{\reg{AR}} V^\dagger ) \Big) \geq 1 - \eps~.
    \end{equation}
    Suppose we measure $D_{\reg{B} \to \reg{A}'}\big( V \, \ketbra{\Phi}{\Phi}_{\reg{AR}} V^\dagger\big)$ with the projector $\proj{\Phi}$ and succeed. The post-measurement state is thus $\ketbra{\Phi}{\Phi} \ot \rho_{\reg{C}}$ for some density matrix $\rho$. Since the measurement succeeds with probability at least $1 - \eps$, by the Gentle Measurement Lemma we get 
    \begin{equation}
        \label{eq:decodable-2}
        \fidelity \Big( D_{\reg{B} \to \reg{A}'}( V \, \ketbra{\Phi}{\Phi}_{\reg{AR}} V^\dagger ), \ketbra{\Phi}{\Phi}_{\reg{A}' \reg{R}} \ot \rho_{\reg{C}} \Big) \geq 1 - \eps~.
    \end{equation}
    Tracing out register $\reg{A}'$ from both sides, which does not reduce the fidelity, yields
    \begin{equation}
        \label{eq:decodable-3}
        \fidelity \Big ( \cal{N}^c_{\reg{A} \to \reg{C}}(\ketbra{\Phi}{\Phi}_{\reg{AR}}) , \, \rho_{\reg{C}} \otimes \frac{\id_{\reg{R}}}{\dim \reg{R}} \Big) \geq 1 - \eps
    \end{equation}
    as desired.
    
    On the other hand, tracing out registers $\reg{A}' \reg{R}$ in \Cref{eq:decodable-2} also yields 
    \begin{equation}
        \label{eq:decodable-4}
        \fidelity \Big (\cal{N}^c_{\reg{A} \to \reg{C}} \Big ( \frac{\id_{\reg{A}}}{\dim \reg{A}} \Big) \, , \,\rho_{\reg{C}} \Big) \geq 1 - \eps~.
    \end{equation}
    Combining \Cref{eq:decodable-3,eq:decodable-4}, tracing out register $\reg{A}'$, and using Fuchs-van de Graaf twice, and we get
    \[
        \fidelity \Big( \cal{N}^c_{\reg{A} \to \reg{C}}(\ketbra{\Phi}{\Phi}_{\reg{AR}}) \, , \,  \cal{N}^c_{\reg{A} \to \reg{C}} \Big ( \frac{\id_{\reg{A}}}{\dim \reg{A}} \Big) \ot \frac{\id_{\reg{R}}}{\dim \reg{R}}  \Big) \geq 1 - 2\sqrt{\eps} \,,
    \]
    which is the desired decoupling condition.
\end{proof}

\subsubsection{Complexity of decoding quantum channels}

Previously we identified necessary and sufficient conditions for when a channel is \emph{information-theoretically} decodable. Now we investigate when a decodable channel can be \emph{efficiently} decoded. First we define a computational problem corresponding to decoding a given channel.

\begin{definition}[$\eps$-Decodable Channel Problem]
    \label{def:decodable-channel-problem}
    Let $\eps,\delta:\N \to [0,1]$ be functions such that $\delta(n) \geq \eps(n)$. We say that $D$ \emph{solves the $\eps$-Decodable Channel Problem with error $\delta$} if for all $x = (1^m,1^n,C)$  where $C$ is an explicit description of a quantum circuit that maps $m$ qubits to $n$ qubits and is a $\eps$-decodable channel,
    the circuit $D$ takes as input $n$ qubits and satisfies
    \[
        \fidelity \Big( (D_x \circ C)(\ketbra{\Phi}{\Phi}), \ketbra{\Phi}{\Phi} \Big) \geq 1 - \delta(|x|) \,,
    \]
    where $\ket{\Phi}$ is the maximally entangled state on $m$ qubits. 
\end{definition}

Even though a channel $\cal{N}$ may be $\eps$-decodable, it may be computationally intractable to decode it to $\eps$ error. The $\delta$ parameter quantifies the gap between the error achieved by the decoding algorithm and the information-theoretic limit.

The main result of this section is to show that the complexity of the Decodable Channel Problem is equivalent to the complexity of the (distributional) Uhlmann Transformation Problem. 

\begin{theorem}
    \label{thm:complexity-decodable-channels}
Let $\eps: \N \to [0,1]$ be a negligible function. If $\DistUhlmann_{1 - 2\sqrt{\eps}}$ can be solved in polynomial time with inverse polynomial error, then the $\eps$-Decodable Channel Problem is solvable with inverse polynomial error. Conversely, if the $O(\sqrt{\eps})$-Decodable Channel Problem is solvable in polynomial time with inverse polynomial error then $\DistUhlmann_{1 - \eps}$ can be solved in polynomial time with inverse polynomial error. 
\end{theorem}

\begin{proof}
    \noindent \textbf{Upper bound.} We start by proving the the ``only if'' direction (if $\DistUhlmann_{1 - \eps}$ is easy, then the Decodable Channel Problem is easy). We present an algorithm $D$ that solves the $\eps$-Decodable Channel Problem, and is efficient under the assumption about $\DistUhlmann$. 

Let $x=(1^m,1^n,C)$ be an instance of the $\eps$-Decodable Channel Problem be such that $C$ is a quantum circuit computing an $\eps$-decodable channel mapping $m$ qubits (which we label as register $\reg{A}$) to $n$ qubits (which we label as register $\reg{B}$). Let $V$ denote the unitary purification of $C$ (see \Cref{def:unitary-purification}) of $C$, which we view also as a Stinespring dilation of $C$ that maps register $\reg{A}$ to registers $\reg{BC}$. Let $\reg{A}', \reg{R}'$ denote registers isomorphic to $\reg{A}, \reg{R}$, respectively. Consider the pure states $\ket{E}_{\reg{RBCA}'\reg{R}'}$ and $\ket{F}_{\reg{RA}'\reg{BCR}'}$ defined in the proof of \Cref{prop:decodable-conditions} with respect to the dilation $V$. Note that these states can be computed by circuits $E,F$ with size $\poly(|C|)$. By padding we can assume without loss of generality that $E,F$ act on $2k$ qubits where $k \geq |x|$. 

Since the channel $C$ is $\eps$-decodable, then by \Cref{prop:decodable-conditions} it satisfies the decoupling condition with error $2\sqrt{\eps}$. Therefore it follows that $y = (1^k,E,F)$ is a valid $\Uhlmann_{1 - 2\sqrt{\eps}}$ instance (where the registers are divided into two groups $\reg{CR}$ and $\reg{BA}'\reg{R}'$). Thus by assumption there is a polynomial-time algorithm $M = (M_{y,r})_{y,r}$ that implements $\DistUhlmann_{1 - 2\sqrt{\eps}} \in \avgUnitaryBQP$~.
By \Cref{prop:operational-avg-case-uhlmann}, it follows that for $y = (1^k,E,F)$ with $k = \poly(|x|)$, the algorithm $M$ satisfies, for all precision $r$, 
\begin{equation*}
    \td \Big((\id \ot M_{y,r})(\ketbra{E}{E}), \ketbra{F}{F} \Big) \leq \frac{1}{r} + (4\eps(k))^{1/4}~. \end{equation*}
By Fuchs-van de Graaf this implies
\begin{equation}
    \label{eq:decodable-0a}
    \fidelity \Big((\id \ot M_{y,r})(\ketbra{E}{E}), \ketbra{F}{F} \Big) \geq \Big(1- \frac{1}{r} - (4\eps(k))^{1/4} \Big)^2~. \end{equation}

Fix a polynomial $q$. The algorithm $D$ behaves as follows on instance $x = (1^m,1^n,C)$ of the $\eps$-Decodable Channel Problem. It receives as input a register $\reg{B}$. It first computes the description of the $\Uhlmann_{1 - 2\sqrt{\eps}}$ instance $y = (1^k,E,F)$ described above. It initializes ancilla registers $\reg{A}' \reg{R}'$ in the zero state, and then applies the algorithm $(M_{y,r})_{y,r}$ that solves the $\DistUhlmann_{1-2\sqrt\eps}$-problem to registers $\reg{BA}'\reg{R}'$ with the precision parameter $r = 1/4q(|x|)$. Finally, the algorithm $D$ then traces out registers $\reg{BR}'$ and outputs the remaining register $\reg{A}'$. 

Now we analyze the behavior of the algorithm $D$ when it receives the $\reg{B}$ register of the state $C_{\reg{A} \to \reg{B}}(\ketbra{\Phi}{\Phi}_{\reg{AR}})$. Note that 
\begin{align*}
    \Big( (D_{x})_{\reg{B} \to \reg{A}'} \circ C_{\reg{A} \to \reg{B}} \Big)(\ketbra{\Phi}{\Phi}_{\reg{RA}}) &= \Tr_{\reg{CBR}'} \Big( (\id \ot M_{y,r})( \ketbra{E}{E}) \Big) \\ \ketbra{\Phi}{\Phi}_{\reg{R} \reg{A}'} &= \Tr_{\reg{ABCR}'} \Big( \ketbra{F}{F} \Big)~.
\end{align*}
By \Cref{eq:decodable-0a} and the fact that the fidelity does not decrease under partial trace we have
\begin{align*}
    \fidelity \Big(   \Big( (D_x)_{\reg{B} \to \reg{A}'} \circ C_{\reg{A} \to \reg{B}} \Big)(\ketbra{\Phi}{\Phi}_{\reg{RA}}) \, , \,\ketbra{\Phi}{\Phi}_{\reg{R} \reg{A}'} \Big) &\geq \fidelity \Big((\id \ot M_{y,r})(\ketbra{E}{E}), \ketbra{F}{F} \Big) \\
    \geq 1 - \frac{2}{r} - 2(4\eps(k))^{1/4} \geq 1 - O(\frac{1}{q(|x|)})~.
\end{align*}
In the last inequality we used that $O(\eps(k)^{1/4})$ is a negligible function of $k$ and thus of $|x|$ (because $k = \poly(|x|)$). 
Thus we have shown that $D$ solves the $\eps$-Decodable Channel Problem up to error $1/q$ in time $\poly(|x|,r) = \poly(|x|)$ as desired. This concludes the ``only if'' direction.

\paragraph{Lower bound.} We now prove the ``if'' part of the theorem (if the Decodable Channel Problem is easy, then $\DistUhlmann$ is easy). The intuition behind the proof is as follows: if $\DistUhlmann$ were hard, then we can construct a family of hard instances of the Decodable Channel Problem. These hard instances, intuitively, will be decodable channels $\cal{N}$ that take as input $b \in \{0,1\}$ and output an \emph{encryption} $\rho_b$. The states $\rho_0$ and $\rho_1$ are far from each other, but are computationally indistinguishable (this is also known as an \emph{EFI pair}~\cite{brakerski2022computational}). Thus no efficient decoder can correctly recover the bit $b$, even though the channel $\cal{N}$ is information-theoretically decodable by construction. 

We describe an efficient reduction from instances of the Uhlmann Transformation Problem to instances of the Decodable Channel Problem. Let $x = (1^n,C_0,C_1)$ be an instance of $\Uhlmann_{1 - \eps}$ for some negligible $\eps$, where the circuits $C_0,C_1$ output a state on registers $\reg{XY}$ and the reduced density matrices of $\ket{C_0},\ket{C_1}$ on register $\reg{X}$ have fidelity at least $1 - \eps$. For $b \in \{0,1\}$ define the circuit $E_b$ that prepares the state 
\[
    \ket{E_b} := \frac{1}{\sqrt{2}}  \ket{C_0}_{\reg{XY}} \ket{0}_{\reg{E}} + (-1)^b \frac{1}{\sqrt{2}} \ket{C_1}_{\reg{XY}}\ket{1}_{\reg{E}} ~.
\]
Define the channel $\cal{N}$ that takes as input a qubit $\ket{b}_{\reg{A}}$, maps it to $\ket{E_b}_{\reg{XYE}}$, and outputs the register $\reg{YE}$. Let $F$ denote the circuit that computes this channel. Note that $F$ is a polynomial-sized circuit in the length of the Uhlmann instance $x$. 

We first argue that the channel $\cal{N}$ is $O(\sqrt{\eps})$-decodable. By Uhlmann's theorem there exists a unitary $V$ acting on register $\reg{YE}$ such that
\[
    \frac{1}{2} \Big(\bra{C_1} \bra{1} (\id_{\reg{X}} \otimes V_{\reg{YE}}) \ket{C_0} \ket{0} + \bra{C_0} \bra{0} (\id_{\reg{X}} \otimes V_{\reg{YE}}) \ket{C_1} \ket{1} \Big) \geq 1 - \eps~. 
\]
By the swapping-distinguishing equivalence of~\cite[Theorem 2]{aaronson2020hardness}, this implies a (not necessarily efficient) measurement $M$ on registers $\reg{YE}$ that distinguishes between $\ket{E_0}$ and $\ket{E_1}$ with bias $1 - \eps$. Consider the following thought experiment: apply the channel $\cal{N}$ to register $\reg{A}$ of $\ket{\Phi}_{\reg{RA}}$ to obtain 
$\frac{1}{\sqrt{2}} \sum_b \ket{b}_{\reg{R}} \ket{E_b}_{\reg{XYE}}$. Trace out the registers $\reg{YE}$; because we can imagine that the measurement $M$ was performed on registers $\reg{YE}$, the resulting density matrix must be $O(\eps)$-close in trace distance to
\[
    \frac{1}{2} \ketbra{0}{0} \otimes \rho + \frac{1}{2} \ketbra{1}{1} \otimes \sigma~.
\]
Since $\rho,\sigma$ are $\sqrt{\eps}$-close in trace distance, this implies that the density matrix must be $O(\sqrt{\eps})$-close to $\frac{\id_{\reg{R}}}{2} \otimes \rho_{\reg{X}}$, which means that the channel $\cal{N}$ satisfies the decoupling condition with error $O(\sqrt{\eps})$. \Cref{prop:decodable-conditions} implies that $\cal{N}$ is $O(\sqrt{\eps})$-decodable.

Let $q$ be a polynomial and let $p(|x|) = O(1/q(|x|)^2)$. By assumption there exists a polynomial-time algorithm $D$ that solves the $O(\sqrt{\eps})$-Decodable Channel Problem with error $1/p$.
In particular, letting $y = (1^k,F)$ for the circuit $F$ computing channel $\cal{N}$ described above, we have 
\[
    \fidelity \Big( (D_y \circ F)(\ketbra{\Phi}{\Phi}), \ketbra{\Phi}{\Phi} \Big) \geq 1 - 1/p(|y|)~\,.
\]
By measuring the register $\reg{R}$ in the standard basis and using the monotonicity of the fidelity function under quantum operations, we get that 
\[
    \frac{1}{2} \sum_b \fidelity \Big((\id_{\reg{X}} \otimes D_y) (\proj{E_b}) \, , \, \proj{b} \Big) \geq 1 - 1/p(|y|)~.
\]
In other words, there is an efficient measurement on registers $\reg{YE}$ that distinguishes $\ket{E_0},\ket{E_1}$ with bias at least $1 - O(1/p(|y|))$. Again by the swapping-distinguishing equivalence~\cite[Theorem 2]{aaronson2020hardness} we get that there is an efficient algorithm $A$  acting on $\reg{YE}$, plus some ancilla registers, that maps $\ket{C_0}$ to have overlap at least $1 - O(1/p(|y|))$ with $\ket{C_1}$. 

By \Cref{prop:operational-avg-case-uhlmann}, this means that the algorithm $A$ implements the canonical Uhlmann transformation corresponding to $(\ket{C_0},\ket{C_1})$ with error at most $O(\sqrt{1/p(|y|)} + \eps(|x|)^{1/4})$. Since $|y| \geq |x|$, the error function $\eps$ is negligible and we assume that all polynomials and error functions are monotonic, this is error bound asymptotically at most $O(\sqrt{1/p(|x|)}) \leq O(1/q(|x|))$. The running time of $A$ is polynomial in $|x|$ as it runs the decoder $D$ on input $y$, which has polynomial size in $|x|$. This concludes the ``if'' direction.
\end{proof}

\subsection{Compressing quantum information}\label{ssec:information_compression}

In this section we show that the computational complexity of performing optimal compression of a quantum state (that can be efficiently prepared) is related to the complexity of the Uhlmann Transformation Problem. 

We consider the \emph{one-shot} version of the information compression task, where one is given just one copy of a density matrix $\rho$ (rather than many copies) and the goal is to compress it to as few qubits as possible while being able to recover the original state within some error. The task is defined formally as follows:

\begin{definition}[Information compression task]\label{def:information_compression_task}
    Let $\delta \geq 0$ and let $\rho$ be an $n$-qubit density matrix. We say that a pair of (not necessarily efficient) quantum circuits $(E,D)$ \emph{compresses $\rho$ to $s$ qubits with error $\delta$} if 
    \begin{enumerate}
        \item $E$ is a quantum circuit that takes as input $n$ qubits and outputs $s$ qubits,
        \item $D$ is a quantum circuit that takes as input $s$ qubits and outputs $n$ qubits,
        \item For all purifications $\ket{\psi}_{\reg{AR}}$ of $\rho$ (where $\reg{R}$ is the purifying register), we have
    \[
        \td \Big ( (D \circ E)(\psi), \psi \Big) \leq \delta
    \]
    where the composite channel $D \circ E$ acts on register $\reg{A}$ of $\ket{\psi}$.
    \end{enumerate}
    Define the \emph{$\delta$-error communication cost of $\rho$}, denoted by $K^\delta(\rho)$, as the minimum integer $s$ such that there exists a pair of quantum circuits $(E,D)$ that compresses $\rho$ to $s$ qubits with error $\delta$.
\end{definition}

In this section, we first analyze what is information-theoretically achievable for one-shot compression. Then, we study the complexity of compressing quantum information to the information-theoretic limit; we will show that it is closely related to the complexity of the Uhlmann Transformation Problem.

\subsubsection{Information-theoretic compression}

It is well-known that $n$ copies of a quantum state $\rho$ can be compressed to $n$ times the von Neumann entropy of $\rho$, in the limit of large $n$~\cite{schumacher1995quantum}. In the one-shot setting the state $\rho$ can be (information-theoretically) compressed to its \emph{smoothed max entropy} and no further. The smoothed max entropy is defined as follows:

\begin{definition}[Smoothed max-entropy]
Let $\eps \geq 0$ and let $\psi_{\reg{AB}}$ be a density matrix on registers $\reg{AB}$. The \emph{max-entropy of register $\reg{A}$ conditioned on register $\reg{B}$ of the state $\psi$} is
        \[
            H_{\max}(\reg{A} | \reg{B})_\psi \deq \sup_{\sigma \in \mathrm{Pos}(\reg{B}) : \Tr(\sigma) \leq 1} \log \| \sqrt{\psi_{\reg{AB}}} \sqrt{\id_{\reg{A}} \otimes \sigma_{\reg{B}}} \|_1^2~.
        \]
        The \emph{$\eps$-smoothed conditional max-entropy} is 
        \[
            H^{\eps}_{\max}(\reg{A} | \reg{B})_\psi \deq \inf_{\sigma:\td(\sigma,\psi) \leq \eps} H_{\max}(\reg{A} | \reg{B})_\sigma~.
        \]
\end{definition}

The following theorem shows that the smoothed max-entropy characterizes, up to additive constants and different smoothing parameters, the limits of one-shot compression. 

\begin{restatable}[Information-theoretic one-shot compression]{theorem}{infotheorycompression}
\label{thm:info-theory-compression}
For all $\delta > 0$ and all density matrices $\rho$,
\[
    H^{\eps_1}_{\max}(\rho) \leq K^\delta(\rho) \leq H^{\eps_2}_{\max}(\rho) + 8 \log \frac{4}{\delta}
\]
where $\eps_1 \deq 2\delta^{1/4}$ and $\eps_2 \deq (\delta/40)^4$. 
\end{restatable}

We prove \Cref{thm:info-theory-compression} in \Cref{sec:info-theoretic-compression}. We note that for tensor product states $\rho^{\otimes k}$, the smoothed max-entropy converges to the well-known von Neumann entropy: 
\[
    \lim_{\eps \to 0} \lim_{k \to \infty} \frac{1}{k} H^{\eps}_{\max} (\rho^{\otimes k}) = H(\rho)~.
\]
This is an instance of the \emph{quantum asymptotic equipartition property}, which roughly states that the min, max, and R\'{e}nyi entropies approach the von Neumann entropy in the limit of many copies of a state~\cite{tomamichel2009fully}.\footnote{In fact, one can give stronger quantitative bounds on the convergence to the von Neumann entropy as a function of the number of copies $k$ and the error $\eps$.} Thus \Cref{thm:info-theory-compression} applied to tensor product states $\rho^{\otimes k}$ recovers Schumacher compression~\cite{schumacher1995quantum}. 

\subsubsection{Complexity of near-optimal compression}

We now study the computational complexity of compressing quantum information to the information-theoretic limit, i.e., to the smoothed max-entropy of a state. We begin by defining compression as a computational task.

\begin{definition}[Compression as a computational task]\label{def:comp_compression}
Let $\eps,\eta:\N \to [0,1]$ be functions. Let $E = (E_x)_x$ and $D = (D_x)_x$ be quantum algorithms. We say that $(E,D)$ \emph{compresses to the $\eps$-smoothed max-entropy with error $\eta$} if for all $x = (1^n,C)$ where $C$ is a quantum circuit that outputs $n$ qubits, we have that $(E_x,D_x)$ compresses $\rho_x \deq C(\proj{0})$ to at most $ H^{\eps(n)}_{\max}(\rho_x) + O(\log \frac{1}{\eps(n)})$ qubits with error at most $\eta(n)$.
\end{definition}

This brings us to the main result of the section, which are upper and lower bounds on the complexity of the compression task. 

\begin{theorem}[Near-optimal compression via Uhlmann transformations]\label{thrm:compression_to_uhlmann_reduction}
\label{thm:comp-compression}
    Let $\eps(n)$ be a negligible function. If $\DistUhlmann_{1 - \eps} \in \avgUnitaryBQP$, 
    then for all polynomials $q(n)$ there exists a pair of polynomial-time algorithms $(E,D)$ that compresses to the $\eps$-smoothed max-entropy with error $\eta(n) = 1/q(n)$.
\end{theorem}

\begin{proof}

Let $x = (1^n,C)$ where $C$ is a quantum circuit that outputs $n$ qubits, and let $\rho_x = C(\proj{0})$. Let $\eps = \eps(n)$. The proof of the upper bound of \Cref{thm:info-theory-compression} involves the following two states:
\begin{gather*}
    \ket{F} \deq \ket{\Phi}_{\reg{E} \reg{E}'} \otimes \ket{\rho}_{\reg{AR}} \,,\\
    \ket{G} \deq \sum_y \ket{y}_{\reg{E}} \otimes (\Pi_y U \ot \id_{\reg{R}})\ket{\rho}_{\reg{AR}} \otimes \ket{0}_{\reg{F}}\,.
\end{gather*}
(The state $\ket{G}$ was called $\ket{\theta}$ in \Cref{thm:info-theory-compression}). Here, $\ket{\Phi}_{\reg{E} \reg{E}'}$ denotes the maximally entangled state on $\reg{E}\reg{E}'$, $\ket{\rho}_{\reg{AR}}$ is the pure state resulting from evaluating a purification of the circuit $C$ on the all zeroes input, the projector $\Pi_y$ denotes projecting the first $n - s$ qubits of register $\reg{A}$ onto $\ket{y}$, and $U$ is a Clifford unitary. Note that $\ket{F},\ket{G}$ can be prepared by circuits $F,G$ whose sizes are polynomial in $n$ and in the size of $C$; this uses the fact that Clifford unitaries can be computed by a circuit of size $O(n^2)$~\cite{aaronson2004improved}. 

The proof of \Cref{thm:info-theory-compression} shows that the reduced density matrices of $\ket{F},\ket{G}$ on registers $\reg{EA}$ have fidelity at least $1 - 2\nu = 1 - \eps^2/16 \geq 1 - \eps$. Thus $(1^m,F,G)$ is a valid $\Uhlmann_{1 - \eps}$ instance. Since $\DistUhlmann_{1 - \eps} \in \avgUnitaryBQP$ by assumption there exists $\poly(n,|C|)$-size circuit $L$ mapping registers $\reg{E}' \reg{A}$ to $\reg{E}' \reg{CF}$ and a channel completion $\Xi$ of the canonical Uhlmann transformation $V$ corresponding to $(\ket{F},\ket{G})$ such that
\[
    \td \Big ( (\id \ot L)(\ketbra{F}{F}) , \, \, (\id \ot \Xi) (\ketbra{F}{F}) \Big) \leq \frac{1}{r(n)}
\]
where $r(n)$ is a polynomial such that $2/r(n) + \eps(n) \leq 1/q(n)$, which is possible because $\eps(n)$ is a negligible function. Similarly there exists a $\poly(n,|C|)$-size circuit $M$ and a channel completion $\Lambda$ of the Uhlmann transformation $V^\dagger$ corresponding to $(\ket{G},\ket{F})$ such that
\[
    \td \Big ( (\id \ot M)(\ketbra{G}{G}) , \, \, (\id \ot \Lambda) (\ketbra{G}{G}) \Big) \leq \frac{1}{r(n)}~.
\]

The proof of \Cref{thm:info-theory-compression} shows shows that there exists a pair of uniformly-computable circuits $(E_{x}^*,D_{x}^*)$ that compresses $\rho_x$ to $s = H^{\eps}_{\max}(\rho_x) + O(\log \frac{1}{\eps})$ qubits with error $\eps$. Notice that the circuits $E_x^*,D_x^*$ are $\poly(n)$-size circuits that make one call to channels $\Xi,\Lambda$, respectively. Now the idea is to ``plug in'' the circuits $L,M$ to implement the call to the channel $\Xi,\Lambda$, respectively. Let $E_x,D_x$ denote the resulting $\poly(n,|C|)$-sized circuits. Using $L,M$ instead of the channels $\Xi,\Lambda$ incurs at most $O(1/r(n))$ error, i.e., $\td \Big( (D_x \circ E_x)(\ketbra{\rho}{\rho}_{\reg{AR}}),(D_x^* \circ E_x^*)(\ketbra{\rho}{\rho}_{\reg{AR}}) \Big) \leq 2/r(n)$. Therefore
\begin{gather*}
    \td \Big( (D_x \circ E_x)(\ketbra{\rho}{\rho}_{\reg{AR}}), \ketbra{\rho}{\rho}_{\reg{AR}} \Big) \leq 2/r(n) + \eps(n) \leq 1/q(n)\,.
\end{gather*}

Letting $E = (E_x)_x$ and $D = (D_x)_x$ we get the desired pair of uniform polynomial-time algorithms that compresses to the $\eps$-smoothed max entropy with inverse polynomial error. 
\end{proof}

We now turn to proving a hardness result for near-optimal compression; it cannot be performed in polynomial-time if \emph{stretch pseudorandom state (PRS) generators} exist. Pseudorandom state generators are a quantum analogue of classical pseudorandom generators (PRGs) and in fact can be constructed from post-quantum pseudorandom generators~\cite{JLS18}, but there is evidence that the assumption of PRS is less stringent than the assumption of post-quantum PRGs~\cite{Kretschmer21,kqst23}. We first recall the definition of a PRS generator:

\begin{definition}[{Pseudorandom state generator \cite[Definition 3]{JLS18}}]
\label{def:vanilla-prs}
We say that a (uniform) polynomial-time algorithm $G = (G_\lambda)_\lambda$ is a \emph{pseudorandom state (PRS) generator} if the following holds. 
\begin{enumerate}    
    \item (\emph{State generation}). For all $\lambda$, on input $k \in \{0,1\}^k$ the algorithm $G$ outputs 
    \[
        G_\lambda(k) = \ketbra{\psi_{k}}{\psi_{k}}
    \]
    for some $m(\lambda)$-qubit pure state $\ket{\psi_k}$. 
    
    \item (\emph{Strong pseudorandomness}). For all polynomials $t(\lambda)$ and non-uniform polynomial-time distinguishers $A = (A_\lambda)_\lambda$ there exists a negligible function $\eps(\lambda)$ such that for all $\lambda$,  we have
    \[
        \left | \Pr_{k \leftarrow \{0,1\}^\lambda} \left [ A_\lambda^{O_{\psi_k}} (G_\lambda(k)^{\otimes t(\lambda)}) = 1 \right] - \Pr_{\ket{\vartheta} \leftarrow \Haar_{m(\lambda)}} \left [ A_\lambda^{O_\vartheta} (\ketbra{\vartheta}{\vartheta}^{\otimes t(\lambda)}) = 1 \right] \right | \leq \eps(\lambda),
    \]
    where $O_\psi \deq \id - 2\ketbra\psi\psi$ is the reflection oracle for $\ket\psi$.
\end{enumerate}
We say that $G$ is a \emph{stretch} PRS generator if $m(\lambda) > \lambda$.
\end{definition}
Here we use the strong pseudorandomness guarantee, which is known to be equivalent to the weaker (standard) pseudorandomness guarantee where the adversary does not get access to the reflection oracle \cite[Theorem 4]{JLS18}. We also note that PRS generators do not necessarily provide any \emph{stretch}; there are nontrivial PRS generators where the output length $m(\lambda)$ can be smaller than the $\lambda$. Furthermore, unlike classical PRGs, it is not known whether PRS can be generically stretched (or shrunk); see~\cite{ananth2022cryptography} for a longer discussion of this. 

We now state our hardness result.

\begin{theorem}[Hardness of near-optimal compression]\label{thrm:PRS_to_compression}
\label{thm:comp-compression-lb}
Let $\eps(n)$ be a function. Let $m(\lambda)$ be a function satisfying 
\[
m(\lambda) > \lambda + O \Big (\log \frac{1}{\eps(m(\lambda))} \Big) + 2
\]
for all sufficiently large $\lambda$. If stretch pseudorandom state generators that output $m(\lambda)$ qubits exist, then there is no non-uniform polynomial-time algorithm $(E,D)$ that compresses to the $\eps$-smoothed max-entropy with error $\frac{1}{2}$.
\end{theorem}

\begin{proof}
    Let $G$ be a PRS generator that outputs $m(\lambda)$-qubit states for $m(\lambda)$ satisfying the conditions stated in \Cref{thm:comp-compression-lb}, and fix a sufficiently large $\lambda \in \N$ for which the condition is satisfied. Define the pure state $\ket{\varphi_\lambda}$ that represents running a unitary purification of the generator $G$ coherently with the keys $k$ in superposition:
    \[
        \ket{\varphi_\lambda}_{\reg{KQA}} \deq 2^{-\lambda/2} \sum_{k \in \{0,1\}^\lambda} \ket{k}_{\reg{K}} \otimes \ket{\tau_k}_{\reg{Q}} \otimes \ket{\psi_k}_{\reg{A}}
    \]
    where $\ket{\psi_k}$ denotes the pseudorandom state output by $G$ on key $k$, and $\ket{\tau_k}$ denotes 
    the state of the ancilla qubits of $G$. Let $\reg{R} \deq \reg{KQ}$. The reduced density matrix of $\ket{\varphi_\lambda}$ on register $\reg{A}$ is the following mixed state:
    \[
    \rho_\lambda \deq 2^{-\lambda} \sum_{k \in \{0, 1\}^\lambda} \ketbra{\psi_k}{\psi_k}~.
    \]
    By the second item of \Cref{prop:entropy-relations} we have $H^{\eps}_{\max}(\rho_{\lambda}) \leq \lambda$.

    Assume for contradiction that there exists a polynomial-time pair of quantum algorithms $(E,D)$ that compresses to the $\eps$-smoothed max-entropy with error $\frac{1}{2}$. Let $x = (1^n,C)$ where $C$ outputs the state $\rho_\lambda$ by first synthesizing the state $\ket{\varphi_\lambda}$ and then tracing out register $\reg{R}$. Clearly $C$ is a $\poly(\lambda)$-sized circuit. Therefore $(E_x,D_x)$ runs in $\poly(\lambda)$ time and compresses $\rho_\lambda$ to $r_\lambda \deq H^{\eps}_{\max}(\rho_\lambda) + O \Big(\log \frac{1}{\eps(m(\lambda))} \Big) \leq \lambda + O \Big(\log \frac{1}{\eps(m(\lambda))} \Big)$ qubits. By assumption we have
    \[
        \td \Big( (D_x \circ E_x)(\ketbra{\varphi_\lambda}{\varphi_\lambda}) , \, \ketbra{\varphi_\lambda}{\varphi_\lambda} \Big ) \leq \frac{1}{2}~.
    \]
    By measuring register $\reg{K}$ and tracing out register $\reg{Q}$ on both arguments (which does not increase the trace distance), we have that
    \begin{equation}
        \label{eq:compression-2}
        \E_k \td \Big ( (D_x \circ E_x)(\ketbra{\psi_k}{\psi_k}) \, , \, \ketbra{\psi_k}{\psi_k}  \Big) \leq \frac{1}{2}~.
    \end{equation}
    Now consider the following distinguisher $A = (A_\lambda)_\lambda$: it gets as input $\ket{\theta}$ where $\ket{\theta}$ is either $\ket{\psi_k}$ for a randomly sampled $k$ or $\ket{\vartheta}$ sampled from the Haar measure; it also gets access to a (controlled) reflection oracle $O_\theta = \id - 2\ketbra{\theta}{\theta}$. It then 
    \begin{enumerate}
        \item applies the channel $D_x \circ E_x$ to input $\ket{\theta}$;
        \item measures $\{ \ketbra{\theta}{\theta}, \id - \ketbra{\theta}{\theta} \}$ using the reflection oracle, and accept if measurement accepts. 
\end{enumerate}
    From \Cref{eq:compression-2} we have that, since the measurement step with respect to $O_{\psi_k}$ accepts on $\ket{\psi_k}$ with probability $1$, then $A_\lambda$ with oracle access to $O_{\psi_k}$ accepts $\ket{\psi_k}$ with probability at least $1 - \eta$ over the choice of key $k$ and the randomness of $A_\lambda$. 

    Now consider what happens when we run $A_\lambda$ with $\ket{\vartheta}$ as input where $\ket{\vartheta}$ is sampled from the Haar measure, as well as with the reflection oracle $O_{\vartheta}$. Since $A$ runs in $\poly(\lambda)$ time, by the pseudorandomness property of $G$ the probability that $A_\lambda$ accepts $\ket{\vartheta}$ is at least $\frac{1}{2} - \negl(\lambda)$. 

    On the other hand we show that since a Haar-random state cannot be compressed, $A_\lambda$ cannot accept with high probability. Let $R \deq 2^{r_\lambda}$ denote the dimensionality of the output of $E_\lambda$, and let $M = 2^{m(\lambda)}$ denote the dimensionality of register $\reg{A}$. For brevity we abbreviate $E_x,D_x$ as $E,D$ respectively. The success probability of $A_\lambda$ given a Haar-random state $\ket{\vartheta}$ and the reflection oracle $O_\vartheta$ can be calculated as follows. First, observe that
    \[
        \int_{\vartheta} \Tr \Big ( (D \circ E)(\ketbra{\vartheta}{\vartheta}) \, \ketbra{\vartheta}{\vartheta} \Big ) \, \mathrm{d}\vartheta = \int_{\vartheta}  \Tr \Big ( E (\ketbra{\vartheta}{\vartheta}) \, D^* (\ketbra{\vartheta}{\vartheta}) \Big)  \, \mathrm{d}\vartheta
    \]    
    where $D^*$ denotes the \emph{adjoint channel} corresponding to $D$; it is the unique superoperator mapping register $\reg{A}'$ to $\reg{B}$ satisfying $\Tr(X D(Y)) = \Tr(D^*(X) Y)$ for all operators $X,Y$. Viewing $E \otimes D^*$ as a superoperator mapping registers $\reg{A}_1\reg{A}_2$ to $\reg{B}_1 \reg{B}_2$ and letting $S_{\reg{B}_1 \reg{B}_2}$ denote the swap operator on registers $\reg{B}_1\reg{B}_2$ the above is equal to
    \[
    \Tr \Big ( S_{\reg{B}_1 \reg{B}_2} (E \otimes D^*)( \int_{\vartheta} \ketbra{\vartheta}{\vartheta}^{\otimes 2} \mathrm{d}\vartheta) \Big)~.
    \]
    Now, it is well-known~\cite{harrow2013church} that the integral over two copies of an $m(\lambda)$-qubit Haar-random state is proportional to the projector $\frac{1}{2} (\id + S)$ onto the \emph{symmetric subspace} of $(\C^M)^{\otimes 2}$. The dimension of the projector is $M(M+1)/2$. Thus the above is equal to
    \begin{align*}
        &\frac{1}{M(M+1)} \Tr \Big ( S_{\reg{B}_1 \reg{B}_2} \, (E \otimes D^*)(\id_{\reg{A}_1 \reg{A}_2} + S_{\reg{A}_1 \reg{A}_2}) \Big) \\
        &\qquad \leq  \frac{1}{M(M+1)}  \Tr \Big ((E \otimes D^*)(\id_{\reg{A}_1 \reg{A}_2} + S_{\reg{A}_1 \reg{A}_2}) \Big) \\
        &\qquad = \frac{1}{M(M+1)}   \Big [  \Tr \Big((E \otimes D^*)(\id_{\reg{A}_1 \reg{A}_2}) \Big) + \Tr \Big((E \otimes D^*)(S_{\reg{A}_1 \reg{A}_2}) \Big) \Big ] \\
        &\qquad = \frac{1}{M(M+1)}   \Big [  \Tr \Big(\id_{\reg{A}_1} \ot D^*(\id_{\reg{A}_2}) \Big) + \Tr \Big((\id_{\reg{A}_1} \ot D^*)(S_{\reg{A}_1 \reg{A}_2}) \Big) \Big ] \\
        &\qquad  = \frac{1}{M(M+1)}   \Big [  \Tr \Big(\id_{\reg{A}_1} \Big) \, \Tr \Big( D^*(\id_{\reg{A}_2}) \Big) + \Tr \Big(D^*(\id_{\reg{A}_2}) \Big) \Big ] \\
        &\qquad  =\frac{1}{M(M+1)}   \Big [  R M  + R \Big ] \\
        &\qquad  = R/M = 2^{-(m(\lambda) - \lambda - O(\log 1/\eps))} \leq \frac{1}{4}~.
    \end{align*}
    The second line follows from the fact that $|\Tr(A^\dagger B)| \leq \| A \|_\infty \| B \|_1$ for all operators $A,B$ and $\norm{S}_{\infty} \leq 1$. The fourth line follows from the fact that $E$ is a trace-preserving superoperator. The sixth line follows from the fact that since $D$ is a channel that takes as input $\reg{B}$, $\Tr(D^*(\id_{\reg{A}_2})) = \Tr(\id_{\reg{B}}) = R$. The last line follows because our assumption about the stretch of the PRS. This shows that the acceptance probability of $A_\lambda$ given a Haar random state and access to its reflection oracle is at most $\frac{1}{4}$, which is less than $\frac{1}{2} - \negl(\lambda)$ for sufficiently large $\lambda$. 

    Thus we have arrived at a contradiction. There is no polynomial-time pair of algorithms that compresses to the $\eps$-smoothed max entropy. 
\end{proof}

We compare our hardness result with the upper bound proved in \Cref{thm:comp-compression}. As an example, let $\eps(n) = 2^{-\log^2(n)}$, which is a negligible function. Then roughly, if $\DistUhlmann_{1 - \eps}$ is easy, then compressing to $H^{\eps}_{\max}(\rho) + O(\log 1/\eps) = H^{\eps}_{\max}(\rho) + O(\log^2(n))$ is easy. On the other hand, the lower bound shows that if PRS generators with output length $m(\lambda) \geq \lambda + \Omega(\log^2(\lambda))$ exist, then compressing to $H^{\eps}_{\max}(\rho) + O(\log^2(n))$ is not easy. 

We remark that it should be possible to base the lower bound on seemingly weaker assumptions, such as one-way state generators~\cite{morimae2022quantum}. However, ideally we would be able to base the hardness on an assumption such as the existence of quantum commitments or the hardness of the Uhlmann transformation problem, which would give a true converse to the upper bound of \Cref{thm:comp-compression}. However the main issue is \emph{verifiability}: with pseudorandom states or one-way state generators (with pure-state outputs), one can check whether the state has been compressed and decompressed; it is not clear whether this is possible with quantum commitments. We leave it as an open problem to prove matching upper and lower complexity bounds on compression. 

\begin{openproblem}
    Is the complexity of optimal compression equivalent to the complexity of the Uhlmann Transformation Problem?
\end{openproblem}

There are many more that have been studied  information-theoretically (including a whole family tree of them~\cite{abeyesinghe2009mother}), and one can ask about the complexity of each of these tasks. 
\begin{openproblem}
    What is the complexity of other quantum Shannon theory tasks, such as achieving capacity over a noisy channel, entanglement distillation, or quantum state redistribution? 
\end{openproblem}

 \section{Black-Hole Radiation Decoding}
\label{sec:gravity}

In this section, we discuss connections between the Uhlmann Transformation Problem and computational tasks motivated by questions in high-energy physics. We focus on the \emph{black hole radiation decoding task}, which was introduced by Harlow and Hayden~\cite{Harlow_2013}. We argue that the complexity of this task is characterized by the complexity of the distributional Uhlmann Transformation Problem.

The black hole radiation decoding task is motivated by the following thought experiment of Almheiri, Marolf, Polchinski, Sully~\cite{almheiri2013black}: imagine that Alice creates a maximally entangled pair of qubits $\ket{\mathrm{EPR}} = \frac{1}{\sqrt{2}} (\ket{00} + \ket{11})$ and throws one half into a newly-formed black hole. After a long time, Alice could potentially decode the Hawking radiation of the black hole and recover the qubit she threw in. However, Alice could then jump into the black hole and find another qubit that is supposed to be maximally entangled with the qubit that was not thrown in -- witnessing a violation of the monogamy of entanglement. These conclusions were derived assuming supposedly uncontroversial principles of quantum field theory and general relativity. 

Harlow and Hayden proposed a resolution to this paradox via a computational complexity argument~\cite{Harlow_2013}: it may not be \emph{feasible} for Alice to decode the black hole's Hawking radiation in any reasonable amount of time --- by the time she decodes the qubit that she threw in, the black hole may have evaporated anyways! They argued that, assuming $\mathsf{SZK} \not\subseteq \mathsf{BQP}$ -- note that these are classes of \emph{decision} problems --- a formulation of the black hole radiation decoding task cannot be done in polynomial time. 

What about the converse? That is, does a traditional complexity class statement such as $\mathsf{SZK} \subseteq \mathsf{BQP}$ imply that the black hole radiation decoding task is solvable in polynomial time? As pointed out by Aaronson~\cite{aaronson2016complexity}, it is not even clear that the black hole radiation decoding task is easy even if we assume $\mathsf{P} = \mathsf{PSPACE}$. As with all the other ``fully quantum'' tasks considered in this paper, it appears difficult to characterize the complexity of the black hole decoding problem in terms of traditional notions from complexity theory. 

Brakerski recently gave a characterization of the hardness of the black hole radiation task in terms of the existence of a cryptographic primitive known as \emph{quantum EFI pairs}~\cite{brakerski2022blackhole}, which are in turn equivalent to quantum commitments (as well as many other quantum cryptographic primitives, see~\cite{brakerski2022computational} for an in-depth discussion). Given the discussion in \Cref{sec:qcrypto} that connects quantum commitments with the Uhlmann Transformation Problem, one would then expect an equivalence between black hole radiation decoding and the Uhlmann Transformation Problem. 

We spell out this equivalence by showing that the complexity of the black hole radiation decoding task is the same as the complexity of the Decodable Channel Problem, which we showed to be equivalent to the (distributional) Uhlmann Transformation Problem in \Cref{ssec:Noisy_Channel_Decoding}. We believe that the direct reduction to and from the Decodable Channel Problem is natural, and may be useful to those who are more comfortable with quantum Shannon theory.

We first describe a formulation of the black hole radiation decoding task, which is an adaptation of the formulations of~\cite{Harlow_2013,brakerski2022blackhole}.

\begin{figure}[h]
    \centering

\tikzset{
pattern size/.store in=\mcSize, 
pattern size = 5pt,
pattern thickness/.store in=\mcThickness, 
pattern thickness = 0.3pt,
pattern radius/.store in=\mcRadius, 
pattern radius = 1pt}
\makeatletter
\pgfutil@ifundefined{pgf@pattern@name@_vpfldcoi3}{
\makeatletter
\pgfdeclarepatternformonly[\mcRadius,\mcThickness,\mcSize]{_vpfldcoi3}
{\pgfpoint{-0.5*\mcSize}{-0.5*\mcSize}}
{\pgfpoint{0.5*\mcSize}{0.5*\mcSize}}
{\pgfpoint{\mcSize}{\mcSize}}
{
\pgfsetcolor{\tikz@pattern@color}
\pgfsetlinewidth{\mcThickness}
\pgfpathcircle\pgfpointorigin{\mcRadius}
\pgfusepath{stroke}
}}
\makeatother

\tikzset{
pattern size/.store in=\mcSize, 
pattern size = 5pt,
pattern thickness/.store in=\mcThickness, 
pattern thickness = 0.3pt,
pattern radius/.store in=\mcRadius, 
pattern radius = 1pt}
\makeatletter
\pgfutil@ifundefined{pgf@pattern@name@_cx6jy1sst}{
\makeatletter
\pgfdeclarepatternformonly[\mcRadius,\mcThickness,\mcSize]{_cx6jy1sst}
{\pgfpoint{-0.5*\mcSize}{-0.5*\mcSize}}
{\pgfpoint{0.5*\mcSize}{0.5*\mcSize}}
{\pgfpoint{\mcSize}{\mcSize}}
{
\pgfsetcolor{\tikz@pattern@color}
\pgfsetlinewidth{\mcThickness}
\pgfpathcircle\pgfpointorigin{\mcRadius}
\pgfusepath{stroke}
}}
\makeatother

\tikzset{
pattern size/.store in=\mcSize, 
pattern size = 5pt,
pattern thickness/.store in=\mcThickness, 
pattern thickness = 0.3pt,
pattern radius/.store in=\mcRadius, 
pattern radius = 1pt}
\makeatletter
\pgfutil@ifundefined{pgf@pattern@name@_a5dj9q0fh}{
\makeatletter
\pgfdeclarepatternformonly[\mcRadius,\mcThickness,\mcSize]{_a5dj9q0fh}
{\pgfpoint{-0.5*\mcSize}{-0.5*\mcSize}}
{\pgfpoint{0.5*\mcSize}{0.5*\mcSize}}
{\pgfpoint{\mcSize}{\mcSize}}
{
\pgfsetcolor{\tikz@pattern@color}
\pgfsetlinewidth{\mcThickness}
\pgfpathcircle\pgfpointorigin{\mcRadius}
\pgfusepath{stroke}
}}
\makeatother
\tikzset{every picture/.style={line width=0.75pt}} 

\begin{tikzpicture}[x=0.75pt,y=0.75pt,yscale=-0.8,xscale=0.8]

\draw  [pattern=_vpfldcoi3,pattern size=14.774999999999999pt,pattern thickness=0.75pt,pattern radius=0.75pt, pattern color={rgb, 255:red, 0; green, 0; blue, 0}] (3,126.85) .. controls (3,87.28) and (35.08,55.2) .. (74.64,55.2) .. controls (114.21,55.2) and (146.28,87.28) .. (146.28,126.85) .. controls (146.28,166.41) and (114.21,198.49) .. (74.64,198.49) .. controls (35.08,198.49) and (3,166.41) .. (3,126.85) -- cycle ;
\draw  [fill={rgb, 255:red, 255; green, 255; blue, 255 }  ,fill opacity=1 ] (28.42,126.85) .. controls (28.42,102.81) and (49.12,83.32) .. (74.64,83.32) .. controls (100.17,83.32) and (120.86,102.81) .. (120.86,126.85) .. controls (120.86,150.88) and (100.17,170.37) .. (74.64,170.37) .. controls (49.12,170.37) and (28.42,150.88) .. (28.42,126.85) -- cycle ;
\draw [line width=1.5]    (168.62,12.45) .. controls (168.58,14.81) and (167.38,15.97) .. (165.02,15.92) .. controls (162.67,15.88) and (161.47,17.04) .. (161.42,19.39) .. controls (161.37,21.74) and (160.17,22.9) .. (157.82,22.85) .. controls (155.47,22.81) and (154.27,23.97) .. (154.22,26.32) .. controls (154.17,28.67) and (152.97,29.83) .. (150.62,29.79) .. controls (148.27,29.75) and (147.07,30.91) .. (147.02,33.26) .. controls (146.97,35.62) and (145.77,36.78) .. (143.41,36.73) .. controls (141.06,36.68) and (139.86,37.84) .. (139.81,40.19) .. controls (139.76,42.54) and (138.56,43.7) .. (136.21,43.66) .. controls (133.86,43.62) and (132.66,44.78) .. (132.61,47.13) .. controls (132.56,49.48) and (131.36,50.64) .. (129.01,50.6) .. controls (126.66,50.56) and (125.46,51.72) .. (125.41,54.07) .. controls (125.36,56.43) and (124.16,57.59) .. (121.8,57.54) .. controls (119.45,57.49) and (118.25,58.65) .. (118.2,61) .. controls (118.15,63.35) and (116.95,64.51) .. (114.6,64.47) .. controls (112.25,64.43) and (111.05,65.59) .. (111,67.94) .. controls (110.95,70.29) and (109.75,71.45) .. (107.4,71.41) .. controls (105.05,71.37) and (103.85,72.53) .. (103.8,74.88) .. controls (103.75,77.24) and (102.55,78.4) .. (100.19,78.35) .. controls (97.84,78.3) and (96.64,79.46) .. (96.59,81.81) .. controls (96.54,84.16) and (95.34,85.32) .. (92.99,85.28) .. controls (90.64,85.24) and (89.44,86.4) .. (89.39,88.75) .. controls (89.34,91.1) and (88.14,92.26) .. (85.79,92.22) -- (85.43,92.57) -- (85.43,92.57) ;
\draw  [fill={rgb, 255:red, 0; green, 0; blue, 0 }  ,fill opacity=1 ] (167.85,10.91) .. controls (167.85,8.57) and (169.75,6.67) .. (172.09,6.67) .. controls (174.43,6.67) and (176.33,8.57) .. (176.33,10.91) .. controls (176.33,13.25) and (174.43,15.15) .. (172.09,15.15) .. controls (169.75,15.15) and (167.85,13.25) .. (167.85,10.91) -- cycle ;
\draw  [fill={rgb, 255:red, 0; green, 0; blue, 0 }  ,fill opacity=1 ] (81.19,92.57) .. controls (81.19,90.23) and (83.09,88.33) .. (85.43,88.33) .. controls (87.77,88.33) and (89.66,90.23) .. (89.66,92.57) .. controls (89.66,94.91) and (87.77,96.8) .. (85.43,96.8) .. controls (83.09,96.8) and (81.19,94.91) .. (81.19,92.57) -- cycle ;
\draw  [pattern=_cx6jy1sst,pattern size=14.774999999999999pt,pattern thickness=0.75pt,pattern radius=0.75pt, pattern color={rgb, 255:red, 0; green, 0; blue, 0}] (232,127.36) .. controls (232,87.79) and (264.08,55.72) .. (303.64,55.72) .. controls (343.21,55.72) and (375.28,87.79) .. (375.28,127.36) .. controls (375.28,166.93) and (343.21,199) .. (303.64,199) .. controls (264.08,199) and (232,166.93) .. (232,127.36) -- cycle ;
\draw  [fill={rgb, 255:red, 255; green, 255; blue, 255 }  ,fill opacity=1 ] (258.42,133.44) .. controls (258.42,116.63) and (272.9,103) .. (290.75,103) .. controls (308.6,103) and (323.08,116.63) .. (323.08,133.44) .. controls (323.08,150.26) and (308.6,163.89) .. (290.75,163.89) .. controls (272.9,163.89) and (258.42,150.26) .. (258.42,133.44) -- cycle ;
\draw [line width=1.5]    (397.62,12.96) .. controls (397.68,15.32) and (396.53,16.53) .. (394.17,16.58) .. controls (391.82,16.63) and (390.66,17.84) .. (390.71,20.19) .. controls (390.76,22.54) and (389.61,23.75) .. (387.26,23.8) .. controls (384.9,23.85) and (383.75,25.06) .. (383.8,27.42) .. controls (383.85,29.77) and (382.7,30.98) .. (380.35,31.03) .. controls (378,31.08) and (376.84,32.29) .. (376.89,34.64) .. controls (376.94,37) and (375.79,38.21) .. (373.43,38.26) .. controls (371.08,38.31) and (369.93,39.52) .. (369.98,41.87) .. controls (370.03,44.23) and (368.88,45.44) .. (366.52,45.49) .. controls (364.17,45.54) and (363.02,46.75) .. (363.07,49.1) .. controls (363.12,51.45) and (361.96,52.66) .. (359.61,52.71) .. controls (357.26,52.77) and (356.11,53.98) .. (356.16,56.33) .. controls (356.21,58.68) and (355.05,59.89) .. (352.7,59.94) .. controls (350.35,59.99) and (349.19,61.2) .. (349.24,63.55) .. controls (349.29,65.9) and (348.14,67.11) .. (345.79,67.17) .. controls (343.44,67.22) and (342.28,68.43) .. (342.33,70.78) .. controls (342.38,73.13) and (341.23,74.34) .. (338.88,74.39) .. controls (336.52,74.44) and (335.37,75.65) .. (335.42,78.01) .. controls (335.47,80.36) and (334.32,81.57) .. (331.97,81.62) .. controls (329.62,81.67) and (328.46,82.88) .. (328.51,85.23) .. controls (328.56,87.59) and (327.41,88.8) .. (325.05,88.85) .. controls (322.7,88.9) and (321.55,90.11) .. (321.6,92.46) .. controls (321.65,94.82) and (320.5,96.03) .. (318.14,96.08) .. controls (315.79,96.13) and (314.64,97.34) .. (314.69,99.69) .. controls (314.74,102.04) and (313.58,103.25) .. (311.23,103.3) .. controls (308.87,103.35) and (307.72,104.56) .. (307.77,106.92) .. controls (307.82,109.27) and (306.67,110.48) .. (304.32,110.53) -- (301,114) -- (301,114) ;
\draw  [fill={rgb, 255:red, 0; green, 0; blue, 0 }  ,fill opacity=1 ] (396.85,11.42) .. controls (396.85,9.08) and (398.75,7.19) .. (401.09,7.19) .. controls (403.43,7.19) and (405.33,9.08) .. (405.33,11.42) .. controls (405.33,13.76) and (403.43,15.66) .. (401.09,15.66) .. controls (398.75,15.66) and (396.85,13.76) .. (396.85,11.42) -- cycle ;
\draw [line width=1.5]    (356,140) .. controls (354.13,141.45) and (352.48,141.24) .. (351.04,139.37) .. controls (349.6,137.5) and (347.95,137.29) .. (346.08,138.73) .. controls (344.21,140.18) and (342.56,139.97) .. (341.12,138.1) .. controls (339.68,136.23) and (338.03,136.02) .. (336.16,137.47) .. controls (334.29,138.91) and (332.64,138.7) .. (331.2,136.83) .. controls (329.76,134.96) and (328.11,134.75) .. (326.24,136.2) .. controls (324.37,137.65) and (322.72,137.44) .. (321.28,135.57) .. controls (319.84,133.7) and (318.19,133.49) .. (316.32,134.93) .. controls (314.45,136.38) and (312.8,136.17) .. (311.36,134.3) -- (309,134) -- (309,134) ;
\draw [line width=1.5]    (356,140) .. controls (354.98,137.88) and (355.53,136.31) .. (357.65,135.28) .. controls (359.78,134.26) and (360.33,132.69) .. (359.31,130.56) .. controls (358.29,128.44) and (358.84,126.87) .. (360.96,125.85) .. controls (363.09,124.83) and (363.64,123.26) .. (362.62,121.13) .. controls (361.6,119.01) and (362.15,117.44) .. (364.27,116.41) .. controls (366.4,115.39) and (366.95,113.82) .. (365.93,111.69) .. controls (364.91,109.57) and (365.46,108) .. (367.58,106.97) .. controls (369.71,105.95) and (370.26,104.38) .. (369.24,102.25) .. controls (368.22,100.13) and (368.77,98.56) .. (370.89,97.54) .. controls (373.02,96.52) and (373.57,94.95) .. (372.55,92.82) .. controls (371.53,90.7) and (372.08,89.13) .. (374.2,88.1) .. controls (376.33,87.08) and (376.88,85.51) .. (375.86,83.38) .. controls (374.84,81.26) and (375.39,79.69) .. (377.51,78.66) .. controls (379.64,77.64) and (380.19,76.07) .. (379.17,73.94) .. controls (378.15,71.82) and (378.7,70.25) .. (380.82,69.23) .. controls (382.94,68.2) and (383.49,66.63) .. (382.47,64.51) .. controls (381.45,62.38) and (382,60.81) .. (384.13,59.79) .. controls (386.25,58.76) and (386.8,57.19) .. (385.78,55.07) .. controls (384.76,52.94) and (385.31,51.37) .. (387.44,50.35) .. controls (389.56,49.32) and (390.11,47.75) .. (389.09,45.63) .. controls (388.07,43.5) and (388.62,41.93) .. (390.75,40.92) .. controls (392.87,39.89) and (393.42,38.32) .. (392.4,36.2) .. controls (391.38,34.07) and (391.93,32.5) .. (394.06,31.48) .. controls (396.18,30.45) and (396.73,28.88) .. (395.71,26.76) .. controls (394.69,24.63) and (395.24,23.06) .. (397.37,22.04) .. controls (399.49,21.02) and (400.04,19.45) .. (399.02,17.33) .. controls (398,15.2) and (398.55,13.63) .. (400.68,12.61) -- (401.09,11.42) -- (401.09,11.42) ;
\draw [line width=1.5]    (292,177) .. controls (290.41,175.25) and (290.49,173.59) .. (292.24,172.01) .. controls (293.99,170.42) and (294.07,168.76) .. (292.48,167.01) .. controls (290.89,165.27) and (290.97,163.61) .. (292.71,162.02) .. controls (294.46,160.43) and (294.54,158.77) .. (292.95,157.02) -- (293,156) -- (293,156) ;
\draw [line width=1.5]    (276,118) .. controls (273.64,118) and (272.46,116.82) .. (272.46,114.46) .. controls (272.46,112.11) and (271.28,110.93) .. (268.93,110.93) .. controls (266.57,110.93) and (265.39,109.75) .. (265.39,107.39) .. controls (265.39,105.04) and (264.21,103.86) .. (261.86,103.86) -- (259,101) -- (259,101) ;
\draw [line width=1.5]    (48,105) .. controls (45.64,105) and (44.46,103.82) .. (44.46,101.46) .. controls (44.46,99.11) and (43.28,97.93) .. (40.93,97.93) .. controls (38.57,97.93) and (37.39,96.75) .. (37.39,94.39) .. controls (37.39,92.04) and (36.21,90.86) .. (33.86,90.86) -- (31,88) -- (31,88) ;
\draw [line width=1.5]    (130,147) .. controls (127.89,148.05) and (126.31,147.53) .. (125.26,145.42) .. controls (124.2,143.31) and (122.62,142.79) .. (120.51,143.84) .. controls (118.4,144.89) and (116.82,144.37) .. (115.77,142.26) .. controls (114.72,140.15) and (113.14,139.63) .. (111.03,140.68) .. controls (108.92,141.73) and (107.33,141.2) .. (106.28,139.09) -- (106,139) -- (106,139) ;
\draw [line width=1.5]    (55,179) .. controls (53.95,176.89) and (54.47,175.31) .. (56.58,174.26) .. controls (58.69,173.2) and (59.21,171.62) .. (58.16,169.51) .. controls (57.11,167.4) and (57.63,165.82) .. (59.74,164.77) .. controls (61.85,163.72) and (62.37,162.14) .. (61.32,160.03) -- (62,158) -- (62,158) ;
\draw  [pattern=_a5dj9q0fh,pattern size=14.774999999999999pt,pattern thickness=0.75pt,pattern radius=0.75pt, pattern color={rgb, 255:red, 0; green, 0; blue, 0}] (460,128.36) .. controls (460,88.79) and (492.08,56.72) .. (531.64,56.72) .. controls (571.21,56.72) and (603.28,88.79) .. (603.28,128.36) .. controls (603.28,167.93) and (571.21,200) .. (531.64,200) .. controls (492.08,200) and (460,167.93) .. (460,128.36) -- cycle ;
\draw  [fill={rgb, 255:red, 255; green, 255; blue, 255 }  ,fill opacity=1 ] (486.42,134.44) .. controls (486.42,117.63) and (500.9,104) .. (518.75,104) .. controls (536.6,104) and (551.08,117.63) .. (551.08,134.44) .. controls (551.08,151.26) and (536.6,164.89) .. (518.75,164.89) .. controls (500.9,164.89) and (486.42,151.26) .. (486.42,134.44) -- cycle ;
\draw  [fill={rgb, 255:red, 0; green, 0; blue, 0 }  ,fill opacity=1 ] (624.85,12.42) .. controls (624.85,10.08) and (626.75,8.19) .. (629.09,8.19) .. controls (631.43,8.19) and (633.33,10.08) .. (633.33,12.42) .. controls (633.33,14.76) and (631.43,16.66) .. (629.09,16.66) .. controls (626.75,16.66) and (624.85,14.76) .. (624.85,12.42) -- cycle ;
\draw [line width=1.5]    (554,103) .. controls (553.86,105.35) and (552.61,106.46) .. (550.26,106.32) .. controls (547.91,106.18) and (546.66,107.29) .. (546.53,109.64) .. controls (546.39,111.99) and (545.14,113.1) .. (542.79,112.97) .. controls (540.44,112.83) and (539.19,113.94) .. (539.05,116.29) -- (536,119) -- (536,119) ;
\draw [line width=1.5]    (616,84.24) .. controls (614.66,82.3) and (614.96,80.66) .. (616.9,79.32) .. controls (618.84,77.98) and (619.14,76.34) .. (617.79,74.4) .. controls (616.45,72.46) and (616.75,70.82) .. (618.69,69.48) .. controls (620.63,68.14) and (620.93,66.5) .. (619.59,64.56) .. controls (618.24,62.62) and (618.54,60.98) .. (620.48,59.64) .. controls (622.42,58.3) and (622.72,56.66) .. (621.38,54.72) .. controls (620.04,52.78) and (620.34,51.14) .. (622.28,49.8) .. controls (624.21,48.46) and (624.51,46.82) .. (623.17,44.89) .. controls (621.83,42.95) and (622.13,41.31) .. (624.07,39.97) .. controls (626.01,38.63) and (626.31,36.99) .. (624.97,35.05) .. controls (623.62,33.11) and (623.92,31.47) .. (625.86,30.13) .. controls (627.8,28.79) and (628.1,27.15) .. (626.76,25.21) .. controls (625.42,23.27) and (625.72,21.63) .. (627.66,20.29) .. controls (629.6,18.95) and (629.9,17.31) .. (628.55,15.37) -- (629.09,12.42) -- (629.09,12.42) ;
\draw [line width=1.5]    (520,178) .. controls (518.41,176.25) and (518.49,174.59) .. (520.24,173.01) .. controls (521.99,171.42) and (522.07,169.76) .. (520.48,168.01) .. controls (518.89,166.27) and (518.97,164.61) .. (520.71,163.02) .. controls (522.46,161.43) and (522.54,159.77) .. (520.95,158.02) -- (521,157) -- (521,157) ;
\draw [line width=1.5]    (504,119) .. controls (501.64,119) and (500.46,117.82) .. (500.46,115.46) .. controls (500.46,113.11) and (499.28,111.93) .. (496.93,111.93) .. controls (494.57,111.93) and (493.39,110.75) .. (493.39,108.39) .. controls (493.39,106.04) and (492.21,104.86) .. (489.86,104.86) -- (487,102) -- (487,102) ;
\draw  [fill={rgb, 255:red, 0; green, 0; blue, 0 }  ,fill opacity=1 ] (611.76,84.24) .. controls (611.76,81.9) and (613.66,80) .. (616,80) .. controls (618.34,80) and (620.24,81.9) .. (620.24,84.24) .. controls (620.24,86.58) and (618.34,88.47) .. (616,88.47) .. controls (613.66,88.47) and (611.76,86.58) .. (611.76,84.24) -- cycle ;

\draw (50.66,123.42) node [anchor=north west][inner sep=0.75pt]    {$H_{\mathrm{early}}$};
\draw (116.46,193.89) node [anchor=north west][inner sep=0.75pt]    {$R_{\mathrm{early}}$};
\draw (176.49,0.23) node [anchor=north west][inner sep=0.75pt]    {$B$};
\draw (85.25,93.44) node [anchor=north west][inner sep=0.75pt]    {$A$};
\draw (267.66,129.93) node [anchor=north west][inner sep=0.75pt]    {$H_{\mathrm{old}}$};
\draw (336.46,194.4) node [anchor=north west][inner sep=0.75pt]    {$R_{\mathrm{old}}$};
\draw (405.49,0.74) node [anchor=north west][inner sep=0.75pt]    {$B$};
\draw (499.66,126.93) node [anchor=north west][inner sep=0.75pt]    {$H_{\mathrm{old}}$};
\draw (564.46,195.4) node [anchor=north west][inner sep=0.75pt]    {$R'_{\mathrm{old}}$};
\draw (633.49,1.74) node [anchor=north west][inner sep=0.75pt]    {$B$};
\draw (609.25,90.44) node [anchor=north west][inner sep=0.75pt]    {$A$};
\draw (67,7) node [anchor=north west][inner sep=0.75pt]   [align=left] {(a)};
\draw (292,7) node [anchor=north west][inner sep=0.75pt]   [align=left] {(b)};
\draw (521,9) node [anchor=north west][inner sep=0.75pt]   [align=left] {(c)};

\end{tikzpicture}
     \caption{Decoding black hole radiation. (a) Qubit $A$, maximally entangled with qubit $B$, falls into an early black hole $H_{\mathrm{early}}$, which is entangled with some early Hawking radiation $R_{\mathrm{early}}$. (b) After evaporating much of its mass, the old black hole $H_{\mathrm{old}}$ is entangled with the radiation $R_{\mathrm{old}}$ which is entangled with the qubit $B$. (c) By performing a computation on the radiation only, the partner qubit $A$ can be decoded.}
    \label{fig:black-holes}
\end{figure}
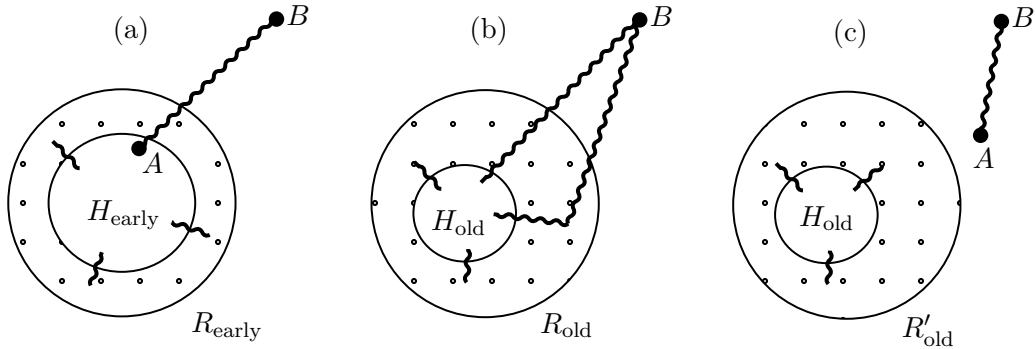

\begin{definition}[Decodable black hole states]
    Let $P$ denote a unitary quantum circuit mapping registers $\reg{A} \reg{G}$ to $\reg{H} \reg{R}$ where $\reg{A}$ is a single qubit register. Consider the state 
    \[
        \ket{\psi}_{\reg{BHR}} \deq (\id_{\reg{B}} \ot P_{\reg{AG} \to \reg{HR}}) \ket{\mathrm{EPR}}_{\reg{BA}} \ot \ket{0}_{\reg{G}}~.
    \]
    We say that \emph{$\ket{\psi}$ is an $\eps$-decodable black hole state} if 
    there exists a quantum circuit $D$ that takes as input register $\reg{R}$ and outputs a qubit labelled $\reg{A}$, such that letting $\rho_{\reg{HBA}}$ denote the state $(\id \ot D)(\ketbra{\psi}{\psi})$, we have
\[
    \fidelity \Big( \ketbra{\mathrm{EPR}}{\mathrm{EPR}}_{\reg{AB}} \, ,\, \rho_{\reg{AB}} \Big ) \geq 1 - \eps
    \]
    i.e., measuring the registers $\reg{BA}$ in the Bell basis yields the state $\ket{\mathrm{EPR}} \deq \frac{1}{\sqrt{2}} (\ket{00} + \ket{11})$ with probability at least $1 - \eps$. We say that the circuit $D$ is a \emph{$\eps$-decoder} for the state $\ket{\psi}$.

\end{definition}
The circuit $P$ generating the decodable black hole state can be thought of as a unitary that encodes the laws of black hole evolution: given a qubit in register $\reg{A}$ and a fixed number of ancilla qubits, it forms a black hole in register $\reg{H}$ as well as the outgoing Hawking radiation in register $\reg{R}$. The decodability condition implies that, by acting on the radiation only, it is information-theoretically possible to decode the original qubit that was input. \noindent See \Cref{fig:black-holes} for an illustration of black hole radiation decoding. We formalize black hole radiation decoding as a computational task. 
\begin{definition}[Black hole radiation decoding task]\label{def:BHRD}
    Let $\eps(n),\delta(n)$ be functions. We say that a quantum algorithm $D = (D_x)_x$ solves the \emph{$\eps$-black hole radiation decoding task with error $\delta$} if for all $x = (1^n,P)$ where $P$ is a unitary quantum circuit acting on $n$ qubits and gives rise to an $\eps(n)$-decodable black hole state $\ket{\psi}$, the circuit $D_x$ is a $\delta(n)$-decoder for $\ket{\psi}$. 
\end{definition}

We now prove that the task of black hole radiation decoding in \Cref{def:BHRD} is equivalent to the Decodable Channel Problem in \Cref{def:decodable-channel-problem}, which results in the following theorem.
\begin{theorem}
\label{thm:black-holes}
Let $\eps(n)$ be a negligible function. 
$\DistUhlmann_{1-\eps}$ is solvable in polynomial-time with inverse polynomial error if and only if the $\eps(n)$-black hole radiation decoding task is solvable in polynomial-time with inverse polynomial error.
\end{theorem}

\begin{proof}
We prove this via reduction to the Decodable Channel Problem described in \Cref{ssec:Noisy_Channel_Decoding}. First, observe (from the proof) that the statement in \Cref{thm:complexity-decodable-channels} still holds when considering instances of the $\eps$-Decodable Channel Problem of the form $y=(1^1,1^r,C)$, i.e., where we restrict $C$ to single qubit inputs only. 
Define the following bijection $\varphi$: for every $x = (1^n,P)$, where $P: \reg{AG} \rightarrow \reg{HR}$ is a unitary quantum circuit acting on $n$ qubits and where $r$ is the size of the register $\reg R$, define $\varphi(x) = (1^1,1^r,\tilde{P})$, where $\tilde{P}$ is the quantum circuit first appends $n-1$ qubits initialized to $\ket{0}$ to its input and then runs $P$. 

It is clear that $x$ corresponds to an $\eps$-decodable black hole state if and only if $\varphi(x)$ corresponds to an $\eps$-decodable channel: the channel can be viewed as taking the input qubit, dumping it in the black hole, and the outputting the radiation emitted by the black hole. Decoding the EPR pair from the channel associated with $\tilde{P}$ exactly corresponds to decoding the EPR pair from the black hole associated with $P$. Therefore, the claim follows from \Cref{thm:complexity-decodable-channels}, which shows that the complexity of the Decodable Channel Problem is equivalent to the complexity of $\DistUhlmann$. 
\end{proof}

\begin{remark}
We remark that Brakerski proved a stronger theorem by relating the black hole radiation task to EFI~\cite{brakerski2022computational}. For simplicity, we focus on the task of decoding the EPR pair with fidelity $1 - \eps$, for a small $\epsilon$, whereas Brakerski~\cite{brakerski2022blackhole} used amplification to boost weak decoders that succeed with fidelity much smaller than $1$.
\end{remark}

\newpage
\part{Appendix}
\label{part:appendices}

\appendix

\section{Weak polarization for Uhlmann transformations}
\label{sec:polarization}

In this section we prove the following weak polarization lemma, which we can interpret as evidence for a stronger polarization statement, \Cref{conj:polarization}.

\weakpolarization*

\vspace{5pt}

At a high level, the algorithm operates as follows. Given a $\Uhlmann_\kappa$ instance  $x = (1^n,C,D)$, it computes the description of an instance $y = (1^k,F_+,F_-)$ of $\Uhlmann_{1 - 2^{-p(n)}}$, resulting from applying parallel repetition, XOR repetition, and then parallel repetition again to the original instance $x$. We show that an algorithm that approximately implements the Uhlmann transform for $y$ can be transformed into an algorithm that approximately implements the Uhlmann transform for $x$. This makes use of reductions from the following results.

The first is a parallel repetition theorem for Uhlmann transforms, which is a special case of the parallel repetition theorem for $3$-message arguments proved by~\cite{BQSY24}:
\begin{theorem}[Efficient parallel repetition of Uhlmann transformations~\cite{BQSY24}]
\label{thm:parallel_repetition}
    Let $0 < \delta < 1$ and $\ket{C}$ and $\ket{D}$ be a pair of bipartite states prepared by circuits $C$ and $D$ respectively.  Suppose that there is a quantum algorithm $Q$ such that
    \begin{equation*}
        \fidelity(Q(\ketbra{C}{C}^{\ot t}), \ketbra{D}{D}^{\ot t}) \geq \delta^t~.
    \end{equation*}
    Then there is a $\poly(\delta^{-t},\eps^{-1})$-time query algorithm $R$ acting only on $\reg{B}$ (and additional ancilla), makes queries to the unitary purification of $Q$ and its inverse, $C^\dagger$, and $D$, such that
    \begin{equation*}
        \fidelity(R(\proj{C}), \proj{D}) \geq \delta - \epsilon\,.
    \end{equation*}
\end{theorem}

The second result is the so-called ``swapping-distinguishing duality'' theorem of Aaronson, Atia, and Susskind~\cite{aaronson2020hardness}:
\begin{theorem}[Swapping-distinguishing duality, Theorem 1 from~\cite{aaronson2020hardness}]
\label{thm:swapping_distinguishing_duality}
    \leavevmode
    \begin{enumerate}
        \item Let $\ket{x},\ket{y}$ be orthogonal states, and suppose that $\bra{y}U \ket{x} = a$ and $\bra{x} U \ket{y} = b$. Then using a single black-box call to controlled-$U$, plus $O(1)$ additional gates, we can distinguish $\ket{\psi} = \frac{1}{\sqrt{2}} (\ket{x} + \ket{y})$ and $\ket{\phi} = \frac{1}{\sqrt{2}}(\ket{x} - \ket{y})$ with bias $\Delta = \frac{1}{2} |a + b|$.
        \item Let $\ket{\psi},\ket{\phi}$ be orthogonal states, and suppose that a quantum circuit $A$ accepts $\ket{\psi}$ with probability $p$ and accepts $\ket{\phi}$ with probability $p - \Delta$. Then using a single black-box call to $A$ and $A^\dagger$, plus $O(1)$ additional gates, we can apply a unitary $U$ such that
        \[
            \frac{1}{2} \Big( |\bra{y} U \ket{x}| + |\bra{x} U \ket{y}| \Big) = \Delta
        \]
        where $\ket{x} = \frac{1}{\sqrt{2}} (\ket{\psi} + \ket{\phi})$ and $\ket{y} = \frac{1}{\sqrt{2}} (\ket{\psi} - \ket{\phi})$.
    \end{enumerate}
\end{theorem}

\DeclarePairedDelimiter{\ceil}{\lceil}{\rceil}
\begin{proof}[Proof of \Cref{thm:polarization}]
We present the following argument for general $\kappa$, even though the theorem is only specified for $\kappa = 1/2$. Set parameters
    \[
        m = \ceil[\bigg]{\frac{\ln 8p(n) }{\ln \frac{1}{1 - \eps/2}}} \quad \text{and} \quad t = \ceil[\bigg]{ \frac{8p(n)}{\kappa^m}}~.
    \]
    Note that this setting of parameters satisfies 
    \[
        \kappa^m \geq \frac{2p(n)}{t} \quad \text{and} \quad (\kappa - \eps/2)^m \leq \frac{1}{4t}~.
    \]
    
    We now describe the states $\ket{F_\pm}$ in more detail, and analyze their properties. First, define the pair of states
    \begin{gather*}
        \ket{E_+} = \frac{1}{\sqrt{2}} \left(\ket{C}^{\otimes m} \ket{0} + \ket{D}^{\otimes m}\ket{1}\right)_{\reg{A}^m \reg{B}^m \reg{O}} \\
        \ket{E_-} = \frac{1}{\sqrt{2}} \left(\ket{C}^{\otimes m} \ket{0} - \ket{D}^{\otimes m}\ket{1}\right)_{\reg{A}^m \reg{B}^m \reg{O}}~.
    \end{gather*}
    Note that $\ket{E_+},\ket{E_-}$ are orthogonal. Then define
    \begin{gather*}
        \ket{F_+} = \frac{1}{\sqrt{2}} \Big( \ket{E_+}^{\otimes t} + \ket{E_-}^{\otimes t} \Big)_{\reg{A}^{m \times t} \reg{B}^{m \times t} \reg{O}^{t}} \\
        \ket{F_-} = \frac{1}{\sqrt{2}} \Big( \ket{E_+}^{\otimes t} -\ket{E_-}^{\otimes t} \Big)_{\reg{A}^{m \times t} \reg{B}^{m \times t} \reg{O}^{t}}~.
    \end{gather*}
    Index the registers by $\reg{A}_{ij}$, $\reg{B}_{ij}$, and $\reg{O}_j$ where $1 \leq i \leq m$ and $1 \leq j \leq t$. Let $k = t(m+1)$ denote the number of qubits of $\ket{F_\pm}$. 

Let $F_+,F_-$ denote the circuits that prepare the states $\ket{F_+},\ket{F_-}$, respectively. Note that given the instance $x = (1^n,C,D)$, the instance $(1^k,F_+,F_-)$ can be computed in $\poly(m,t,n)$ time.

\begin{claim}
\label{claim:polarization_fidelity}
$(1^k,F_+,F_-)$ is an instance of $\Uhlmann_{1 - 2^{-p(n)}}$.
\end{claim}
\begin{proof}
    Let $\sigma_0$, $\sigma_1$ denote the reduced density matrices of $\ket{C}$ and $\ket{D}$ on $\reg{A}$ respectively. Then the observe that reduced density matrices of $\ket{C}^{\otimes m}$ and $\ket{D}^{\otimes m}$ on register $\reg{A}^m$ are $\sigma_0^{\otimes m}$ and $\sigma_1^{\otimes m}$, respectively. 

    Let $\tau_\pm$ denote the reduced density matrices of $\ket{F_\pm}$ on registers $\reg{A}^{m \times n}$ (i.e., the collection of registers $\reg{A}_{ij}$). A simple calculation shows that 
\begin{gather*}
    \tau_+ = \frac{1}{2^{t-1}} \sum_{x: |x| \text{ even}} \sigma_{x_1}^{\ot m} \otimes \sigma_{x_2}^{\ot m} \otimes \cdots \otimes \sigma_{x_t}^{\ot m} \\
    \tau_- = \frac{1}{2^{t-1}} \sum_{x: |x| \text{ odd}} \sigma_{x_1}^{\ot m} \otimes \sigma_{x_2}^{\ot m} \otimes \cdots \otimes \sigma_{x_t}^{\ot m}~.
\end{gather*}
Watrous~\cite[Proposition 6]{watrous2002limits} showed that $\td(\tau_+,\tau_-) = \td(\sigma_0^{\otimes m},\sigma_1^{\otimes m})^t$. Since $(1^n,C,D)$ is an instance of $\Uhlmann_\kappa$ and by Fuchs-van de Graaf, we have $\fidelity(\sigma_0^{\otimes m},\sigma_1^{\otimes m}) \geq \kappa^m$ and therefore
\[
    \td(\tau_+,\tau_-) \leq (1 - \kappa^m)^{t/2}~.
\]
By Fuchs-van de Graaf again we have
\[
    \fidelity(\tau_+,\tau_-) \geq (1 - \td(\tau_+,\tau_-))^2 \geq 1 - 2\td(\tau_+,\tau_-) \geq 1 - 2(1 - \kappa^m)^{t/2}~.
\]
Using that $1 - x \leq e^{-x}$ and our choice of $m,t$ we have that
\begin{align*}
    \fidelity(\tau_+,\tau_-) \geq 1 - \exp \Big( -\frac{t}{2} \kappa^m \Big) \geq 1 - e^{-p(n)}~.
\end{align*}

\end{proof}

\begin{claim}
\label{claim:polarization_error}
    Suppose there was a quantum algorithm $Q$ acting only on registers $\reg{B}_{ij} \reg{O}_j$ such that 
    \[
        \td((\id \otimes Q)(\ketbra{F_+}{F_+}),\ketbra{F_-}{F_-}) \leq \frac{1}{32}~.
    \]
Then there exists a quantum algorithm $A$ that for all $\eps$ can runs in $\poly(t,m,1/\eps)$ time, makes queries to the unitary purification of $Q$ and its inverse, acts on $\reg{B}$ (plus ancillas), such that 
    \[
        \fidelity((\id \otimes A)(\ketbra{C}{C}),\ketbra{D}{D}) \geq (4t)^{-1/m} - \eps/2~.
    \]
\end{claim}
\begin{proof}
Let $\delta = 1/32$. 
    Let $\wt{Q}$ denote the unitary purification of $Q$. By Uhlmann's theorem and Fuchs-van de Graaf, there exists a pure state $\ket{\eta}$ such that 
    \[
        \Big( \bra{F_-} \otimes \bra{\eta} (\id \otimes \wt{Q}) \ket{F_+} \otimes \ket{0} \Big )^2 \geq (1 - \delta)^2~.
    \]
    Here, the circuit $\wt{Q}$ acts on the registers $\reg{B}_{ij} \reg{O}_j$ as well as the ancilla zeroes. Translating to Euclidean distance, we have
    \[
        \Big \| (\id \otimes \wt{Q}) \ket{F_+} \ket{0} - \ket{F_-} \ket{\eta} \Big \| \leq \sqrt{2\delta}~.
    \]
    The state $\ket{\eta}$ can be (approximately) prepared by using the circuit $\wt{Q}$, its inverse, and the circuits $F_+,F_-$ for preparing $\ket{F_+},\ket{F_-}$ respectively. Therefore there exists a circuit $R$ such that 
    \[
        \Big | \bra{F_-} \otimes \bra{0} (\id \otimes R) \ket{F_+} \otimes \ket{0} \Big |^2 \geq 1 - 8\delta \geq 3/4
    \]
    where we used $\delta = 1/32$.
        By \Cref{thm:swapping_distinguishing_duality}, since $\ket{E_\pm}^{\ot t} = \frac{1}{\sqrt{2}} \Big( \ket{F_+} \pm \ket{F_-} \Big)$ there exists an efficient quantum operation $S$ that makes a single call to $R$ and $R^\dagger$ that acts only on registers $\reg{B}_{ij} \reg{O}_j$ (plus ancillas) and distinguishes between $\ket{E_+}^{\otimes t}$ and $\ket{E_-}^{\otimes t}$ with advantage at least $3/8 \geq 1/4$. 

        By a hybrid argument, there exists a quantum operation $S'$ (whose complexity is that of $S$ plus the complexity of preparing $t-1$ copies of $\ket{E_+}$ or $\ket{E_-}$) that acts only on registers $\reg{B}^m \reg{O}$ and distinguishes between $\ket{E_+}$ and $\ket{E_-}$ with advantage at least $1/4t$.  

        By \Cref{thm:swapping_distinguishing_duality} again, there exists an operation $U$ acting only on registers $\reg{B}^m \reg{O}$ (plus ancillas) and makes a single call to $S'$ to transform $\ket{C}^{\otimes m} \ket{0}$ to have overlap at least $1/4t$ with $\ket{D}^{\otimes m}\ket{1}$. By \Cref{thm:parallel_repetition}, there exists an algorithm $A$ acting only on register $\reg{B}$ (plus ancillas) with running time $\poly(m,1/\eps)$, makes queries to $U$ and $U^\dagger$, and maps $\ket{C}$ to have squared overlap at least $(4t)^{-1/m} - \eps/2$ with $\ket{D}$. 

        Taking into account the complexity of $\tilde{Q},\tilde{R},S,S',U$, we get that $A$ runs in time $\poly(m,t,1/\eps)$. 
\end{proof}

Thus weak polarization theorem follows immediately from \Cref{claim:polarization_fidelity,claim:polarization_error}. The algorithm $A_x$ behaves as follows: it computes the description of $y = (1^k,F_+,F_-)$ in $\poly(m,t,n)$ time. By \Cref{claim:polarization_fidelity}, this is an instance of $\Uhlmann_{1 - 2^{-p(n)}}$. Then, invoke the algorithm $A$ from \Cref{claim:polarization_error} that queries the algorithm $Q_y$ which implements the Uhlmann transform corresponding to $y$ with error at most $1/32$. The resulting fidelity is at least $(4t)^{-1/m} - \eps/2$, which by our choice of parameters is at least $\kappa - \eps$. 
\end{proof}
 \section{Information-theoretic one-shot compression}
\label{sec:info-theoretic-compression}

In this section we prove \Cref{thm:info-theory-compression}, which for convenience we restate here:

\infotheorycompression*

This theorem shows that the smoothed max-entropy of a quantum state characterizes the extent to which it can be compressed in the one-shot setting. 

The smoothed max entropy is just one of a rich zoo of entropy measures that are used in the setting of non-asymptotic quantum information theory~\cite{tomamichel2012framework}. To prove \Cref{thm:info-theory-compression} we employ the following entropy measures:

\begin{definition}[Min-, max-, and R\'{e}nyi 2-entropy]\label{def:entropies}
    Let $\eps \geq 0$ and let $\psi_{\reg{AB}}$ be a density matrix on registers $\reg{AB}$.
    \begin{itemize}
        \item The \emph{min-entropy of register $\reg{A}$ conditioned on register $\reg{B}$ of the state $\psi$} is
        \[
            H_{\min}(\reg{A} | \reg{B})_\psi \deq -\log \inf_{\substack{\sigma \in \mathrm{Pos}(\reg{B}) : \psi_{\reg{AB}} \leq \id_{\reg{A}} \ot \sigma_{\reg{B}}}} \Tr(\sigma)
        \]
        The \emph{$\eps$-smoothed conditional min-entropy} is
        \[
            H^{\eps}_{\min}(\reg{A} | \reg{B})_\psi \deq \sup_{\sigma:P(\sigma,\psi) \leq \eps} H_{\min}(\reg{A} | \reg{B})_\sigma~,
        \]
        where $P(\sigma,\psi)$ is the purified distance (whose definition need not concern us, see \cite[Definition 3.15]{tomamichel2012framework}).
        \item The \emph{max-entropy of register $\reg{A}$ conditioned on register $\reg{B}$ of the state $\psi$} is
        \[
            H_{\max}(\reg{A} | \reg{B})_\psi \deq \sup_{\sigma \in \mathrm{Pos}(\reg{B}) : \Tr(\sigma) \leq 1} \log \| \sqrt{\psi_{\reg{AB}}} \sqrt{\id_{\reg{A}} \otimes \sigma_{\reg{B}}} \|_1^2~.
        \]
        The \emph{$\eps$-smoothed conditional max-entropy} is 
        \[
            H^{\eps}_{\max}(\reg{A} | \reg{B})_\psi \deq \inf_{\sigma:\td(\sigma,\psi) \leq \eps} H_{\max}(\reg{A} | \reg{B})_\sigma~.
        \]
        \item The \emph{R\'{e}nyi $2$-entropy of register $\reg{A}$ conditioned on register $\reg{B}$ of the state $\psi$} is~\cite[Definition 2.11]{dupuis2010decoupling}
        \[
            H_2(\reg{A} | \reg{B})_\psi \deq - \log \inf_{\sigma>0} \Tr \Big(  \Big((\id_{\reg{A}} \ot \sigma_{\reg{B}})^{-1/2} \psi_{\reg{AB}} \Big)^2\Big)
    \]
    where the infimum is over all positive definite density operators $\sigma$ acting on register $\reg{B}$.
    The \emph{$\eps$-smoothed conditional R\'{e}nyi $2$-entropy} is 
        \[
            H^{\eps}_2(\reg{A} | \reg{B})_\psi \deq \sup_{\sigma:\td(\sigma,\psi) \leq \eps} H_2(\reg{A} | \reg{B})_\sigma~.
        \]
    \end{itemize}
\end{definition}

We do not elaborate further on the meaning or motivation for the definitions of these entropy measures (we refer the reader to~\cite{tomamichel2012framework,konig2009operational} for deeper discussions); we will only use the following properties of them:
\begin{proposition}[Relations between the entropy measures]
\label{prop:entropy-relations}
    Let $\eps \geq 0$ and let $\ket{\psi}_{\reg{ABC}}$ be a tripartite pure state. The following relationships hold:
    \begin{itemize}
        \item (\emph{Duality relation}) $H^\eps_{\min}(\reg{A} | \reg{B})_\psi = -H^{\eps}_{\max}(\reg{A} | \reg{C})_\psi$. We note that this duality relation only holds when $\psi$ is a pure state on registers $\reg{ABC}$.

        \item (\emph{Bounds for conditional min/max-entropy}) Both $H^{\eps}_{\min}(\reg{A} | \reg{B})_\psi$ and $H^{\eps}_{\max}(\reg{A} | \reg{B})_\psi$ are bounded below by $-\log \mathrm{rank}(\psi_{\reg{A}})$, and bounded above by $\log \mathrm{rank}(\psi_{\reg{A}})$.

        \item (\emph{Isometric invariance}) For all isometries $V$ mapping register $\reg{A}$ to $\reg{A}'$ we have $H_{\min}(\reg{A} | \reg{B})_\psi = H_{\min}(\reg{A}' | \reg{B})_{V \psi V^\dagger}$.

        \item (\emph{Min- versus $2$-entropy}) $H_{\min}(\reg{A} | \reg{B})_\psi \leq H_2(\reg{A} | \reg{B})_\psi$.

        \item (\emph{Operational interpretation of min-entropy}) When $\psi_{\reg{AB}}$ is diagonal (i.e., it corresponds to a bipartite probability distribution $p(a,b)$), $2^{-H_{\min}(\reg{A} | \reg{B})_\psi} = \sum_b p(b) \, \max_a p(a | b)$, i.e., the maximum probability of guessing the state of $\reg{A}$ given the state of $\reg{B}$.

        \item (\emph{Max-entropy does not decrease after appending a state}) For all density matrices $\sigma \in \states(\reg{D})$, we have $H^\eps_{\max}(\reg{A})_{\psi} \leq H^\eps_{\max}(\reg{AD})_{\psi \ot \sigma}$.

    \end{itemize}
\end{proposition}
\begin{proof}
    A proof of the duality relation can be found in ~\cite[Theorem 5.4]{tomamichel2012framework}. The bounds for the conditional min-entropy can be found in~\cite[Proposition 4.3]{tomamichel2012framework}; the bounds on the conditional max-entropy follow via the duality relation.  The isometric invariance property follows directly from the definition of the (smoothed) conditional min-entropy. The min- versus $2$-entropy bound is proved in~\cite[Lemma 2.3]{dupuis2010decoupling}. The operational interpretation of min-entropy is given in~\cite{konig2009operational}. The fact that the max-entropy does not decrease after appending a state follows from~\cite[Theorem 5.7]{tomamichel2012framework}, which states that the smoothed max-entropy is non-decreasing under trace-preserving quantum operations; consider the quantum operation $\psi_{\reg{A}} \mapsto \psi_{\reg{A}} \ot \sigma_{\reg{D}}$, which is clearly trace-preserving. 
\end{proof}

Having established the definitions and properties of these entropy measures, we now prove the characterization of the fundamental limits on one-shot compression for quantum states.

\begin{proof}[Proof of \Cref{thm:info-theory-compression}]
\textbf{Lower bound.}  We first prove the lower bound $H^{2\delta^{1/4}}_{\max}(\rho) \leq K^\delta(\rho)$. Let $(E,D)$ denote a pair of quantum circuits that compresses $\rho$ to $s = K^\delta(\rho)$ qubits with error $\delta$. Let $\ket{\psi}_{\reg{AR}}$ denote a purification of $\rho$. Then using the Fuchs-van de Graaf inequality we get that
\begin{equation}
    \label{eq:compression-0}
    \fidelity \Big ( (D \circ E)(\psi), \psi \Big) \geq 1 - 2\delta~.
\end{equation}
Let $\hat{E}: \reg{A} \to \reg{CE},\hat{D}: \reg{C} \to \reg{AF}$ denote the unitary purifications of the channels corresponding to $E$ and $D$, respectively. 
Then by Uhlmann's theorem, since $(\hat{D} \hat{E} \ot \id_{\reg{R}}) \ket{\psi}_{\reg{RA}}$ is a purification of $(D \circ E)(\psi)$ and $\ket{\psi}_{\reg{AR}}$ is pure,~\cref{eq:compression-0} implies that there exists a pure state $\ket{\theta}_{\reg{EF}}$ such that
\begin{align*}
    1 - 2 \delta \leq \fidelity \Big ( (D \circ E)(\psi), \psi \Big) &= 
    \fidelity\Big ( (\hat{D} \circ \hat{E})(\psi), \psi_{\reg{AR}} \otimes \theta_{\reg{EF}} \Big) 
    \leq \fidelity\Big ( \Tr_{\reg{BC}} \Big ( \hat{D} \circ \hat{E} (\psi) \Big), \rho_{\reg{A}} \otimes \theta_{\reg{F}} \Big)\,.
\end{align*}
The last inequality follows from monotonicity of the fidelity under partial trace. By Fuchs-van de Graaf we have
\begin{equation}
\label{eq:compression-0aa}
    \td \Big( \Tr_{\reg{BC}}( \hat{D} \circ \hat E(\psi)) , \rho_{\reg{A}} \otimes \theta_{\reg{F}}  \Big) \leq \sqrt{2\delta}~.
\end{equation}
Next consider the following entropy bounds using the properties given by \Cref{prop:entropy-relations}:
\begin{align*}
s = \dim(\reg{C}) &\geq - H_{\min}(\reg{C}|\reg{RE})_{\hat E \ket{\psi}} \\
&= - H_{\min}(\reg{AF}|\reg{RE})_{\hat D \hat E \ket{\psi}} \\
&= H_{\max}(\reg{AF})_{\hat D \hat E \ket{\psi}} \\
&\geq H_{\max}^{2 \delta^{1/4}}(\reg{AF})_{\rho_{\reg{B}} \ot \theta_{\reg{R}}} \\
&\geq H_{\max}^{2 \delta^{1/4}}(\reg{A})_{\rho}.
\end{align*}
The first item follows from the bounds on min-entropy. The second line follows from the isometric invariance of the min-entropy. The third line follows from the duality relation between min- and max-entropy. The fourth line follows from the definition of the smoothed max-entropy~\eqref{eq:compression-0aa} and the relationship between the purified distance and trace distance~\cite[Lemma 3.17]{tomamichel2012framework}. The last line follows from the fact that the smoothed max-entropy does not decrease when appending a state. Putting everything together we have $H^{2\delta^{1/4}}_{\max}(\rho) \leq s = K^\delta(\rho)$ as desired.

\paragraph{Upper bound.} We now prove the upper bound, i.e., show that there exists a pair of circuits $(E,D)$ that compresses $\rho$ to $s \deq H^\eps_{\max}(\rho) + 4\log \frac{8}{\delta}$ qubits with error $\delta$, where $\eps = \delta^2/512$. Let $\rho_{\reg{AR}}$ be an arbitrary purification of $\rho$ (with purifying register $\reg{R}$).

We leverage the following \emph{decoupling theorem}, which has been a ubiquitous tool in quantum information theory. Informally, a decoupling theorem states that applying a Haar-random unitary to the $\reg{A}$ system of a bipartite state $\rho_{\reg{AR}}$ and then tracing out an appropriately large subsystem of $\reg{A}$ will result in the remainder of $\reg{A}$ being \emph{decoupled} (i.e., in tensor product) from the reference register $\reg{R}$. There have been many decoupling theorems proved over the years (see, e.g.,~\cite{hayden2008decoupling,dupuis2010decoupling,dupuis2014one,berta2016smooth}); we use the following one due to Dupuis (together with the standard fact that Clifford unitaries form a 2-design). 

\begin{theorem}[{Decoupling Theorem \cite[Theorem 3.8]{dupuis2010decoupling}}]
\label{thm:decoupling}
Let $\rho_{\reg{AB}}$ be a density matrix, $\cal{T}: \states(\reg{A}) \to \states(\reg{E})$ be a completely positive superoperator, $\omega_{ \reg{E} \reg{A}'} = (\cal{T} \ot \id_{\reg{A}'})(\Phi_{\reg{A} \reg{A}'})$ (where $\Phi$ denotes the maximally entangled state), and $\eps \geq 0$. Then
\[
    \int \, \| (\cal{T} \circ U)(\rho_{\reg{AB}}) - \omega_{\reg{E}} \ot \rho_{\reg{B}} \|_1 \, \mathrm{d}U \leq 2^{-\frac{1}{2}H^\eps_2(\reg{A}' | \reg{E})_\omega - \frac{1}{2} H^\eps_2(\reg{A} | \reg{B})_\rho} + 8\eps
\]
where the integral is over the uniform measure on Clifford unitary matrices acting on $\reg{B}$, and $\cal{T} \circ U$ denotes the superoperator where the input state is conjugated by $U$ first, and then $\cal{T}$ is applied.
\end{theorem}

Define the following channel $\cal{T}$ that acts on $\reg{A}$: it measures the first $n - s$ qubits of $\reg{A}$ in the standard basis to obtain a classical outcome $y \in \{0,1\}^{n - s}$, traces out $\reg{A}$, and outputs $y$ in register $\reg{E}$. We now evaluate the state $\omega_{ \reg{E} \reg{A}'} = (\cal{T} \ot \id_{\reg{A}'})(\Phi_{\reg{A} \reg{A}'})$. This can be seen to be
\[
    \omega_{ \reg{E} \reg{A}'} = \sum_{y \in \{0,1\}^{n-s}} \ketbra{yy}{yy}_{\reg{E} \reg{A}_1' } \ot 2^{-s} \, \id_{\reg{A}_2'}
\]
where $\reg{A}'$ is subdivided into two registers $\reg{A}_1' \reg{A}_2'$ with $\reg{A}_1'$ isomorphic to $\reg{E}$. 
The entropy $H^\eps_2(\reg{A}' | \reg{E})_\omega$ can be calculated as follows:
\[
    H_2^\eps(\reg{A}' | \reg{E})_\omega \geq H_2(\reg{A}' | \reg{E})_\omega \geq H_{\min}(\reg{A}' | \reg{E})_\omega~.
\]
The first inequality follows from the definition of the smoothed $2$-entropy. The second inequality follows from \Cref{prop:entropy-relations}. Note that $\omega_{\reg{A}' \reg{E}}$ is a classical state (i.e., it is diagonal in the standard basis); using the operational definition of the min-entropy in this case we see that $H_{\min}(\reg{A}' | \reg{E}) = s$. 

Now we bound the entropy $H^\eps_2(\reg{A} | \reg{R})_\rho$. Since $\rho_{\reg{AR}}$ is pure, \Cref{prop:entropy-relations} gives us 
\[
    -H^\eps_2(\reg{A} | \reg{R})_\rho \leq -H^\eps_{\min}(\reg{A} | \reg{R})_\rho = H^\eps_{\max}(\reg{A})_\rho~.
\]

By \Cref{thm:decoupling}, by averaging there exists a Clifford unitary $U$ such that
\[
\| (\cal{T} \circ U)(\rho_{\reg{AR}}) - \omega_{\reg{E} } \ot \rho_{\reg{R}} \|_1  \leq 2^{-\frac{1}{2}(s - H^\eps_{\max}(\reg{A})_\rho)} + 8\eps  \deq \nu~.
\]
Consider the following two purifications:
\begin{enumerate}
\item $\ket{\Phi}_{\reg{E} \reg{E}'} \otimes \ket{\rho}_{\reg{AR}}$ where $\ket{\Phi}_{\reg{E} \reg{E}'}$ denotes the maximally entangled state on two isomorphic registers $\reg{E}, \reg{E}'$. This is a purification of the density matrix $\omega_{\reg{E}} \otimes \rho_{\reg{R}}$.

\item $\ket{\theta}_{\reg{E} \reg{E}' \reg{CR} \reg{F}} \deq \sum_y \ket{y}_{\reg{E}} \otimes (\Pi_y U \ot \id_{\reg{R}})\ket{\rho}_{\reg{AR}} \otimes \ket{0}_{\reg{F}}$ where $\Pi_y$ is the projection that maps $\reg{A}$ into $ \reg{E}' \reg{C}$ with $\reg{C}$ being an $s$ qubit register and $\reg{E}'$ being $n-s$ qubit register, projecting the first $n-s$ qubits of $\reg{A}$ into the $\ket{y}$ state. The register $\reg{F}$ is isomorphic to $\reg{E}$ and is used to ensure that the dimensions of both purifications are the same. This is a purification of $(\cal{T} \circ U)(\rho_{\reg{AR}})$.
\end{enumerate}

By Fuchs-van de Graaf and Uhlmann's theorem there exist a partial isometry $V$ mapping registers $\reg{E}' \reg{A}$ to $\reg{C} \reg{E}'  \reg{F}$ such that 
\[
    \td \Big ( V (\Phi_{\reg{E} \reg{E}' } \otimes \rho_{\reg{AR}}) V^\dagger \, , \, \theta_{\reg{EE}' \reg{CRF}} \Big) \leq \sqrt{2\nu}~.
\]
Let $\Xi$ be an arbitrary channel completion of $V$. We show that $\Xi$ can be used in place of $V$ with small error. Let $P$ denote the projection onto the support of $V$. Then we have
\begin{align*}
    \Big | \Tr(P (\Phi_{\reg{E} \reg{E}' } \otimes \rho_{\reg{AR}}) ) - 1 \Big | \leq \td \Big ( P (\Phi_{\reg{E} \reg{E}' } \otimes \rho_{\reg{AR}}) P, \theta_{\reg{EE}' \reg{CRF}} \Big) \leq \td \Big ( V (\Phi_{\reg{E} \reg{E}' } \otimes \rho_{\reg{AR}}) V^\dagger, \theta_{\reg{EE}' \reg{CRF}} \Big) \leq \sqrt{2 \nu} \,.
\end{align*}
Let $\tau$ denote the post-measurement state of $\Phi_{\reg{E} \reg{E}' } \otimes \rho_{\reg{AR}}$ after measuring the projector $P$; by the Gentle Measurement Lemma~\cite{winter1999coding} we have $\td(\tau,\Phi_{\reg{E} \reg{E}' } \otimes \rho_{\reg{AR}}) \leq 4\nu^{1/4}$. Thus
\begin{align}
    \td \Big ( \Xi (\Phi_{\reg{E} \reg{E}' } \otimes \rho_{\reg{AR}}) \, , \, \theta_{\reg{EE}' \reg{CRF}} \Big) &\leq \td \Big ( \Xi (\Phi_{\reg{E} \reg{E}' } \otimes \rho_{\reg{AR}}) \, , \, \Xi(\tau) \Big) + \td \Big ( \Xi(\tau), V \tau V^\dagger \Big) \notag \\
    & \qquad + \td \Big(V \tau V^\dagger, V (\Phi_{\reg{E} \reg{E}' } \otimes \rho_{\reg{AR}}) V^\dagger \Big) +    \td \Big ( V (\Phi_{\reg{E} \reg{E}' } \otimes \rho_{\reg{AR}}) V^\dagger \, , \, \theta_{\reg{EE}' \reg{CRF}} \Big) \notag \\
    &\leq 4\nu^{1/4} + 4\nu^{1/4} + \sqrt{2\nu} \leq 10\nu^{1/4}\,,
    \label{eq:compression-0a}
\end{align}
where we used that $\Xi(\tau) = V\tau V^\dagger$ by definition of channel completion. 

Similarly, let $\Lambda$ be an arbitrary channel completion of the partial isometry $V^\dagger$. A similar argument shows that
\[
    \td \Big ( \Phi_{\reg{E} \reg{E}' } \otimes \rho_{\reg{AR}} \, , \, \Lambda ( \theta_{\reg{EE}' \reg{CRF}}) \Big) \leq 10\nu^{1/4}~.
\]

We now continue with $\Xi$ instead of $V$ and $\Lambda$ instead of $V^\dagger$. Applying the channel that measures the register $\reg{E}$ in the standard basis to both arguments of the left-hand side of~\cref{eq:compression-0a} and using that the trace distance is non-increasing under quantum operations we have
\[
    \E_y \td \Big ( \Xi ( \ketbra{y}{y}_{\reg{E}'} \ot \ketbra{\rho}{\rho}_{\reg{AR}}) \, , \, 2^{n-s} \alpha_y  \ketbra{y}{y}_{\reg{E}'} 
 \ot \ketbra{\rho_{U,y}}{\rho_{U,y}}_{\reg{CR}} \ot \ketbra{0}{0}_{\reg{F}} \Big) \leq 10\nu^{1/4}~,
\]
where the expectation is over a uniformly random $y$, and $\alpha_y \deq \| \Pi_y U \ket{\rho}_{\reg{AR}} \|^2$ and the pure state $\ket{\rho_{U,y}}_{\reg{RC}}$ is defined so that
\[
     \alpha_y^{-1/2} \, \Pi_y U \ket{\rho}_{\reg{AR}} = \ket{y}_{\reg{E}'} \ot \ket{\rho_{U,y}}_{\reg{CR}}~.
\]
By monotonicity of the trace distance this implies that $\E_y |2^{n-s} \alpha_{y} -1 |\leq 10\nu^{1/4}$. Therefore by triangle inequality we have
\begin{equation}
    \label{eq:compression-1}
    \E_y \td \Big ( \Xi ( \ketbra{y}{y}_{\reg{E}'} \ot \ketbra{\rho}{\rho}_{\reg{AR}}) \, , \, \ketbra{y}{y}_{\reg{E}'} 
 \ot \ketbra{\rho_{U,y}}{\rho_{U,y}}_{\reg{CR}} \ot \ketbra{0}{0}_{\reg{F}} \Big) \leq 20\nu^{1/4}~,
\end{equation}

Define the following quantum circuits:
\begin{enumerate}
    \item The circuit $E$ acts on register $\reg{A}$ and behaves as follows: it appends a randomly chosen $\ket{y}$ in register $\reg{E}'$, applies the channel $\Xi$, and then traces out registers $\reg{E}' \reg{F}$. In other words, it implements the following channel:
    \[
        E(\sigma_{\reg{A}}) = \E_y \Tr_{\reg{E}' \reg{F}} \Big( \Xi (\ketbra{y}{y}_{\reg{E}'} \ot \sigma_{\reg{A}} ) \Big)~.
    \]
    \item The circuit $D$ takes as input register $\reg{C}$ and behaves as follows: it appends a randomly chosen $\ket{y}$ in register $\reg{E}'$ and $\ket{0}$ in register $\reg{F}$, applies the channel $\Lambda$, and then traces out register $\reg{E}'$. In other words, it implements the following channel:
    \[
        D(\tau_{\reg{C}}) = \E_y \Tr_{\reg{E}'} \Big( \Lambda ( \ketbra{y}{y}_{\reg{E}'} \ot \tau_{\reg{C}} \ot \ketbra{0}{0}_{\reg{F}})  \Big)~.
    \]
\end{enumerate}
Then~\cref{eq:compression-1} implies that
\begin{gather*}
\E_y \td \Big ( E( \ketbra{\rho}{\rho}_{\reg{AR}}), \, \ketbra{\rho_{U,y}}{\rho_{U,y}}_{\reg{CR}} \Big) \leq 20\nu^{1/4} \\
\E_y \td \Big ( \ketbra{\rho}{\rho}_{\reg{AR}} \, , \, D(\ketbra{\rho_{U,y}}{\rho_{U,y}}_{\reg{CR}}) \Big) \leq 20\nu^{1/4}~.
\end{gather*}
Put together this means
\[
    \E_y \td \Big( (D \circ E) ( \ketbra{\rho}{\rho}_{\reg{AR}}), \ketbra{\rho}{\rho}_{\reg{AR}} \Big)\leq 40\nu^{1/4}~.
\]
Although we have defined the circuits $E,D$ in terms of the purification $\ket{\rho}_{\reg{AR}}$, observe that Uhlmann's theorem implies that the same circuits works for \emph{all} purifications of $\rho_{\reg{A}}$. Thus, since the output of channel $E$ is register $\reg{C}$ which has size $s$ qubits, this shows that $(E,D)$ compresses $\rho$ to $s$ qubits with error $40\nu^{1/4}$. By our choice of $s = H^{\eps}_{\max}(\reg{B})_\rho + 8\log\frac{4}{\delta}$ and $\eps = (\delta/40)^4$, this error is at most $\delta$.
\end{proof}

\newpage
\printbibliography

\end{document}